\documentclass[12pt,epsf]{report}
\usepackage{ics} 
\usepackage[dvips]{graphicx} 
\usepackage{tabularx} 
\usepackage{latexsym} 
\usepackage{verbatim}
\usepackage{moreverb}
\usepackage{makeidx} 
\usepackage{amssymb}
\usepackage{amsmath}
\usepackage{bussproofs}
\usepackage{amsthm}
\usepackage{txfonts}
\usepackage{setspace}
\usepackage{enumitem}
\newtheorem{definition}{Definition}[chapter]
\newtheorem{theorem}[definition]{Theorem}
\newtheorem{lemma}[definition]{Lemma}
\newtheorem{prop}[definition]{Proposition}
\newtheorem{corollary}[definition]{Corollary}
\newtheorem{example}[definition]{Example}
\newtheorem{problem}[definition]{Problem}

\newcommand{\fracinline}[2]{
\raisebox{0.4ex}{\small $#1$}
\raisebox{0ex}{\large $/$}
\raisebox{-0.2ex}{\small $#2$}
}

\def\epsfsize#1#2{\ifnum#1>\hsize\hsize\else#1\fi}

\begin{document}


\title{Trace Equivalence and Epistemic Logic to Express Security Properties} 

\major{Division of Mathematics and Mathematical Sciences}
\course{Mathematical Sciences Course}

\author{Kiraku Minami}


\date{\today}


\labaffiliation{Graduate School of Science, Kyoto University}


\maketitle

\pagenumbering{roman}
\begin{abstract}
\addcontentsline{toc}{chapter}{Abstract}
In process algebras, security properties are expressed as equivalences between processes, but which equivalence is suitable is not clear. This means that there is a gap between an intuitive security notion and the formulation. Appropriate formalization is essential for verification, and our purpose is bridging this gap. By chasing scope extrusions, we prove that trace equivalence is congruent. Moreover, we construct an epistemic logic for the applied pi calculus and show that its logical equivalence agrees with the trace equivalence. We use the epistemic logic to show that trace equivalence is pertinent 
in the presence of a non-adaptive attacker.
\end{abstract}

\chapter*{Acknowledgements}
\addcontentsline{toc}{chapter}{Acknowledgements}
\setcounter{page}{2}
\noindent
I am sincerely appreciative of Professor Masahito Hasegawa who provided helpful comments and suggestions. I would also like to convey a sense of gratitude to my family for their moral support and warm encouragements. 

\tableofcontents

\chapter{Introduction}
\pagenumbering{arabic}
\setcounter{page}{1}

\section{Background}
In modern society, information and communications technology is indispensable to our daily lives, and a number of communication protocols are proposed to securely transmit information using a network. Verification of safety of these protocols is also essential, but it is not easy.

In the first place, how to define security notions is not clear. As a matter of fact, various different definitions of the same security property have been proposed. In addition, how to model communication is also not clear; many such models have been developed. In this work, we focus on process algebras because it allows us to handle composition easily. In process algebras, typical security properties such as secrecy are represented by an equivalence between processes.

For instance, secrecy of a message $M$ is expressed as equivalence between a process sending $M$ and a process sending $M'$.

Many process equivalences exist in the literature, but which is the most suitable for capturing security properties is not always clear. For instance, Delaune, Kremer and Ryan  \cite{delaune2009verifying} proposed the definition of privacy in electronic voting in terms of the applied pi calculus \cite{DBLP:journals/corr/AbadiBF16} as follows. 
\begin{definition}[{\cite[Definition 9]{delaune2009verifying}}]
A voting protocol respects vote-privacy (or just privacy) if
\begin{equation*}
S[V_A \{\fracinline{a}{v}\} | V_B \{\fracinline{b}{v}\}] \approx _l S[V_A \{\fracinline{b}{v}\} | V_B \{\fracinline{a}{v}\}]
\end{equation*}
for all possible votes $a$ and $b$.
\end{definition}
\noindent
Intuitively, this definition states that an attacker cannot distinguish two situations where votings are swapped.
Note that indistinguishability is expressed by labelled bisimilarity $\approx_l$ in this definition. Is it the most suitable? This is a nontrivial question. Chadha, Delaune and Kremer \cite{chadha2009epistemic} claim that trace equivalence is more suitable regarding privacy.

In fact, an outputted message is not explicitly expressed in the applied pi calculus. It is represented via an alias variable. This feature enables us to handle cryptographic protocols and suggests that trace equivalence means an attacker's indistinguishability. This is because an attacker can observe only labelled transition. We recall the syntax and semantics of the calculus in Chapter \ref{applied-pi}.

Bisimilarity and traced equivalences are both well-studied equivalence relations on labelled transition systems.
However, trace equivalence in the applied pi calculus (and
other variants of the pi calculus \cite{Milner:1992:CMP:162037.162038,Milner:1992:CMP:162037.162039,Milner:1999:CMS:329902}) has not drawn much attention.
This is probably because trace equivalence is the coarsest among major equivalences. However, security properties sometimes require that actually different processes are regarded as the same. Thereby, a fine equivalence such as bisimilarity does not have to be always adequate.

Epistemic logic is often used to capture security notions (e.g. \cite{chadha2009epistemic,mano2010role,tsukada2016compositional}). It enables us to directly express security properties. For example, when a message $M$ sent by an agent $a$ is anonymous, we can say that an adversary cannot know who sent $M$. In an epistemic logic, we can express it with a formula such like $\lnot K\mathrm{Send}(a, M)$. This logical formulation is close to our intuition. Nevertheless, research about an epistemic logic for the applied pi calculus is few.
\section{Contributions}
We prove that trace equivalence for the applied pi calculus is a congruence in Chapter \ref{cong}. Secondly, we give an epistemic logic that characterizes trace equivalence in Chapter \ref{epi}. In addition, we define security properties such as secrecy and role-interchangeability \cite{mano2010role,tsukada2016compositional} using trace equivalence and our epistemic logic. 
Moreover, we show that secrecy and openness generally are not preserved by parallel composition.
On the other hand, total secrecy is characterized by trace equivalence, so it is preserved by an application of a context. 
The same partly holds for role-interchangeability.
Finally, we prove that the satisfaction problem for the epistemic logic is undecidable.

Our results suggest that trace equivalence is more suitable to express security notions than labelled bisimilarity. 

\chapter{The Applied Pi Calculus}\label{applied-pi}
\section{Syntax}
The applied pi calculus \cite{DBLP:journals/corr/AbadiBF16} is an extension of the pi calculus \cite{Milner:1992:CMP:162037.162038,Milner:1992:CMP:162037.162039,Milner:1999:CMS:329902} and can handle cryptographic protocols.

Let $\Sigma$ be a signature equipped with an equational theory. $\Sigma$ consists of a finite set of function symbols with their arity. 
Names, variables, and terms are defined on $\Sigma$ by the following grammar:

\begin{align*} 
M, N::=&n, a, b, c, d &\mathrm{name}\\
    |&x, y, z, w &\mathrm{variable}\\
    |&f(M_1,...,M_n) &\mathrm{function\ application}
\end{align*}

The above $f$ is a function symbol with arity $n$ on $\Sigma$. 
We say that a term is ground when it contains no variables. Meta-variables $u, v$ are used to represent over both names and variables. We describe a sequence of terms $M_1, ..., M_n$ as $\widetilde{M}$. We denote by n($M$) and v($M$) the set of names and variables in $M$.
The grammar for processes is given below:

\begin{align*} 
P, Q, R, S::=& 0 &\mathrm{nil}\\ 
   |& \overline{M} \langle N \rangle .P &\mathrm{output}\\
   |& M (x).P &\mathrm{ input}\\
   |& \nu n.P &\mathrm{ name\ restriction}\\
   |& \mathrm{if} ~M=N~ \mathrm{then} ~P~ \mathrm{else} ~Q &\mathrm{ conditional\ branch}\\
   |& P+Q &\mathrm{ nondeterministic\ choice}\\
   |& P|Q &\mathrm{ parallel\ composition}\\
   |& !P &\mathrm{ replication}
\end{align*}

0 cannot do anything. $\overline{M} \langle N \rangle .P$ sends $N$ on a channel $M$ and then behaves as $P$. $M(x).P$ waits for input from a channel $M$ and then behaves as $P$ with the variable $x$ replaced by the received message. In both cases, we omit $P$ if it is 0. The process $\nu n.P$ creates a new restricted name $n$ and then behaves as $P$. A restricted name is often regarded as a nonce, a key and so on. The conditional branch $\mathrm{if} ~M=N~ \mathrm{then} ~P~ \mathrm{else} ~Q$ behaves as P when $M=N$ and as $Q$ otherwise. Note that $M = N$ means equality on $\Sigma$. We omit $Q$ if $Q$ is 0. $P+Q$ nondeterministically behaves $P$ or $Q$. $P|Q$ is the parallel composition of $P$ and $Q$; the replication $!P$ behaves as unbounded copies of $P$ running concurrently.

We call processes defined above plain processes. Moreover, we define extended processes:

\begin{align*} 
A, B, C, D, E::=& P &\mathrm{ plain\ process}\\
   |& \nu n.A &\mathrm{ name\ restriction}\\
   |& \nu x.A &\mathrm{ variable\ restriction}\\
   |& A|B &\mathrm{ paralell\ composition}\\
   |& \{ M/x \} &\mathrm{ active\ substitution}
\end{align*}

We consider that $\{M/x\}$ is floating around processes. If it touches other processes, a substituion happens. We will describe in detail later. If we want to impose reins on this effect, we restrict variables. $\{M/x\}$ often represents messages outputted to an environment. Moreover, it is sometimes interpreted as attacker's knowledge in view of security because an active substitution $\{M/x\}$ appears when a process sends a message $M$ to an environment. When $\tilde{M}=M_1,\dots,M_n$ and $\tilde{x}=x_1,\dots,x_n, \{\tilde{M}/\tilde{x}\}$ expresses $\{M_1/x_1\}|\dots|\{M_n/x_n\}$. In particular, if $n=0$, it is a null process 0. 

Let $\widetilde{u} = u_1,...,u_n$. Then $\nu \widetilde{u}. A = \nu u_1. ... \nu u_n. A$.

Let $\sigma = \{\widetilde{M}/\widetilde{x}\}$, then $M\sigma$ is a term obtained by replacing all $x_i$ with $M_i$.  We denote by fn($A$) and fv($A$) the sets of free names and variables respectively. These are inductively defined below.

\begin{align*}  
\mathrm{ fn}(0) =& \emptyset \\
\mathrm{ fn}(\overline{M} \langle N \rangle .P) =& \mathrm{ n}(M) \cup \mathrm{ n}(N) \cup \mathrm{ fn}(P) \\
\mathrm{ fn}(M (x).P) =& \mathrm{ n}(M) \cup \mathrm{ fn}(P) \\
\mathrm{ fn}(\nu n.P) =& \mathrm{ fn}(P)\setminus \{n\} \\
\mathrm{ fn}(\mathrm{if} ~M=N~ \mathrm{then} ~P~ \mathrm{else} ~Q) =& \mathrm{ n}(M) \cup \mathrm{ n}(N) \cup \mathrm{ fn}(P) \cup \mathrm{ fn}(Q) \\
\mathrm{ fn}(P+Q) =& \mathrm{ fn}(P) \cup \mathrm{ fn}(Q) \\
\mathrm{ fn}(P|Q) =& \mathrm{ fn}(P) \cup \mathrm{ fn}(Q) \\
\mathrm{ fn}(!P) =& \mathrm{ fn}(P) \\
\mathrm{ fn}(\nu n.A) =& \mathrm{ fn}(A) \setminus \{n\} \\
\mathrm{ fn}(\nu x.A) =& \mathrm{ fn}(A) \\
\mathrm{ fn}(A|B) =& \mathrm{ fn}(A) \cup \mathrm{ fn}(B) \\
\mathrm{ fn}(\{M/x\}) =& \mathrm{ n}(M)
\end{align*}

\begin{align*} 
\mathrm{ fv}(0) =& \emptyset \\
\mathrm{ fv}(\overline{M} \langle N \rangle .P) =& \mathrm{ v}(M) \cup \mathrm{ v}(N) \cup \mathrm{ fv}(P) \\
\mathrm{ fv}(M (x).P) =& \mathrm{ v}(M) \cup \mathrm{ fv}(P) \setminus \{x\} \\
\mathrm{ fv}(\nu n.P) =& \mathrm{ fv}(P) \\
\mathrm{ fv}(\mathrm{if} ~M=N~ \mathrm{then} ~P~ \mathrm{else} ~Q) =& \mathrm{ v}(M) \cup \mathrm{ v}(N) \cup \mathrm{ fv}(P) \cup \mathrm{ fv}(Q) \\
\mathrm{ fv}(P+Q) =& \mathrm{ fv}(P) \cup \mathrm{ fv}(Q) \\
\mathrm{ fv}(P|Q) =& \mathrm{ fv}(P) \cup \mathrm{ fv}(Q) \\
\mathrm{ fv}(!P) =& \mathrm{ fv}(P) \\
\mathrm{ fv}(\nu n.A) =& \mathrm{ fv}(A) \\
\mathrm{ fv}(\nu x.A) =& \mathrm{ fv}(A) \setminus \{x\} \\
\mathrm{ fv}(A|B) =& \mathrm{ fv}(A) \cup \mathrm{ fv}(B) \\
\mathrm{ fv}(\{M/x\}) =& \mathrm{ v}(M) \cup \{x\}
\end{align*}

Furthermore, We denote by bn($A$) and rv($A$)  the sets of bound names and restricted variables respectively. Note that $\nu x.$ appears only in an extended process.

\begin{align*} 
\mathrm{ bn}(0) =& \emptyset \\
\mathrm{ bn}(\overline{M} \langle N \rangle .P) =& \mathrm{ bn}(P) \\
\mathrm{ bn}(M (x).P) =& \mathrm{ bn}(P) \\
\mathrm{ bn}(\nu n.P) =& \mathrm{ bn}(P) \cup \{n\} \\
\mathrm{ bn}(\mathrm{if} ~M=N~ \mathrm{then} ~P~ \mathrm{else} ~Q) =& \mathrm{ bn}(P) \cup \mathrm{ bn}(Q) \\
\mathrm{ bn}(P+Q) =& \mathrm{ bn}(P) \cup \mathrm{ bn}(Q) \\
\mathrm{ bn}(P|Q) =& \mathrm{ bn}(P) \cup \mathrm{ bn}(Q) \\
\mathrm{ bn}(!P) =& \mathrm{ bn}(P) \\
\mathrm{ bn}(\nu n.A) =& \mathrm{ bn}(A) \cup \{n\} \\
\mathrm{ bn}(\nu x.A) =& \mathrm{ bn}(A) \\
\mathrm{ bn}(A|B) =& \mathrm{ bn}(A) \cup \mathrm{ bn}(B) \\
\mathrm{ bn}(\{M/x\}) =& \emptyset
\end{align*}

\begin{align*} 
\mathrm{ rv}(P) =& \emptyset \\
\mathrm{ rv}(\nu n.A) =& \mathrm{ rv}(A) \\
\mathrm{ rv}(\nu x.A) =& \mathrm{ rv}(A) \cup \{x\} \\
\mathrm{ rv}(A|B) =& \mathrm{ rv}(A) \cup \mathrm{ rv}(B) \\
\mathrm{ rv}(\{M/x\}) =& \emptyset
\end{align*}

We also define $\mathrm{n}(A) = \mathrm{fn}(A) \cup \mathrm{bn}(A)$ and $\mathrm{v}(A) = \mathrm{fv}(A) \cup \mathrm{rv}(A)$. 
The domain dom($A$) of an extended process $A$ is also inductively defined below. If variables in other concurrently running processes are in $\mathrm{dom}(A)$, $A$ affects them.

\begin{align*} 
\mathrm{ dom}(P)            =& \emptyset \\
\mathrm{ dom}(\nu n.A)    =& \mathrm{ dom}(A) \\
\mathrm{ dom}(\nu x.A)    =& \mathrm{ dom}(A) \setminus \{ x \} \\
\mathrm{ dom}(A|B)         =& \mathrm{ dom}(A) \cup \mathrm{ dom}(B) \\
\mathrm{ dom}(\{ M/x \} ) =& \{ x \}
\end{align*}

We basically identify an extended process with an $\alpha$-equivalent process, but we carefully handle $\alpha$-equivalent processes in Chapter \ref{cong}, \ref{epi}.
Hereafter, we assume three restrictions. First, suppose that substitutions are cycle-free with respect to variables.
\begin{definition}
An active substitution $\sigma$ is cycle-free $\overset{\mathrm{def}}{\Leftrightarrow}$ $\sigma$ is written as $\{M_1/x_1,...,M_n/x_n\}$ by reordering so that each $M_i$ does not contain $x_1,...,x_i$ by reordering.
\end{definition}
Secondly, we can compose $A$ and $B$ in parallel only if $\mathrm{ dom}(A) \cap \mathrm{ dom}(B) = \emptyset$. Lastly, we can restrict only variables in a domain. We say that $A$ is closed if $\mathrm{ fv}(A) = \mathrm{ dom}(A)$.

Note the proposition below.

\begin{prop}
Consider an active substitution $\sigma = \{M_1/x_1,...,M_n/x_n\}$ where each $M_i$ does not contain $x_1,...,x_i$. Then, there exists an active substitution $\rho \equiv \sigma$ such that each $x_i\rho$ contains no variables other than ones in $\mathrm{fv}(\sigma) \setminus \mathrm{dom}(\sigma)$.
\end{prop}
\begin{proof}
We prove by induction on $n$. The base case is trivial. We consider the induction case.

Let $\sigma' = \{M_n/x_n\}$ and $\sigma'' = \{M_1\sigma'/x_1,...,M_{n-1}\sigma'/x_{n-1}\}$. Each $M_i\sigma'$ does not contain $x_n$, so there exists $\rho' \equiv \sigma''$ such that each $x_i\rho'$ contains no variables other than ones in $\mathrm{fv}(\sigma'') \setminus \mathrm{dom}(\sigma'')$ by induction hypothesis. Then,
\begin{equation*}
\sigma = \{M_1/x_1,...,M_{n-1}/x_{n-1}\} | \sigma' \equiv \{M_1\sigma'/x_1,...,M_{n-1}\sigma'/x_{n-1}\} | \sigma' \equiv \rho' | \sigma'.
\end{equation*}
$\rho' | \sigma'$ meets the condition against $n$. Note that $\mathrm{fv}(\sigma'') \setminus \mathrm{dom}(\sigma'') \subseteq \mathrm{fv}(\sigma) \setminus \mathrm{dom}(\sigma)$.
\end{proof}

\section{Semantics}
\subsection{Internal Reduction}
Internal communication, conditional branch and nondeterministic choice are unobservable actions.

A context is an expression containing one hole. The hole is represented by [\_]. An evaluation context is a context whose hole is neither under a replication, a conditional, nor an action prefix. 
Structural equivalence $\equiv$ is the smallest equivalence relation between extended processes which is closed under an application of an evaluation context and $\alpha$-conversion, such that:

\begin{align*} 
A|0 & \equiv A \\
(A|B)|C & \equiv A|(B|C) \\
A|B & \equiv B|A \\
(\nu u.A)|B & \equiv \nu u.(A|B) ~~ \mathrm{if} ~u \notin \mathrm{ fn}(B) \cup \mathrm{ fv}(B) \\
\nu u. \nu v.A & \equiv \nu v. \nu u.A \\
!P & \equiv P|!P \\
P+Q & \equiv Q+P \\
\nu x.\{ M/x \} & \equiv 0 \\
A| \{ M/x \} & \equiv A \{ M/x \} | \{ M/x \} \\
\{ M/x \} & \equiv \{ N/x \} ~~ \mathrm{if}\ \Sigma \vdash M = N
\end{align*}

That is, $|$ beheves like a commutative monoid, and + behaves like a  commutative magma. Restriction is also commutative.
Note that $\{M/x\}$ in $A\{M/x\}$ is not an active substitution, it is an ordinal substitution. $A\{M/x\}$ is the same as $A$ expect that all free occurences of $x$ is replaced with $M$, and it is a capture-avoiding substitution. The last rule depends on $\Sigma$.
We omit redundant parentheses when distinction between structural equivalent processes is not necessary.

\begin{definition}\label{int_red}
Internal reduction $\rightarrow$ is the smallest relation between extended processes closed under structural equivalence and application of evaluation contexts such that:

\begin{align*} 
\mathrm{if} ~M=N~ \mathrm{then} ~P~ \mathrm{else} ~Q & \rightarrow P ~~ \ \mathrm{when} ~ \Sigma \vdash M = N \\
\mathrm{if} ~M=N~ \mathrm{then} ~P~ \mathrm{else} ~Q & \rightarrow Q ~~ \ \mathrm{when} ~ \Sigma \not\vdash M = N\ (M\ \mathrm{ and}\ N\ \mathrm{ are\ ground.}) \\
P+Q & \rightarrow P \\
\overline{M} \langle N \rangle .P | M (x).Q & \rightarrow P | Q \{ N/x \}
\end{align*}
\end{definition}

The second rule requires that $M$ and $N$ are ground, and it is essential. If this condition is not satisfied, unreasonable reduction happens when an active substituion concurrently exists. For example, consider $\mathrm{if} ~x=a~ \mathrm{then} ~P~ \mathrm{else} ~Q | \{a/x\}$.

An internal reduction is often called silent action.

\subsection{Labelled Transition}
Labelled semantics enables us to express unsafe communication.

We denote by $\alpha$ a labelled action. The labelled semantics gives a ternary relation $A \overset{\alpha}{\longrightarrow} A'$ on two processes and a label. A label is an input form or an output form. In the former case, $\alpha$ is written as $M(N)$. It means that a message $N$ is received on a channel $M$. In the latter case, $\alpha$ is written as $\nu x. \overline{M} \langle x \rangle$, where $x$ is a variable which is fresh. It means that a message is sent on channel $M$ via an alias variable $x$. It is a notable feature that a sent message does not explicitly appear in a label.

We define n($\alpha$), fv($\alpha$) and rv($\alpha$) as with a process. 
\begin{align*}
\mathrm{ n}(M(N)) =& \mathrm{ n}(M) \cup \mathrm{ n}(N) \\
\mathrm{ n}(\nu x. \overline{M} \langle x \rangle) =& \mathrm{ n}(M)\\
\mathrm{ fv}(M(N)) =& \mathrm{ v}(M) \cup \mathrm{ v}(N)\\
\mathrm{ fv}(\nu x. \overline{M} \langle x \rangle) =& \mathrm{ v}(M) \\
\mathrm{ rv}(M(N)) =& \emptyset \\
\mathrm{ rv}(\nu x. \overline{M} \langle x \rangle) =& \{x\}.
\end{align*}
Additionally, we regard that a silent action contains no variables and names.

We recall structural operational semantics for the applied pi calculus.

\begin{prooftree}
\AxiomC{}
\UnaryInfC{$ M(x).P \overset{M(N)}{\longrightarrow} P \{ N/x \} $}
\end{prooftree}
\begin{prooftree}
\AxiomC{$ x \notin \mathrm{ fv}(\overline{M} \langle N \rangle .P) $}
\UnaryInfC{$ \overline{M} \langle N \rangle .P \overset{\nu x.\overline{M} \langle x \rangle}{\longrightarrow} P| \{ N/x \} $}
\end{prooftree}
\begin{prooftree}
\AxiomC{$A \overset{\alpha}{\longrightarrow} A' \land u \notin \mathrm{ n}(\alpha) \cup \mathrm{ v}(\alpha)$}
\UnaryInfC{$\nu u.A \overset{\alpha}{\longrightarrow} \nu u.A'$}
\end{prooftree}
\begin{prooftree}
\AxiomC{$ A \overset{\alpha}{\longrightarrow} A' \land \mathrm{ rv}(\alpha) \cap \mathrm{ fv}(B) = \emptyset $}
\UnaryInfC{$ A|B \overset{\alpha}{\longrightarrow} A'|B $}
\end{prooftree}
\begin{prooftree}
\AxiomC{$ A \equiv A' \land A' \overset{\alpha}{\longrightarrow} B' \land B' \equiv B $}
\UnaryInfC{$ A \overset{\alpha}{\longrightarrow} B $}
\end{prooftree}

We denote by $\mu$ a silent action and a labelled action. We define $\Longrightarrow$ as the transitive reflective closure of $\longrightarrow$, and $\overset{\alpha}{\Longrightarrow}$ as $\Longrightarrow \overset{\alpha}{\longrightarrow} \Longrightarrow$. $\overset{\mu}{\Longrightarrow}$ is the former when $\mu$ is silent and is the latter when $\mu = \alpha$.

\section{Equivalence}
\subsection{Labelled Bisimulation}

If $\mathrm{ fn}(A) \cap \mathrm{ bn}(A) = \emptyset$ and any bound name is not restricted twice, we say that $A$ is name-distinct \cite{chadha2009epistemic}. If $\mathrm{ fv}(A) \cap \mathrm{ rv}(A) = \emptyset$ and any restricted variable is not restricted twice, we say that $A$ is variable-distinct.
When $A$ and $B$ are name-variable-distinct and $\mathrm{fn}(A) \cap \mathrm{bn}(B) = \mathrm{bn}(A) \cap \mathrm{fn}(B) = \mathrm{bn}(A) \cap \mathrm{bn}(B) = \emptyset$, we say that $A$ and $B$ are  bind-exclusive.

A frame is an extended process made from 0 and active substitutions by restriction and parallel composition. $\varphi$ and $\psi$ are used to represent frames. 
We define a map fr from an extended process to a frame below.

\begin{align*} 
\mathrm{ fr}(P) & = 0 \\
\mathrm{ fr}(\nu n.A) & =\nu n.\mathrm{ fr}(A) \\
\mathrm{ fr}(\nu x.A) & =\nu x.\mathrm{ fr}(A) \\
\mathrm{ fr}(A|B) & =\mathrm{ fr}(A)|\mathrm{ fr}(B) \\
\mathrm{ fr}(\{ M/x \} ) &=\{ M/x \}
\end{align*}

From now, we identify a frame with a structural equivalent frame.

\begin{definition} 
$(M = N)\varphi \overset{\mathrm{def}}{\Leftrightarrow} \mathrm{ v} (M) \cup \mathrm{ v}(N) \subseteq \mathrm{ dom}(\varphi) \land M\sigma = N\sigma$ where $\varphi \equiv \nu \widetilde{n}.\sigma \land \widetilde{n} \cap (\mathrm{ n}(M) \cup \mathrm{ n}(N)) = \emptyset$ for some names $\widetilde{n}$ and active substituions $\sigma$.
\end{definition}

It follows from  \cite[Lemma D.1]{DBLP:journals/corr/AbadiBF16} that this definition is well-defined.

\begin{definition} 
The static equivalence $\approx_s$ on closed frames is given by
\begin{equation*}
\varphi \approx_s \psi \overset{\mathrm{def}}{\Leftrightarrow} \mathrm{dom}(\varphi) = \mathrm{dom}(\psi) \land \forall M, N; (M=N)\varphi \Leftrightarrow (M=N)\psi.
\end{equation*}
for closed frames $\varphi$ and $\psi$.

The static equivalence on closed processes is given by
\begin{equation*}
A \approx_s B \overset{\mathrm{def}}{\Leftrightarrow} \mathrm{fr}(A) \approx_s \mathrm{fr}(B)
\end{equation*}
for closed processes $A$ and $B$.
\end{definition}

We say that $\varphi$ and $\psi$ are statically equivalent when $\varphi \approx_s \psi$. In addition, we say that $A$ and $B$ are static equivalent when $A \approx_s B$.
It is proved in  \cite[Lemma 4.1]{DBLP:journals/corr/AbadiBF16} that a static equivalence is closed by an application of an evaluation context.

We recall labelled bisimulation. It intuitively means that two processes operate in the same manner.

\begin{definition} 
A binary relation $\mathcal{R}$ between closed extended processes is called a labelled simulation if and only if whenever $A\mathcal{R}B$,
\begin{enumerate}
\item $A \approx_s B$
\item $A \overset{\mu}{\longrightarrow} A' \land A' : \mathrm{closed} \land \mathrm{fv}(\mu) \subseteq \mathrm{dom}(A) \Rightarrow \exists B'\ \mathrm{s.t.}\ B \overset{\mu}{\Longrightarrow} B' \land A\mathcal{R}B$
\end{enumerate}
If both $\mathcal{R}$ and $\mathcal{R}^{-1}$ are labelled simulations, we say that $\mathcal{R}$ is a labelled bisimulation.
We call the largest labelled bisimulation a labelled bisimilarity.
\end{definition}

It is proved in  \cite[Theorem 4.1]{DBLP:journals/corr/AbadiBF16} the labelled bisimilarity is closed by an application of an evaluation context.

Note that structural equivalence is stronger than labelled bisimilarity.

\begin{lemma}
Let $\sigma$ and $\rho$ be active substitutions. We assume that their images of variables contain no variables in domains. Then, $\sigma \equiv \rho \Rightarrow [\mathrm{dom}(\sigma) = \mathrm{dom}(\rho) \land \forall x \in \mathrm{dom}(\sigma); x\sigma = x\rho]$.
\end{lemma}
\begin{proof}
We assume that $\sigma \equiv \rho$. It immediately follows that $\mathrm{dom}(\sigma) = \mathrm{dom}(\rho)$. We arbitrarily take $x \in \mathrm{dom}(\sigma)$.
Because structural equivalence is preserved by a parallel composition,
\begin{equation*}
\mathrm{if} ~x=x\sigma~ \mathrm{then} ~\overline{a}\langle s \rangle~ \mathrm{else} ~\overline{b}\langle s \rangle | \sigma \equiv \mathrm{if} ~x=x\sigma~ \mathrm{then} ~\overline{a}\langle s \rangle~ \mathrm{else} ~\overline{b}\langle s \rangle | \rho.
\end{equation*}
Because $x\sigma$ does not contain variables in $\mathrm{dom}(\sigma)$,
\begin{equation*}
\mathrm{if} ~x\sigma=x\sigma~ \mathrm{then} ~\overline{a}\langle s \rangle~ \mathrm{else} ~\overline{b}\langle s \rangle | \sigma \equiv \mathrm{if} ~x\rho=x\sigma~ \mathrm{then} ~\overline{a}\langle s \rangle~ \mathrm{else} ~\overline{b}\langle s \rangle | \rho.
\end{equation*}
The left-hand side can act as below:
\begin{equation*}
\mathrm{if} ~x\sigma=x\sigma~ \mathrm{then} ~\overline{a}\langle s \rangle~ \mathrm{else} ~\overline{b}\langle s \rangle | \sigma \longrightarrow \overline{a}\langle s \rangle | \sigma \overset{\nu y. \overline{a} \langle y \rangle}{\longrightarrow} \{s/y\} | \sigma
\end{equation*}
Hence, the right-hand side must act as below:
\begin{equation*}
\mathrm{if} ~x\rho=x\sigma~ \mathrm{then} ~\overline{a}\langle s \rangle~ \mathrm{else} ~\overline{b}\langle s \rangle | \rho \longrightarrow \overline{a}\langle s \rangle | \rho \overset{\nu y. \overline{a} \langle y \rangle}{\longrightarrow} \{s/y\} | \rho
\end{equation*}
In other words, it must hold that $x\rho = x\sigma$.
\end{proof}

\subsection{Trace Equivalence}
At first, we define a trace. It can be regarded as a history of observable actions.

\begin{definition} 
A trace $\mathbf{tr}$ is a finite derivation $\mathbf{tr} = A_0 \overset{\mu_1}{\Longrightarrow} ... \overset{\mu_n}{\Longrightarrow} A_n$ such that every $A_i$ is closed and $\mathrm{fv}(\mu_i) \subseteq \mathrm{dom}(A_{i-1})$ for all $i$. If $A_n$ can perform no actions, the trace $\mathbf{tr}$ is said to be complete or maximal.
Given a trace $\mathbf{tr}$, let $\mathbf{tr}[i]$ be its $i$-th process $A_i$ and $\mathbf{tr}[i,j]$ be the subderivation $A_i \overset{\mu_{i+1}}{\Longrightarrow} ... \overset{\mu_j}{\Longrightarrow} A_j$ where $0 \leq i \leq j \leq n$. The length of a trace $\mathbf{tr}$ is denoted by $|\mathbf{tr}| = n$.
If each $\overset{\mu_i}{\Longrightarrow}$ accord with $\overset{\mu_i}{\longrightarrow}$, we say that $\mathbf{tr}$ is full.
\end{definition}
\begin{definition} 
Let $\mathbf{tr}$ be a trace $A_0 \overset{\mu_1}{\Longrightarrow} ... \overset{\mu_n}{\Longrightarrow} A_n$ and $\mathbf{tr}'$ be a trace $B_0 \overset{\mu'_1}{\Longrightarrow} ... \overset{\mu'_m}{\Longrightarrow} B_m$.

$\mathbf{tr} \sim_t \mathbf{tr}' \overset{\mathrm{def}}{\Leftrightarrow} n=m \land \mu_i = \mu'_i \land A_i \approx_s B_i\ for\ all\ i$.
\end{definition}
When \textbf{tr} $\sim_t$ \textbf{tr}$'$, we say that they are statically equivalent.
\begin{definition} 
Let $\mathbf{tr}$ be a trace $A_0 \overset{\mu_1}{\Longrightarrow} ... \overset{\mu_n}{\Longrightarrow} A_n$.
We arbitrarily take processes $A_{01},...,A_{0m_0},...,A_{n-1,1},...,A_{n-1,m_{n-1}}$ such that 
\begin{equation*}
A_0 \longrightarrow A_{01} \longrightarrow ... \longrightarrow A_{0l_0} \overset{\mu_1}{\longrightarrow} A_{0,l_0+1} \longrightarrow ... \longrightarrow A_{0m_0} = A_1 \longrightarrow ... \longrightarrow A_n
\end{equation*}
\noindent
and each transition is derived without $\alpha$-conversion. If $\mu_i$ is silent, $m_i$ can be 0. The above derivation is denoted by $\mathrm{unfold(\mathbf{tr})}$.
\end{definition}
Note that $\mathrm{unfold(\mathbf{tr})}$ is not always unique.
\begin{definition}
A trace $\mathbf{tr}$ is safe with respect to $A$. 
$\overset{\mathrm{def}}{\Leftrightarrow}$ Every action in $\mathbf{tr}$ contains no elements in $\mathrm{bn}(A) \cup \mathrm{rv}(A)$ and each transition is derived without $\alpha$-conversion.

If $\mathbf{tr}$ is a trace of $A$, we merely say that $\mathbf{tr}$ is safe or $\mathbf{tr}$ is a safe trace.
\end{definition}
\begin{definition} 
Let A and B be two closed name-variable-distinct extended processes. $A \subseteq_t B$ if and only if we can always complete the procedure below.
\begin{enumerate}
\item We $\alpha$-convert $A$ and $B$ such that $A$ and $B$ are bind-exclusive.
\item We arbitrarily choose a trace $\mathbf{tr}$ of $A$ which is safe with respect to $A$ and $B$. 
\item We take a safe trace of $B$ which is static equivalent to $\mathbf{tr}$.
\end{enumerate}
When $A \subseteq_t B \land B \subseteq_t A$, we say that they are trace equivalent and denote by $A \approx_t B$.

Let A and B be two name-variable-distinct extended processes. We $\alpha$-convert $A$ and $B$ such that $A$ and $B$ are bind-exclusive. Let $\sigma$ be a map that maps a variable in $(\mathrm{fv}(A)\setminus \mathrm{dom}(A)) \cup (\mathrm{fv}(B)\setminus \mathrm{dom}(B))$ to a ground term. 
When $A\sigma \approx_t B\sigma$ for all $\sigma$ and capture-avoiding, we also say that they are trace equivalent and denote as $A \approx_t B$.
\end{definition}
Note that $\sigma$ causes capture-avoiding substitution.

This definition may look strange, but this is equivalent to the standard definition.
\begin{prop}\label{std-def}\leavevmode \par
Let A and B be two closed name-variable-distinct extended processes.

$A \subseteq_t B \Leftrightarrow$ For any trace \textbf{tr} of $A$, there exists a trace $\mathbf{tr}'$ of $B$ such that $\mathbf{tr} \sim_t \mathbf{tr}'$.
\end{prop}
\begin{proof}
$\Rightarrow$) We arbitrarily take a trace \textbf{tr} of $A$.

First, we $\alpha$-convert $A$ and $B$ to $A'$ and $B'$ respectively such that $A'$ and $B'$ are bind-exclusive and they bind no names in actions that are in \textbf{tr}. We replace $A$ with $A'$ in \textbf{tr}.

Secondly, we convert processes in \textbf{tr} to structural equivalent processes such that 
\begin{equation*}
C \overset{\mu}{\Longrightarrow} D\ \mathrm{in}\ \mathbf{tr}\ \Rightarrow \mathrm{fn}(D) \subseteq \mathrm{fn}(C) \cup \mathrm{n}(\mu).
\end{equation*}
Thirdly, we $\alpha$-convert processes in \textbf{tr} such that every transition is derived without $\alpha$-conversion.

Now, we got a safe trace $\mathbf{tr}''$ of $A'$, so we can obtain a safe trace $\mathbf{tr}'$ of $B'$ such that $\mathbf{tr}'' \sim_t \mathbf{tr}'$ because of $A \subseteq_t B$.

Then, we $\alpha$-convert $B'$ in $\mathbf{tr}''$ to $B$ and obtain a trace of $B$ that is statically equivalent to \textbf{tr}.

$\Leftarrow$) We $\alpha$-convert $A$ and $B$ to $A'$ and $B'$ respectively such that $A'$ and $B'$ are bind-exclusive.

Next, we arbitrarily choose a trace \textbf{tr} of $A'$ which is safe with respect to $A'$ and $B'$.

We $\alpha$-convert $A'$ in $\mathbf{tr}$ to $A$.

By assumption, we obtain a trace $\mathbf{tr}'$ of $B$ that is statically equivalent to \textbf{tr}.

We replace $B$ with $B'$ in $\mathbf{tr}'$.

Moreover, we convert processes in $\mathbf{tr}'$ to structural equivalent processes such that
\begin{equation*}
C \overset{\mu}{\Longrightarrow} D\ \mathrm{in\ tr}'\ \Rightarrow \mathrm{fn}(D) \subseteq \mathrm{fn}(C) \cup \mathrm{n}(\mu).
\end{equation*}
In addition, we $\alpha$-convert processes in $\mathbf{tr}'$ such that every transition is derived without $\alpha$-conversion.

Now, we got a safe trace of $B'$ that is statically equivalent to \textbf{tr}.
\end{proof}
\begin{definition}
\begin{align*}
\mathrm{tr}(A) &= \{\mathbf{tr}|\mathbf{tr}\ \mathrm{is\ a\ safe\ trace\ of}\ A\} \\
\mathrm{tr_{max}}(A) &= \{\mathbf{tr} \in \mathrm{tr}(A)|\mathbf{tr}\ \mathrm{is\ maximal}\}
\end{align*}
\end{definition}
Derivations are closed under structural equivalence, so structural equivalence is finer than trace equivalence. In addition, labelled bisimilarity is also finer than trace equivalence.

\begin{example}
We consider
\begin{align*}
P &= \overline{a}\langle s \rangle. \overline{b}\langle s \rangle + \overline{a}\langle s \rangle. \overline{c}\langle s \rangle \\
Q &= \overline{a}\langle s \rangle. (\overline{b}\langle s \rangle + \overline{c}\langle s \rangle).
\end{align*}
They have the same traces, so they are trace equivalent. However, they are not bisimilar. In fact, $P \longrightarrow \overline{a}\langle s \rangle. \overline{b}\langle s \rangle$ is not simulated by $Q$.
\end{example}

We later show that a non-adaptive active attacker cannot distinguish trace equivalent processes.
\begin{problem}\label{tr-pb}\leavevmode \par
\textbf{Input:} Closed extended processes $A, B$ and a trace $\mathbf{tr} \in \mathrm{tr}(B)$.

\textbf{Question:} Does there exist a trace $\mathbf{tr}' \in \mathrm{tr}(A)$ such that $\mathbf{tr} \sim_t \mathbf{tr}'$?
\end{problem}
\begin{prop}\label{tr-ud}
There are signatures for which Problem \ref{tr-pb} is undecidable, even though the input processes are restricted to plain processes.
\end{prop}
\begin{proof}
It is proved in  \cite[Proposition 5]{abadi2006deciding} that static equivalence is undecidable in general.
We reduce the decision problem for static equivalence to Problem \ref{tr-pb}.

Let $\varphi$ and $\psi$ be traces. We assume that $\mathrm{dom}(\varphi) = \mathrm{dom}(\psi)$.

Let $\varphi = \nu \widetilde{n}. \{\fracinline{M_1}{x_1},...,\fracinline{M_l}{x_l}\}, \psi = \nu \widetilde{m}. \{\fracinline{N_1}{x_1},...,\fracinline{N_l}{x_l}\}$.

Let $P = \nu \widetilde{n}. \overline{a}\langle M_1 \rangle ... \overline{a}\langle M_l \rangle, Q = \nu \widetilde{m}. \overline{a}\langle N_1 \rangle ... \overline{a}\langle N_l \rangle$, where $a \notin \widetilde{n} \cup \widetilde{m}$.

Let \textbf{tr} be $P \overset{\nu x_1. \overline{a}\langle x_1 \rangle}{\longrightarrow} \nu \widetilde{n}.(\overline{a}\langle M_2 \rangle ... \overline{a}\langle M_l \rangle | \{\fracinline{M_1}{x_1}\}) \overset{\nu x_2. \overline{a}\langle x_2 \rangle}{\longrightarrow} ... \overset{\nu x_l. \overline{a}\langle x_l \rangle}{\longrightarrow} \nu \widetilde{n}. \{\fracinline{M_1}{x_1},...,\fracinline{M_l}{x_l}\}$.

We prove that $\varphi \approx_s \psi \Leftrightarrow$ there exists a trace $\mathbf{tr}' \in \mathrm{tr}(A)$ such that $\mathbf{tr} \sim_t \mathbf{tr}'$.

A trace $\mathbf{tr}' \in \mathrm{tr}(Q)$ whose actions correspond to \textbf{tr} is the only below:
\begin{equation*}
Q \overset{\nu x_1. \overline{a}\langle x_1 \rangle}{\longrightarrow} ... \overset{\nu x_l. \overline{a}\langle x_l \rangle}{\longrightarrow} \nu \widetilde{m}. \{\fracinline{N_1}{x_1},...,\fracinline{N_l}{x_l}\}.
\end{equation*}
We assume that $\mathbf{tr} \sim_t \mathbf{tr}'$. Then, $\nu \widetilde{n}. \{\fracinline{M_1}{x_1},...,\fracinline{M_l}{x_l}\} \approx_s \nu \widetilde{m}. \{\fracinline{N_1}{x_1},...,\fracinline{N_l}{x_l}\}$, so $\Leftarrow$ holds.

Next, we assume that $\varphi \approx_s \psi$. Then,
\begin{align*}
\mathrm{fr}(\mathbf{tr}[i]) \equiv& \nu x_{i+1} ... x_l. \varphi \\
\mathrm{fr}(\mathbf{tr}'[i]) \equiv& \nu x_{i+1} ... x_l. \psi.
\end{align*}
Therefore, it follows that $\mathbf{tr} \sim_t \mathbf{tr}'$. Namely, $\Rightarrow$ also holds.
\end{proof}
\begin{prop}\label{tr-de}
If the static equivalence on a signature $\Sigma$ is decidable, Problem \ref{tr-pb} is decidable. 
\end{prop}
\begin{proof}
The number of traces in $\mathrm{tr}(A)$ whose actions correspond to ones in \textbf{tr} is finite because every action yields finite processes.

This is why we only have to check whether each process is statically equivalent to the correspondent process in \textbf{tr}.
\end{proof}
\begin{definition}
A convergent subterm theory is an equational theory defined by finite equations whose each right-hand side is a proper subterm of the left-hand side.
\end{definition}
For example, if $\Sigma = \{\mathrm{dec}, \mathrm{enc}\}$ and an equational theory on $\Sigma$ is defined by
\begin{equation*}
\mathrm{dec}(\mathrm{enc}(M, k), k) = M,
\end{equation*}
this theory is convergent subterm.

It is proved in  \cite[Theorem 1]{abadi2006deciding} that static equivalence on a convergent subterm theory is decidable, so the corollary below immediately follows.
\begin{corollary}
If the equational theory on $\Sigma$ is a convergent subterm theory , then Problem \ref{tr-pb} is decidable.
\end{corollary}
\section{Examples}
We express Diffie-Hellman key exchange and the man-in-the-middle attack by the applied pi calculus.

The fundamental Diffie-Hellman protocol enables two parties to make a shared secret key by exchanging messages on public channels.

We consider the binary function symbol $f$ and the unary function symbol $g$. They satisfy $f(x, g(y)) = f(y, g(x))$. We intend that $f(x, y) = y^x\ \mathrm{mod}\ p$ and $g(x) = \alpha^x\ \mathrm{mod}\ p$ where $p$ is a prime, $\alpha$ is a generator of $(\mathbb{Z}/p\mathbb{Z})^*$, and $x$ and $y$ are elements in $(\mathbb{Z}/p\mathbb{Z})^*$.
However, such a specific definition is not necessary for the discussion below.

Diffie-Hellman protocol is expressed in the applied pi calculus as below. Note that $k_A$ and $k_B$ are variables. Alice behaves as $\nu a, k_A. \overline{c} \langle g(a) \rangle .d(x).(P_A|\{f(a,x)/k_A\})$ and Bob behaves as $\nu b, k_B.c(y). \overline{d} \langle g(b) \rangle .(P_B|\{f(b,y)/k_B\})$. We suppose that $P_A$ does not contain $x$ for simplicity. $P_B$ is similar.
\begin{align*} 
&\nu a, k_A. \overline{c} \langle g(a) \rangle .d(x).(P_A|\{f(a,x)/k_A\}) | \nu b, k_B.c(y). \overline{d} \langle g(b) \rangle .(P_B|\{f(b,y)/k_B\}) \\
\longrightarrow& \nu a,k_A.(d(x).(P_A|\{f(a,x)/k_A\}) | \nu b,k_B. \overline{d} \langle g(b) \rangle .P_B|\{f(b,g(a))/k_B\}) \\
\longrightarrow& \nu a,k_A.\nu b,k_B. (P_A|\{f(a,g(b))/k_A\}|P_B|\{f(b,g(a))/k_B\}).
\end{align*}
In the first transition, Alice sent $g(a)$ on a channel $c$, and Bob received it. In the second transition, Bob sent $g(b)$ on a channel $d$, and Alice received it. 
Then, $Alice$ makes a shared key using $g(b)$. That is, she calculates $f(a, g(b))$. Similarly, $Bob$ calculates $f(b, g(a))$. They coincide by the equation above.
The transitions above represent a good case. If someone eavesdrops, it is expressed as below.
\begin{align*} 
&\nu a,k_A. \overline{c} \langle g(a) \rangle .d(x).(P_A|\{f(a,x)/k_A\}) | \nu b,k_B.c(y). \overline{d} \langle g(b) \rangle .(P_B|\{f(b,y)/k_B\}) \\
\overset{\nu z.\overline{c} \langle z \rangle}{\longrightarrow}& \nu a,k_A.(d(x).(P_A|\{f(a,x)/k_A\}) | \{g(a)/z\}) | \nu b,k_B.c(y). \overline{d} \langle g(b) \rangle .(P_B|\{f(b,y)/k_B\}) \\
\overset{c(z)}{\longrightarrow}& \nu a,k_A.(d(x).(P_A|\{f(a,x)/k_A\}) | \{g(a)/z\} | \nu b,k_B. \overline{d} \langle g(b) \rangle .P_B|\{f(b,g(a))/k_B\}) \\
\overset{\nu w.\overline{d} \langle w \rangle}{\longrightarrow}& \nu a,k_A.(d(x).(P_A|\{f(a,x)/k_A\}) | \{g(a)/z\} | \nu b,k_B. ((P_B|\{f(b,y)/k_B\})\{g(a)/y\} | \{g(b)/w\})) \\
\overset{d(w)}{\longrightarrow}& \nu a,k_A.\nu b,k_B.(P_A|\{f(a,g(b))/k_A\} | \{g(a)/z\} | P_B|\{f(b,g(a))/k_B\} | \{g(b)/w\}).
\end{align*}
Unlike the first example, sending and receiving are observable. Note that the domain of the last process contains $z$ and $w$. The environment ---now it is an eavesdropper--- can use them freely. In other words, the eavesdropper gets the value of $g(a)$ and $g(b)$. On the other hand, she cannot directly use $a$ and $b$ because they are restricted. Moreover, she cannot calculate $a$ and $b$ from $g(a)$ and $g(b)$ because such equations do not exist. This corresponds to the complexity of discrete logarithm problems.

If an active attacker exists, this protocol is no longer safe.
\begin{align*} 
&\nu a,k_A. \overline{c} \langle g(a) \rangle .d(x).(P_A|\{f(a,x)/k_A\}) |\nu b,k_B.c(y). \overline{d} \langle g(b) \rangle .(P_B|\{f(b,y)/k_B\}) \\
\overset{\nu z.\overline{c} \langle z \rangle}{\longrightarrow}& \nu a,k_A.(d(x).(P_A|\{f(a,x)/k_A\}) | \{g(a)/z\}) | \nu b,k_B.c(y). \overline{d} \langle g(b) \rangle .(P_B|\{f(b,y)/k_B\}) \\
\overset{c(g(r))}{\longrightarrow}& \nu a,k_A.(d(x).(P_A|\{f(a,x)/k_A\}) | \{g(a)/z\} )| \nu b,k_B. \overline{d} \langle g(b) \rangle .(P_B|\{f(b,g(r))/k_B\}) \\
\overset{\nu w.\overline{d} \langle w \rangle}{\longrightarrow}& \nu a,k_A.(d(x).(P_A|\{f(a,x)/k_A\}) | \{g(a)/z\} )| \nu b,k_B. (P_B|\{f(b,g(r))/k_B\} | \{g(b)/w\}) \\
\overset{d(g(s))}{\longrightarrow}& \nu a,k_A.(P_A|\{f(a,g(s))/k_A\} | \{g(a)/z\} )|\nu b,k_B. (P_B|\{f(b,g(r))/k_B\} | \{g(b)/w\})
\end{align*}
The attacker sent $g(r)$ to Bob and sent $g(s)$ to Alice.
Then, $Alice$ calculates $f(a, g(s))$ and use it as a secret key. The value is same with $f(s, z)$ in this situation, so the attacker can read secret messages encrypted by the key. Bob is similar. This man-in-the-middle attack is well known. This is why Diffie-Hellman protocol needs authentication.

\chapter{Congruency of Trace Equivalence\label{cong}}
This chapter proves the theorem below by case analysis, and this is our main result.
\begin{theorem}\label{main-cong} 
$\approx_t$ is a congruence.
\end{theorem}
In this chapter and the next chapter, we assume that every process is bind-exclusive. Otherwise, we $\alpha$-convert so that it is satisfied.
This chapter and the appendix use many results in  \cite{DBLP:journals/corr/AbadiBF16}

At first, we summarize lemmas for this chapter. We omit all proofs here. See Appendix. We will define the notion of normal processes later (Definition \ref{nml-sb}, \ref{pnf} and the remark before them).
\begin{lemma}\label{drop-sigma}
$P \overset{\alpha \sigma}{\longrightarrow} A \Rightarrow \sigma | P \overset{\alpha}{\longrightarrow} \sigma | A$.
\end{lemma}
\begin{lemma}\label{drop-nu}
$\nu u.A \overset{\mu}{\longrightarrow} B \land A:\mathrm{closed} \land \mathrm{fv}(\mu) \subseteq \mathrm{dom}(\nu u.A) \land \mathrm{n}(\mu) \cap \mathrm{bn}(\nu u.A) = \emptyset$

$\Rightarrow \exists B' \ s.t.\ A \overset{\mu}{\longrightarrow} B' \land B \equiv \nu u.B'$
\end{lemma}
\begin{lemma}\label{change-label}
Let $\nu \widetilde{n}. (\sigma | P)$ be a closed normal process.

$\nu \widetilde{n}. (\sigma | P) \overset{\alpha}{\longrightarrow} A \land \alpha\sigma = \beta\sigma \land \mathrm{fv}(\alpha) \subseteq \mathrm{dom}(\sigma) \land \widetilde{n} \cap (\mathrm{n}(\alpha) \cup \mathrm{n}(\beta)) = \emptyset$

$ \Rightarrow \nu \widetilde{n}. (\sigma | P) \overset{\beta}{\longrightarrow} A$.
\end{lemma}
\begin{lemma}\label{erase-sigma}
$\sigma | A \overset{\mu}{\longrightarrow} \sigma | B \land \mathrm{dom}(\sigma) \cap \mathrm{fv}(\mu) = \emptyset \land x \in \mathrm{dom}(\sigma) \Rightarrow x\sigma$: closed 

$\Rightarrow A\sigma \overset{\mu}{\longrightarrow} B\sigma$.
\end{lemma}
\begin{lemma}\label{shift-sigma}
$\sigma | P \overset{\mu}{\longrightarrow} B \land \sigma | P$: closed normal $\land \mathrm{fv}(\mu) \subseteq \mathrm{dom}(\sigma)$

$\Rightarrow \exists B'\ \mathrm{s.t.}\ P\sigma \overset{\mu\sigma}{\longrightarrow} B' \land B \equiv \sigma | B'$.
\end{lemma}
\section{The Case of Applying a Context without Composition}
This case is straightforward, and the proof is simple, but some propositions rely on Proposition \ref{para}.
\begin{prop} 
$P \approx_t Q \Rightarrow \overline{M} \langle N \rangle .P \approx_t \overline{M} \langle N \rangle .Q$
\end{prop}
\begin{proof}
Assume that $P \approx_t Q$. 

For all assignments $\sigma$, it is sufficient to prove that $\overline{M\sigma} \langle N\sigma \rangle .P\sigma \approx_t \overline{M\sigma} \langle N\sigma \rangle .Q\sigma$. Thus, we can suppose that $P, Q, M$ and $N$ are closed without loss of generality.
We arbitrarily take a safe trace \textbf{tr} of $\overline{M} \langle N \rangle .P$. 
Here, unfold(\textbf{tr}) is $\overline{M} \langle N \rangle .P \overset{\nu x.\overline{M} \langle x \rangle}{\longrightarrow} A \overset{\mu}{\longrightarrow}$ ... and $A \equiv P|\{N/x\}$ for some $x$.

By Proposition \ref{para}, $P|\{N/x\} \approx_t Q|\{N/x\}$, so there exists a trace of $Q|\{N/x\}$ which is statically equivalent to the part $A \overset{\mu}{\longrightarrow}$ ... of unfold(\textbf{tr}). We add $\overline{M} \langle N \rangle .Q \overset{\nu x.\overline{M} \langle x \rangle}{\longrightarrow}$ to it and we get a trace of $\overline{M} \langle N \rangle .Q$ that is statically equivalent to unfold(\textbf{tr}). We omit some silent actions and get the desired trace. Hence, $ \overline{M} \langle N \rangle .P \subseteq_t  \overline{M} \langle N \rangle .Q$ and vice versa.
\end{proof}
\begin{prop} 
$P \approx_t Q \Rightarrow M(x).P \approx_t M(x).Q$
\end{prop}
\begin{proof}
Assume that $P \approx_t Q$. 

For all assignments $\sigma$, it is sufficient to prove that $M\sigma (x).P\sigma \approx_t M\sigma (x).Q\sigma$. Thus, we can suppose that $M$ is closed and $\mathrm{fv}(P) \cup \mathrm{fv}(Q) \subseteq \{x\}$ without loss of generality. Note that $x$ is not mapped by $\sigma$ because it is not free in $M(x).P$ and $M(x).Q$.
We arbitrarily take a safe trace \textbf{tr} of $M(x).P$.
Here, unfold(\textbf{tr}) is $M(x).P \overset{M(N)}{\longrightarrow} P' \overset{\mu}{\longrightarrow}$ ... and $P' \equiv P\{N/x\}$ for some ground $N$.
By $P \approx_s Q$, there exists a trace of $Q\{N/x\}$ which is statically equivalent to the part $P' \overset{\mu}{\longrightarrow}$ ... of unfold(\textbf{tr}). We add $M(x).Q  \overset{M(N)}{\longrightarrow}$ to it and we get the desired trace of $M(x).Q$ that is statically equivalent to unfold(\textbf{tr}). We omit some silent actions and get the desired trace. Hence, $M(x).P \subseteq_t M(x).Q$ and vice versa.
\end{proof}
\begin{prop} 
$A \approx_t B \Rightarrow \nu u.A \approx_t \nu u.B$ (When $u$ is a variable, $u \in \mathrm{dom}(A)$.)
\end{prop}
\begin{proof}
Assume that $A \approx_t B$. 

For all assignments $\sigma$, it is sufficient to prove that $(\nu u.A)\sigma \approx_t (\nu u.B)\sigma$. Thus, we can suppose that $A$ and $B$ are closed without loss of generality.
We arbitrarily take a safe trace \textbf{tr} of $\nu u.A$.
Here, unfold(\textbf{tr}) is $\nu u.A \overset{\mu_1}{\longrightarrow} C_1 \overset{\mu_2}{\longrightarrow}...$.
By Lemma \ref{drop-nu}, there exists $C_1'$ such that
$A \overset{\mu_1}{\longrightarrow} C_1' \land C_1 \equiv \nu u.C_1'$.
Moreover, we use Lemma \ref{drop-nu} repeatedly, and we get
\begin{align*}
\mathrm{unfold(\textbf{tr})} \sim_t \nu u.A \overset{\mu_1}{\longrightarrow} \nu u.C_1' \overset{\mu_2}{\longrightarrow}& \nu u.C_2' \overset{\mu_3}{\longrightarrow} ... \\
A  \overset{\mu_1}{\longrightarrow}\ \ \ \ \  C_1' \overset{\mu_2}{\longrightarrow}&\ \ \ \ \ C_2' \overset{\mu_3}{\longrightarrow} ...
\end{align*}
We take a trace of $B$ which is statically equivalent to the second line and add $\nu u.$ to each process. This is possible because \textbf{tr} is safe. Recall that every action in a safe trace contains no bound names and restricted variables. Hence, $\mu_i$ does not contain $u$.
\begin{align*}
B  \overset{\mu_1}{\longrightarrow}\ \ \ \ \  D_1 \overset{\mu_2}{\longrightarrow}&\ \ \ \ \ D_2 \overset{\mu_3}{\longrightarrow} ... \\
\nu u.B \overset{\mu_1}{\longrightarrow} \nu u.D_1 \overset{\mu_2}{\longrightarrow}& \nu u.D_2 \overset{\mu_3}{\longrightarrow} ...
\end{align*}
The last line is statically equivalent to unfold(\textbf{tr}). Note that restriction preserves static equivalence. We omit some internal reductions and get the desired trace. Hence, $\nu u.A \subseteq_t \nu u.B$ and vice versa.
\end{proof}
\begin{prop} 
$P \approx_t Q \Rightarrow \mathrm{if} ~M=N~ \mathrm{then} ~P~ \mathrm{else} ~R \approx_t \mathrm{if} ~M=N~ \mathrm{then} ~Q~ \mathrm{else} ~R$.
\end{prop}
\begin{proof}
Assume that $P \approx_t Q$. For all assignments $\sigma$, it is sufficient to prove that 
\begin{equation*}
\mathrm{if} ~M\sigma=N\sigma~ \mathrm{then} ~P\sigma~ \mathrm{else} ~R\sigma \approx_t \mathrm{if} ~M\sigma=N\sigma~ \mathrm{then} ~Q\sigma~ \mathrm{else} ~R\sigma.
\end{equation*}
Thus, we can suppose that $P, Q, R, M$ and $N$ are closed without loss of generality.
We arbitrarily take a safe trace \textbf{tr} of $\mathrm{if} ~M=N~ \mathrm{then} ~P~ \mathrm{else} ~R$.
Here, unfold(\textbf{tr}) is 
\begin{equation*}
\mathrm{if} ~M=N~ \mathrm{then} ~P~ \mathrm{else} ~R \longrightarrow A \overset{\mu}{\longrightarrow} ....
\end{equation*}
We know that $A \equiv P \lor A \equiv R$, so we can regard $A \overset{\mu}{\longrightarrow} ...$ as a trace of $P$ or $R$.
In the former case, there exists a trace of $Q$ which is statically equivalent to it.
We add $\mathrm{if} ~M=N~ \mathrm{then} ~Q~ \mathrm{else} ~R \longrightarrow$ to it and get a trace of $\mathrm{if} ~M=N~ \mathrm{then} ~Q~ \mathrm{else} ~R$ which is statically equivalent to unfold(\textbf{tr}).
In the latter case, we add $\mathrm{if} ~M=N~ \mathrm{then} ~Q~ \mathrm{else} ~R \longrightarrow$ to $A \overset{\mu}{\longrightarrow} ...$ and get a trace of $\mathrm{if} ~M=N~ \mathrm{then} ~Q~ \mathrm{else} ~R$ which is statically equivalent to unfold(\textbf{tr}).
We omit some internal reductions and get the desired trace. Hence, 
\begin{equation*}
\mathrm{if} ~M=N~ \mathrm{then} ~P~ \mathrm{else} ~R \subseteq_t \mathrm{if} ~M=N~ \mathrm{then} ~Q~ \mathrm{else} ~R
\end{equation*}
and vice versa.
\end{proof}
\begin{prop} 
$P \approx_t Q \Rightarrow P+R \approx_t Q+R$.
\end{prop}
\begin{proof}
Assume that $P \approx_t Q$. For all assignments $\sigma$, it is sufficient to prove that 
\begin{equation*}
P\sigma+R\sigma \approx_t Q\sigma+R\sigma.
\end{equation*}
Thus, we can suppose that $P, Q$ and $R$ are closed without loss of generality.
We arbitrarily take a safe trace \textbf{tr} of $P+R$.
Here, unfold(\textbf{tr}) is $P+R \longrightarrow A \overset{\mu}{\longrightarrow} ...$.
We know that $A \equiv P \lor A \equiv R$, so we can regard $A \overset{\mu}{\longrightarrow} ...$ as a trace of $P$ or $R$.
In the former case, there exists a trace of $Q$ which is statically equivalent to it.
We add $Q+R \longrightarrow$ to it and get a trace of $Q+R$ which is statically equivalent to unfold(\textbf{tr}).
In the latter case, we add $Q+R \longrightarrow$ to $A \overset{\mu}{\longrightarrow} ...$ and get a trace of $Q+R$ which is statically equivalent to unfold(\textbf{tr}).
We omit some internal reductions and get the desired trace. Hence, $P+R \subseteq_t Q+R$ and vice versa.
\end{proof}
Hereafter, we often use partial normal forms \cite{DBLP:journals/corr/AbadiBF16}. $\mathrm{pnf}(A)$ is the partial normal form of $A$. We call a process which is partial normal form a normal process.
\begin{definition}\label{nml-sb}
Let $\sigma$ and $\rho$ be active substitutions. 

We assume that $\sigma | \rho \equiv \{M_1/x_1,...,M_n/x_n\}$ where each $M_i$ does not contain $x_1,...,x_i$ by reordering.
We define $\sigma_0 = 0$ and $\sigma_{i+1} = \sigma_i\{M_{i+1}/x_{i+1}\} | \{M_{i+1}/x_{i+1}\}$. Then, we also define $\sigma \uplus \rho = \sigma_n$.
\end{definition}
\begin{definition}\label{pnf}
We suppose that $\mathrm{pnf}(A)=\nu \widetilde{n}. (\sigma | P)$ and $\mathrm{pnf}(B)=\nu \widetilde{m}. (\rho | Q)$.
\begin{align*}
\mathrm{pnf}(P) =& 0|P \\
\mathrm{pnf}(\{M/x\}) =& \{M/x\}|0 \\
\mathrm{pnf}(\nu n.A) =& \nu n\widetilde{n}. (\sigma | P) \\
\mathrm{pnf}(\nu x.A) =& \nu \widetilde{n}. (\sigma_{\mathrm{dom}(\sigma)\setminus \{x\}} | P) \\
\mathrm{pnf}(A|B) =& \nu \widetilde{n}\widetilde{m}. (\sigma \uplus \rho | (P | Q)(\sigma \uplus \rho)).
\end{align*}
\end{definition}
Note that $A \equiv \mathrm{pnf}(A)$.
\begin{lemma}\label{silent-rep-trace} 
If $A$ and $P$ are closed and $A|!P \longrightarrow B$, it holds one of the following:
\begin{enumerate}
\item\label{Aint} $\exists A'\ \mathrm{s.t.}\ A \longrightarrow A' \land B \equiv A'|!P$.
\item\label{Pint} $\exists P''\ \mathrm{s.t.}\ P|P \longrightarrow P'' \land B \equiv A|P''|!P$.
\item\label{APcomm} $\exists E\ \mathrm{s.t.}\ A|P \longrightarrow E \land B \equiv E|!P$.
\end{enumerate}
\begin{proof}
Let $\mathrm{pnf}(A) = \nu \widetilde{n}.(\sigma|Q)$.

Then, $\mathrm{pnf}(A|!P) = \nu \widetilde{n}.(\sigma|Q|!P)$. Note that $\mathrm{n}(P) \cap \widetilde{n} = \emptyset$.

By  \cite[Lemma B.23]{DBLP:journals/corr/AbadiBF16}, $Q|!P \longrightarrow R \land B \equiv  \nu \widetilde{n}.(\sigma|R)$ for some closed $R$.

By  \cite[Lemma B.24]{DBLP:journals/corr/AbadiBF16}, we consider the following four cases:
\begin{enumerate}
\item $Q \longrightarrow Q' \land R \equiv Q'|!P$ for some closed $Q'$.

$A \equiv \nu \widetilde{n}.(\sigma|Q) \longrightarrow  \nu \widetilde{n}.(\sigma|Q')$ and $\nu \widetilde{n}.(\sigma|Q') | !P \equiv \nu \widetilde{n}.(\sigma|Q'|!P) \equiv B$. 

Then, $A' = \nu \widetilde{n}.(\sigma|Q')$ satisfies case \ref{Aint} of this lemma.

\item $!P \longrightarrow P' \land R \equiv Q|P'$ for some closed $P'$.

By  \cite[Lemma B.24]{DBLP:journals/corr/AbadiBF16}, $P|P \longrightarrow P'' \land P' \equiv P''|!P$ for some closed $P''$.

$B \equiv \nu \widetilde{n}. (\sigma | R) \equiv \nu \widetilde{n}.(\sigma|Q)|P''|!P \equiv A | P'' | !P$, so case \ref{Pint} of this lemma is satisfied.

\item $Q \overset{N(x)}{\longrightarrow} B' \land !P \overset{\nu x.\overline{N} \langle x \rangle}{\longrightarrow} C \land R \equiv \nu x.(B'|C)$ for some $B', C, x$ and ground $N$.

By  \cite[Lemma B.18]{DBLP:journals/corr/AbadiBF16}, $P \overset{\nu x.\overline{N} \langle x \rangle}{\longrightarrow} D \land C \equiv D|!P$ for some $D$.

Then, 
\begin{align*}
A|!P \longrightarrow B \equiv& \nu \widetilde{n}.(\sigma|\nu x.(B'|D|!P)) \\
\equiv& \nu \widetilde{n}.(\sigma|\nu x.(B'|D)|!P) \\
\equiv& \nu \widetilde{n}.(\sigma|\nu x.(B'|D))|!P
\end{align*}
Thus, $E = \nu \widetilde{n}.(\sigma|\nu x.(B'|D))$ satisfies case \ref{APcomm} of this lemma.

\item $!P \overset{N(x)}{\longrightarrow} B' \land Q \overset{\nu x.\overline{N} \langle x \rangle}{\longrightarrow} C \land R \equiv \nu x.(B'|C)$ for some $B', C, x$ and ground $N$.

By  \cite[Lemma B.18]{DBLP:journals/corr/AbadiBF16}, $P \overset{N(x)}{\longrightarrow} D \land B' \equiv D|!P$ for some $D$.

Then,
\begin{align*}
A|!P \longrightarrow B \equiv& \nu \widetilde{n}.(\sigma|\nu x.(D|!P|C)) \\
\equiv& \nu \widetilde{n}.(\sigma|\nu x.(D|C)|!P) \\
\equiv& \nu \widetilde{n}.(\sigma|\nu x.(D|C))|!P
\end{align*}
Thus, $E = \nu \widetilde{n}.(\sigma|\nu x.(D|C))$ satisfies case \ref{APcomm} of this lemma.
\end{enumerate}
\end{proof}
\begin{corollary}\label{cor-silent-rep} 
If $A$ and $P$ are closed and $A|!P \longrightarrow B$, then $A|P|P \longrightarrow C \land B \equiv C|!P$ for some $C$.
\end{corollary}
\begin{lemma}\label{label-rep-trace} 
If $A$ and $P$ are closed and $A|!P \overset{\alpha}{\longrightarrow} B$ and $\mathrm{fv}(\alpha) \subseteq \mathrm{dom}(A)$ and $\mathrm{n}(\alpha) \cap \mathrm{bn}(A|!P) = \emptyset$, it holds one of the following:
\begin{enumerate}
\item\label{left} $\exists E\ \mathrm{s.t.}\ A \overset{\alpha}{\longrightarrow} E \land B \equiv E|!P$.
\item\label{right} $\exists F\ \mathrm{s.t.}\ A|P \overset {\alpha}{\longrightarrow} F \land B \equiv F|!P$.
\end{enumerate}
\end{lemma}
\begin{proof}
Let $\mathrm{pnf}(A) = \nu \widetilde{n}.(\sigma|Q)$.

Then, $\mathrm{pnf}(A|!P) = \nu \widetilde{n}.(\sigma|Q|!P)$.

By  \cite[Lemma B.19]{DBLP:journals/corr/AbadiBF16}, $Q|!P \overset{\alpha\sigma}{\longrightarrow} C \land B \equiv \nu \widetilde{n}.(\sigma|C)$ for some $C$.

By  \cite[Lemma B.18]{DBLP:journals/corr/AbadiBF16}, we consider the following two cases:
\begin{enumerate}
\item $Q \overset{\alpha\sigma}{\longrightarrow} D \land C \equiv D|!P$ for some $D$.

By Lemma \ref{drop-sigma}, $\sigma | Q \overset{\alpha}{\longrightarrow} \sigma | D$.

Thus, $A \equiv \nu \widetilde{n}.(\sigma|Q) \overset{\alpha}{\longrightarrow} \nu \widetilde{n}.(\sigma|D)$.

In addition, $\nu \widetilde{n}.(\sigma|D)|!P \equiv \nu \widetilde{n}.(\sigma|D|!P) \equiv B$, so $E = \nu \widetilde{n}.(\sigma|D)$ satisfies case \ref{left} of this lemma.

\item $!P \overset{\alpha\sigma}{\longrightarrow} D \land C \equiv Q|D$ for some $D$.

By  \cite[Lemma B.18]{DBLP:journals/corr/AbadiBF16}, $P \overset{\alpha\sigma}{\longrightarrow} E \land D \equiv E|!P$ for some $E$.

By Lemma \ref{drop-sigma}, $\sigma | P \overset{\alpha}{\longrightarrow} \sigma | E$.

Thus, $A | P \equiv \nu \widetilde{n}.(\sigma|Q)|P \overset{\alpha}{\longrightarrow} \nu \widetilde{n}.(\sigma|Q|E)$.

In addition, $\nu \widetilde{n}.(\sigma|Q|E) | !P \equiv \nu \widetilde{n}.(\sigma|Q|E|!P) \equiv B$, so $F = \nu \widetilde{n}.(\sigma|Q|E)$ satisfies case \ref{right} of this lemma.
\end{enumerate}
\end{proof}
\end{lemma}
\begin{corollary}\label{cor-label-rep} 
If $A$ and $P$ are closed and $A|!P \overset{\alpha}{\longrightarrow} B$ and $\mathrm{fv}(\alpha) \subseteq \mathrm{dom}(A)$ and $\mathrm{n}(\alpha) \cap \mathrm{bn}(A|!P) = \emptyset$, then $A|P \overset{\alpha}{\longrightarrow} C \land B \equiv C|!P$ for some $C$.
\end{corollary}
\begin{prop} 
$P \approx_t Q \Rightarrow !P \approx_t !Q$.
\end{prop}
\begin{proof}
Assume that $P \approx_t Q$.

For all assignments $\sigma$, it is sufficient to prove that $!P\sigma \approx_t !Q\sigma$. Thus, we can suppose that $P$ and $Q$ are closed without loss of generality.

We arbitrarily take a safe trace \textbf{tr} of $!P$. By corollary \ref{cor-silent-rep}, \ref{cor-label-rep}, we obtain a trace which is of the form $P^n|!P \overset{\mu_1}{\longrightarrow} A_1|!P \overset{\mu_2}{\longrightarrow} A_2|!P \overset{\mu_3}{\longrightarrow} ...$ and statically equivalent to unfold(\textbf{tr}). We also obtain $P^n \overset{\mu_1}{\longrightarrow} A_1 \overset{\mu_2}{\longrightarrow} A_2 \overset{\mu_3}{\longrightarrow} ...$ to convert each process in unfold(\textbf{tr}). Here, $n$ depends on \textbf{tr} and $P^n$ is $n$ concurrent processes. Strictly speaking, processes in $P^n$ are not same when $\mathrm{bn}(P) \neq \emptyset$. In this case, we $\alpha$-convert $P$ and make processes in $P^n$ bind-exclusive. 

By Proposition \ref{para}, $P^n \approx_t Q^n$, so there exists a trace of $Q^n$ which is statically equivalent to $P^n \overset{\mu_1}{\longrightarrow} A_1 \overset{\mu_2}{\longrightarrow} A_2 \overset{\mu_3}{\longrightarrow} ...$.

We add $|!Q$ to each process in the trace, omit some internal reductions and get the desired trace. Hence, $!P \subseteq_t !Q$ and vice versa.
\end{proof}
\section{The Case of Parallel Composition}
\begin{prop}\label{para} 
$A \approx_t B \Rightarrow A|C \approx_t B|C$.
\end{prop}
This proof is very long and complex, so we present an outline at first. We suppose that $A, B$ and $C$ are closed as before. 

First, we define a concurrent normal form. This is a special form of a trace of a process which is of the form $A|C$. A concurrent normal trace completely captures the change of scope of bound names.

Secondly, we prove that every safe trace of $A|C$ can be transformed into a concurrent normal form. This proof is constructive.

Thirdly, we constitute a correspondent trace from the concurrent normal trace by division into cases.

Finally, we convert the extracted traces, combine them and prove that the result is statically equivalent to the given trace.

In this section, we always assume that structural equivalent processes are constructed without $\alpha$-conversion.
\subsection{A Definition of a Concurrent Normal Trace}
\begin{definition}
$[\sigma\rho] = (\sigma \uplus \rho)_{|\mathrm{dom}(\sigma)}$.
\end{definition}
\begin{lemma}\label{bra}
$\sigma\rho \equiv [\sigma\rho]$.
\end{lemma}
\begin{proof}
See Appendix.
\end{proof}
\begin{lemma}\label{sub-ex}
$\rho\sigma \equiv \rho[\sigma\rho]$.
\end{lemma}
\begin{proof}
See Appendix.
\end{proof}
\begin{lemma}\label{cket}
$(\sigma | P)\rho \equiv [\sigma\rho] | P[\rho\sigma]$.
\end{lemma}
\begin{proof}
See Appendix.
\end{proof}
Concurrent normal form completely records communication.
\begin{definition}\label{cnf}
A concurrent normal trace $\mathbf{tr}$ of $A|C$ is a trace which satisfies the following conditions.
\begin{enumerate}
\item $\mathbf{tr}$ is full and safe.

\item Each process in $\mathbf{tr}$ is of the form $\nu \widetilde{r_m} \widetilde{s_m}.(\nu \widetilde{x_m}. A_m \rho_m | \nu \widetilde{y_m}. C_m \sigma_m)$. 

In addition, The conditions below are satisfied.
\begin{itemize}
\item $\mathrm{fv}(A_m)\setminus \mathrm{dom}(A_m) \subseteq \mathrm{dom}(\rho_m) \land \mathrm{fv}(C_m)\setminus \mathrm{dom}(C_m) \subseteq \mathrm{dom}(\sigma_m)$.
\item $\mathrm{n}(\rho_m) \cap \widetilde{s_m} = \emptyset \land \mathrm{n}(\sigma_m) \cap \widetilde{r_m} = \emptyset$.  \item$A_m = \sigma_m | P_m \land C_m = \rho_m | Q_m$ and they are normal for some $P_m$ and $Q_m$, and $\mathrm{n}(Q_m) \cap \widetilde{s_m} = \emptyset \land \mathrm{n}(P_m) \cap \widetilde{r_m} = \emptyset$.
\end{itemize}
Let $D_m = \nu \widetilde{r_m} \widetilde{s_m}.(\nu \widetilde{x_m}. A_m \rho_m | \nu \widetilde{y_m}. C_m \sigma_m)$.

\item\label{tau} For every $D_m \longrightarrow D_{m+1}$ in $\mathbf{tr}$, it holds one of the following:
\begin{enumerate}
\item\label{Atau} $A_m \rho_m \longrightarrow A_{m+1}\rho_m$ and elements in $D_{m+1}$ except for $A_{m+1}$ are same as counterparts in $D_m$. 

\item\label{Ctau}  $C_m \sigma_m \longrightarrow C_{m+1}\sigma_m$ and elements in $D_{m+1}$ except for $C_{m+1}$ are same as counterparts in $D_m$. 

\item\label{AinCout} 
\begin{itemize}
\item $A_m \rho_m \overset{N_1(x)}{\longrightarrow} A_{m+1}\rho_m \land C_m\sigma_m \overset{\nu x.\overline{N_2}\langle x \rangle}{\longrightarrow} \nu \widetilde{m}. C_{m+1}\sigma_m$,
\item $\widetilde{r_{m+1}}=\widetilde{r_m}\widetilde{m}$, 
\item $\rho_{m+1} = \rho_{m} \uplus \{M/x\}$,
\item $\widetilde{m} = \mathrm{n}(M) \cap \mathrm{bn}(C_m)$, 
\item $N_1[\sigma_m \rho_m]=N_2[\rho_m \sigma_m]$: ground,
\item $\widetilde{r_m} \cap \mathrm{n}(N_1)=\emptyset \land \widetilde{s_m} \cap \mathrm{n}(N_2)=\emptyset$, 
\item $M[\sigma_m \rho_m]$: ground,
\item $\widetilde{y_{m+1}}=\widetilde{y_m}x$,
\item $\mathrm{n}(M) \cap \widetilde{s_m} = \emptyset$,
\end{itemize}
for some $N_1, N_2, M$ and $x$, and the other parts of $D_{m+1}$ are same as counterparts in $D_m$. 
\item\label{AoutCin}
\begin{itemize}
\item $C_m \sigma_m \overset{N_2(x)}{\longrightarrow} C_{m+1}\sigma_m$,
\item $A_m\rho_m \overset{\nu x.\overline{N_1}\langle x \rangle}{\longrightarrow} \nu \widetilde{m}. A_{m+1}\rho_m$, 
\item $\widetilde{s_{m+1}}=\widetilde{s_m}\widetilde{m}$, 
\item $\sigma_{m+1} = \sigma_{m} \uplus \{M/x\}$,
\item $\widetilde{m} = \mathrm{n}(M) \cap \mathrm{bn}(A_m)$, 
\item $N_1[\sigma_m \rho_m]=N_2[\rho_m \sigma_m]$: ground, 
\item $\widetilde{s_m} \cap \mathrm{n}(N_2)=\emptyset$, 
\item $\widetilde{r_m} \cap \mathrm{n}(N_1)=\emptyset$.,
\item $M[\rho_m \sigma_m]$: ground,
\item $ \widetilde{x_{m+1}}=\widetilde{x_m}x$,
\item $\mathrm{n}(M) \cap \widetilde{r_m} = \emptyset$,
\end{itemize}
for some $N_1, N_2, M$ and $x$, and the other parts of $D_{m+1}$ are same as counterparts in $D_m$. 
\end{enumerate}
\item For every $D_m \overset{N(M)}{\longrightarrow} D_{m+1}$ in $\mathbf{tr}$, it holds one of the following:
\begin{enumerate}
\item\label{Ain} $A_m \rho_m \overset{N(M)[\rho_m\sigma_m]}{\longrightarrow} A_{m+1}\rho_m$ and elements in $D_{m+1}$ except for $A_{m+1}$ are same as counterparts in $D_m$. 

\item\label{Cin}  $C_m \sigma_m \overset{N(M)[\sigma_m\rho_m]}{\longrightarrow} C_{m+1}\sigma_m$ and elements in $D_{m+1}$ except for $C_{m+1}$ are same as counterparts in $D_m$. 
\end{enumerate}
\item For every $D_m \overset{\nu x.\overline{N}\langle x \rangle}{\longrightarrow} D_{m+1}$ in $\mathbf{tr}$, it holds one of the following:
\begin{enumerate}
\item\label{Aout}
\begin{itemize}
\item $A_m \rho_m \overset{\nu x.\overline{N}\langle x \rangle [\rho_m \sigma_m]}{\longrightarrow} \nu \widetilde{m}.A_{m+1}\rho_m$, 
\item $\sigma_{m+1}=\sigma_m \uplus \{M/x\}$,
\item $\widetilde{s_{m+1}}=\widetilde{s_m}\widetilde{m}$,
\item $\widetilde{m} = \mathrm{n}(M) \cap \mathrm{bn}(A_m)$
\item $M[\rho_m \sigma_m]$: ground,
\item $\mathrm{n}(M) \cap \widetilde{r_m} = \emptyset$, 
\end{itemize}
and the other parts of $D_{m+1}$ are same as counterparts in $D_m$. 

\item\label{Cout}
\begin{itemize}
\item $C_m \sigma_m \overset{\nu x.\overline{N}\langle x \rangle [\sigma_m \rho_m]}{\longrightarrow} \nu \widetilde{m}.C_{m+1}\sigma_m$,
\item $\rho_{m+1}=\rho_m \uplus \{M/x\}$,
\item $\widetilde{r_{m+1}}=\widetilde{r_m}\widetilde{m}$, 
\item $\widetilde{m} = \mathrm{n}(M) \cap \mathrm{bn}(C_m)$, 
\item $M[\sigma_m \rho_m]$: ground,
\item $\mathrm{n}(M) \cap \widetilde{s_m} = \emptyset$, 
\end{itemize}
and the other parts of $D_{m+1}$ are same as counterparts in $D_m$. 
\end{enumerate}
\item
\begin{enumerate}
\item At \ref{Atau}, let $A_m'$ be a process obtained by applying $\rho_m$ to only for the part related to the transition from $A_m$. Then, $A_m' \longrightarrow A_{m+1}$, and \ref{Ctau} is similar.


\item At \ref{AinCout}, let $A_m'$ be a process obtained by substituting only for the part related to the transition from $A_m$. Then, $A_m' \overset{N_1[\sigma_m \rho_m](x)}{\longrightarrow} A_{m+1}$. $C_m$ is similar, and \ref{AoutCin} is similar.


\item At \ref{Ain}, let $A_m'$ be a process obtained by substituting only for the part related to the transition from $A_m$. Then, $A_m' \overset{N(\rho_m \uplus \sigma_m)(M)}{\longrightarrow} A_{m+1}$, and \ref{Cin} is similar.

\item At \ref{Aout}, let $A_m'$ be a process obtained by substituting only for the part related to the transition from $A_m$. Then, $A_m' \overset{\nu x.\overline{N(\rho_m \uplus \sigma_m)}\langle x \rangle}{\longrightarrow} A_{m+1}$, and \ref{Cout} is similar.

\end{enumerate}
\end{enumerate}
\end{definition}

\subsection{A Transformation into a Concurrent Normal Form}\label{trans}
We can transform all full safe traces of a parallel composed process into a concurrent normal form.

\begin{lemma}
For any full safe trace of $A|C$, there exists a concurrent normal trace of $A|C$ such that they are statically equivalent.
\end{lemma}
\begin{proof}
We transform the given trace in order from the top. Assuming that we transformed the first $m$-th process.

We suppose that $\nu \widetilde{r} \widetilde{s}.(\nu \widetilde{x}.A_m \rho | \nu \widetilde{y}. C_m \sigma) \overset{\mu}{\longrightarrow} D$.

By Lemma \ref{drop-nu}, $\nu \widetilde{x}.A_m \rho | \nu \widetilde{y}. C_m \sigma  \overset{\mu}{\longrightarrow} D' \land D \equiv \nu \widetilde{r} \widetilde{s}.D'$.

Let $A_m = \sigma | P \land C_m = \rho | Q$. Note that $A_m\rho \equiv [\sigma\rho] | P[\rho\sigma]$ and $C_m\sigma \equiv [\rho\sigma] | Q[\sigma\rho]$.

Then, $\nu \widetilde{x}\widetilde{y}. \sigma \uplus \rho | P[\rho\sigma] | Q[\sigma\rho] \overset{\mu}{\longrightarrow} D'$. That is, $(\sigma \uplus \rho)_{|dom(\sigma \uplus \rho)\setminus \widetilde{x}\widetilde{y}}| P[\rho\sigma] | Q[\sigma\rho] \overset{\mu}{\longrightarrow} D'$.

First, we consider when $\mu$ is silent.

By  \cite[Lemma B.23]{DBLP:journals/corr/AbadiBF16}, $P[\rho\sigma] | Q[\sigma\rho] \longrightarrow R \land D' \equiv (\sigma \uplus \rho)_{|dom(\sigma \uplus \rho)\setminus \widetilde{x}\widetilde{y}}| R$ for some closed $R$.

We consider four cases.
\begin{enumerate}
\item\label{Pint} $P[\rho\sigma] \longrightarrow P'[\rho\sigma] \land R \equiv P'[\rho\sigma] | Q[\sigma\rho]$ for some $P'$ where $\mathrm{fv}(P') \subseteq \mathrm{dom}(\rho)$.
Furthermore, let $P''$ be a process obtained by applying $[\rho\sigma]$ to only for the part related to the transition from $P$. Then, $P'' \longrightarrow P'$.

$A_m \rho \longrightarrow [\sigma \rho] | P'[\rho\sigma] \equiv (\sigma | P') \rho$.

Hence, 
\begin{align*}
\nu \widetilde{r} \widetilde{s}.(\nu \widetilde{x}.A_m \rho | \nu \widetilde{y}. C_m \sigma) \longrightarrow D & \equiv \nu \widetilde{r}\widetilde{s}. ((\sigma \uplus \rho)_{|dom(\sigma \uplus \rho)\setminus \widetilde{x}\widetilde{y}}| P'[\rho\sigma] | Q[\sigma\rho]) \\
& \equiv \nu \widetilde{r} \widetilde{s}.(\nu \widetilde{x}.(\sigma | P') \rho | \nu \widetilde{y}. C_m \sigma).
\end{align*}
This is the desired form.

\item $Q[\sigma\rho] \longrightarrow Q'[\sigma\rho] \land R \equiv P[\rho\sigma] | Q'[\sigma\rho]$ for some $Q'$ where $\mathrm{fv}(Q') \subseteq \mathrm{dom}(\sigma)$.
Furthermore, let $Q''$ be a process obtained by applying $[\sigma\rho]$ to only for the part related to the transition from $Q$. Then, $Q'' \longrightarrow Q'$.

This case is similar to case \ref{Pint}.

\item\label{PinQout} $P[\rho\sigma] \overset{N(x)}{\longrightarrow} E \land Q[\sigma\rho] \overset{\nu x. \overline{N} \langle x \rangle}{\longrightarrow} F \land R \equiv \nu x. (E|F)$ for some $E, F, x$ and ground $N$.

In this case, $A_m \rho \overset{N_1(x)}{\longrightarrow} [\sigma\rho] | E \land C_m \sigma \overset{\nu x. \overline{N_2} \langle x \rangle}{\longrightarrow} [\rho\sigma] | F$ for some $N_1$ and $N_2$ where $N_1[\sigma\rho] = N \land \widetilde{r} \cap \mathrm{n}(N_1) = \emptyset \land N_2[\rho\sigma] = N \land \widetilde{s} \cap \mathrm{n}(N_2) = \emptyset$.

By  \cite[Lemma B.10]{DBLP:journals/corr/AbadiBF16}, there exist $\widetilde{l}, P'$ and $P_2$ such that 
\begin{equation*}
P[\rho\sigma] \equiv (\nu \widetilde{l}. (N_2(x). P' | P_2))[\rho \sigma] \land E \equiv (\nu \widetilde{l}. (P' | P_2))[\rho \sigma].
\end{equation*}
Hence $A_m \rho \overset{N_1(x)}{\longrightarrow} (\sigma | \nu \widetilde{l}. (P' | P_2))\rho$.

We use again  \cite[Lemma B.10]{DBLP:journals/corr/AbadiBF16}. 

There exists $\widetilde{m}, \widetilde{t}, M, Q'$ and $Q_2$ such that $Q[\sigma\rho] \equiv (\nu \widetilde{m}\widetilde{t}. (\overline{N_1}\langle M \rangle. Q' | Q_2))[\sigma\rho]$ and $F \equiv (\nu \widetilde{m}. (\{M/x\} | \nu \widetilde{t}.(Q' | Q_2)))[\rho \sigma]$ where $\widetilde{m} \subseteq \mathrm{n}(M) \land \widetilde{t} \cap \mathrm{n}(M) = \emptyset \land \mathrm{n}(M) \cap \widetilde{s} = \emptyset$ and $M[\sigma\rho]$ is ground.

Hence $C_m \sigma \overset{\nu x. \overline{N_2}\langle x \rangle}{\longrightarrow} (\rho | \nu \widetilde{m}. (\{M/x\} | \nu \widetilde{t}. (Q' | Q_2)))\sigma \equiv \nu \widetilde{m}. (\rho \uplus \{M/x\} | \nu \widetilde{t}. (Q' | Q_2))\sigma$.

Therefore,  
\begin{align*}
&\nu \widetilde{r} \widetilde{s}.(\nu \widetilde{x}.A_m \rho | \nu \widetilde{y}. C_m \sigma) \\
\longrightarrow &\nu \widetilde{r}\widetilde{s}. ((\sigma \uplus \rho)_{|dom(\sigma \uplus \rho)\setminus \widetilde{x}\widetilde{y}}| \nu x. ((\nu \widetilde{l}. (P' | P_2))[\rho \sigma]|(\nu \widetilde{m}. (\{M/x\} | \nu \widetilde{t}.(Q' | Q_2)))[\rho \sigma])) \\
\equiv &\nu \widetilde{r} \widetilde{s}. (\nu \widetilde{x}. (\sigma | \nu \widetilde{l}. (P' | P_2))(\rho \uplus \{M/x\}) | \nu \widetilde{y}x. (\rho \uplus \{M/x\} | \nu \widetilde{t}. (Q' | Q_2))\sigma)
\end{align*}
This is a desired form.

\item $P[\rho\sigma] \overset{\nu x. \overline{N} \langle x \rangle}{\longrightarrow} E \land Q[\sigma\rho] \overset{N(x)}{\longrightarrow} F \land R \equiv \nu x. (E|F)$ for some $E, F, x$ and ground $N$.

This case is similar to case \ref{PinQout}.
\end{enumerate}
Second, we consider when $\mu$ is a labelled action $\alpha$. Note that $\widetilde{r}\widetilde{s} \cap \mathrm{n}(\alpha) = \emptyset$.

By  \cite[Lemma B.19]{DBLP:journals/corr/AbadiBF16}, $P[\rho\sigma] | Q[\sigma\rho] \overset{\alpha(\sigma \uplus \rho)}{\longrightarrow} E \land D' \equiv (\sigma \uplus \rho)_{|dom(\sigma \uplus \rho)\setminus \widetilde{x}\widetilde{y}}|E$ for some $E$.

We consider two cases.
\begin{enumerate}
\item\label{Palpha} $P[\rho\sigma] \overset{\alpha(\sigma \uplus \rho)}{\longrightarrow} F \land E \equiv F | Q[\sigma\rho]$ for some $F$.

Then, $A_m \rho \overset{\alpha[\rho\sigma]}{\longrightarrow} [\sigma\rho] | F$ because of $\sigma \uplus \rho = [\sigma\rho] | [\rho\sigma]$ and Lemma \ref{drop-sigma}.

In addition, we consider two cases.
\begin{enumerate}
\item $\alpha = N(M)$

There exists $\widetilde{l}, P'$ and $P_2$ such that $P[\rho\sigma] \equiv (\nu \widetilde{l}. (N[\sigma\rho](x). P' | P_2))[\rho\sigma]$ and $F \equiv (\nu \widetilde{l}. (P'\{M[\sigma\rho]/x\} | P_2))[\rho\sigma]$.

Hence, $A_m \rho \overset{N(M)[\rho\sigma]}{\longrightarrow} (\sigma | \nu \widetilde{l}. (P'\{M[\sigma\rho]/x\} | P_2))\rho$. 

$\nu \widetilde{x}. A_m \rho \overset{N(M)[\rho\sigma]}{\longrightarrow} \nu \widetilde{x}. (\sigma | \nu \widetilde{l}. (P'\{M[\sigma\rho]/x\} | P_2))\rho$. 

By Lemma \ref{drop-sigma}, $\nu \widetilde{x}. A_m \rho | [\rho\sigma] \overset{N(M)}{\longrightarrow} \nu \widetilde{x}. (\sigma | \nu \widetilde{l}. (P'\{M[\sigma\rho]/x\} | P_2))\rho | [\rho\sigma]$. 

$\nu \widetilde{x}. A_m \rho | [\rho\sigma] | Q[\sigma\rho] \overset{N(M)}{\longrightarrow} \nu \widetilde{x}. (\sigma | \nu \widetilde{l}. (P'\{M[\sigma\rho]/x\} | P_2))\rho | [\rho\sigma] | Q[\sigma\rho]$. 

By Lemma \ref{cket},

$\nu \widetilde{x}. A_m \rho | \nu \widetilde{y}. (\rho | Q)\sigma \overset{N(M)}{\longrightarrow} \nu \widetilde{x}. (\sigma | \nu \widetilde{l}. (P'\{M[\sigma\rho]/x\} | P_2))\rho | \nu \widetilde{y}. (\rho | Q)\sigma$.

Therefore,

$\nu \widetilde{r} \widetilde{s}.(\nu \widetilde{x}.A_m \rho | \nu \widetilde{y}. C_m \sigma) \overset{N(M)}{\longrightarrow} \nu \widetilde{r} \widetilde{s}.(\nu \widetilde{x}. (\sigma | \nu \widetilde{l}. (P'\{M[\sigma\rho]/x\} | P_2))\rho | \nu \widetilde{y}. C_m\sigma)$.

This is the desired form.

\item $\alpha = \nu x. \overline{N}\langle x \rangle$

There exists $\widetilde{m}, \widetilde{t}, M, P'$ and $P_2$ such that
\begin{equation*}
P[\rho\sigma] \equiv (\nu \widetilde{m} \widetilde{t}. (\overline{N[\sigma\rho]}\langle M \rangle). P' | P_2))[\rho\sigma] \land F \equiv (\nu \widetilde{m}. (\{M/x\} | \nu \widetilde{t}. (P' | P_2)))[\rho\sigma],
\end{equation*}
where $\widetilde{m} \subseteq \mathrm{n}(M) \land \widetilde{t} \cap \mathrm{n}(M) = \emptyset$ and $M[\rho\sigma]$ is ground.

$[\sigma\rho] | F \equiv (\sigma | (\nu \widetilde{m}. (\{M/x\} | \nu \widetilde{t}. (P' | P_2)))\rho \equiv \nu \widetilde{m}. (\sigma \uplus \{M/x\} | \nu \widetilde{t}. (P' | P_2))\rho$.

Hence, $A_m \rho \overset{(\nu x. \overline{N}\langle x \rangle)[\rho \sigma]}{\longrightarrow} \nu \widetilde{m}. (\sigma \uplus \{M/x\} | \nu \widetilde{t}. (P' | P_2))\rho$.

$\nu \widetilde{x}. A_m \rho \overset{(\nu x. \overline{N}\langle x \rangle)[\rho \sigma]}{\longrightarrow} \nu \widetilde{x}. \nu \widetilde{m}. (\sigma \uplus \{M/x\} | \nu \widetilde{t}. (P' | P_2))\rho$.

By Lemma \ref{drop-sigma}, $\nu \widetilde{x}. A_m \rho | [\rho\sigma] \overset{\nu x. \overline{N}\langle x \rangle}{\longrightarrow} \nu \widetilde{x}. \nu \widetilde{m}. (\sigma \uplus \{M/x\} | \nu \widetilde{t}. (P' | P_2))\rho | [\rho\sigma]$.

$\nu \widetilde{x}. A_m \rho | [\rho\sigma] | Q[\sigma\rho] \overset{\nu x. \overline{N}\langle x \rangle}{\longrightarrow} \nu \widetilde{x}. \nu \widetilde{m}. (\sigma \uplus \{M/x\} | \nu \widetilde{t}. (P' | P_2))\rho | [\rho\sigma] | Q[\sigma\rho]$.

By Lemma \ref{cket}, 

$\nu \widetilde{x}. A_m \rho | \nu \widetilde{y}. (\rho | Q)\sigma \overset{\nu x. \overline{N}\langle x \rangle}{\longrightarrow} \nu \widetilde{x}. \nu \widetilde{m}. (\sigma \uplus \{M/x\} | \nu \widetilde{t}. (P' | P_2))\rho | \nu \widetilde{y}. (\rho | Q)\sigma$.

Therefore,

$\nu \widetilde{r} \widetilde{s}.(\nu \widetilde{x}. A_m \rho | \nu \widetilde{y}. C_m\sigma) \overset{\nu x. \overline{N}\langle x \rangle}{\longrightarrow} \nu \widetilde{r} \widetilde{s}. (\nu \widetilde{x}. \nu \widetilde{m}. (\sigma \uplus \{M/x\} | \nu \widetilde{t}. (P' | P_2))\rho | \nu \widetilde{y}. C_m\sigma)$.

This is the desired form.
\end{enumerate}
\item $Q[\sigma\rho] \overset{\alpha(\sigma \uplus \rho)}{\longrightarrow} F \land E \equiv P[\sigma\rho] | F$ for some $F$.

This case is similar to case \ref{Palpha}.
\end{enumerate}
\end{proof}
\subsection{Extracting a Trace of a Component Process from Concurrent Normal Form}\label{extra}
Given a concurrent normal trace of $A|C$, we construct traces of $A$ and $C$ which are each process in them is of the form $\nu \widetilde{s_m}. A_m \rho_m$ or $\nu \widetilde{r_m}. C_m \sigma_m$. Assuming that we transformed the first $m$-th process.

We omit many subscripts and symmetric cases.

In case \ref{Atau} in Definition \ref{cnf}, we get \fbox{$\nu \widetilde{s_m}. A_m \rho_m \longrightarrow \nu \widetilde{s_{m+1}}. A_{m+1}\rho_{m+1}$}.

In case \ref{AinCout}, 
\begin{align*}
A_m \rho \overset{N_1(x)}{\longrightarrow}& A_{m+1}\rho. \\
A_m \rho \overset{N_1(M)}{\longrightarrow}& A_{m+1}\{M/x\}\rho \\
[\sigma\rho] | P_m[\rho\sigma] \overset{N_1(M)}{\longrightarrow}& [\sigma\rho] | P_{m+1}[(\rho \uplus \{M/x\})\sigma]. \\
[\rho\sigma] | [\sigma\rho] | P_m[\rho\sigma] \overset{N_1(M)}{\longrightarrow}& [\rho\sigma] | [\sigma\rho] | P_{m+1}[(\rho \uplus \{M/x\})\sigma]. \\
\end{align*} The left-hand side is normal, so by Lemma \ref{change-label} and \ref{sub-ex}, 


$[\rho\sigma] | [\sigma\rho] | P_m[\rho\sigma] \overset{N_2\rho(M)}{\longrightarrow} [\rho\sigma] | [\sigma\rho] | P_{m+1}[(\rho \uplus \{M/x\})\sigma]$.

Note that $N_1[\sigma\rho] = N_2[\rho\sigma]$ and $M[\sigma\rho]$ is ground.

By Lemma \ref{erase-sigma}, $[\sigma\rho] |  P_m[\rho\sigma] \overset{N_2\rho(M)}{\longrightarrow} [\sigma\rho] | P_{m+1}[(\rho \uplus \{M/x\})\sigma]$.

$A_m \rho \overset{N_2\rho (M)}{\longrightarrow} A_{m+1}(\rho \uplus \{M/x\})$.

We get \fbox{$\nu \widetilde{s_m}. A_m\rho_m \overset{N_2\rho_m (M\rho_m)}{\longrightarrow} \nu \widetilde{s_{m+1}}. A_{m+1}\rho_{m+1}$}.

$C_m \sigma \overset{\nu x. \overline{N_2} \langle x \rangle}{\longrightarrow} \nu \widetilde{m}. C_{m+1}\sigma$.

$[\rho\sigma] | Q_m[\sigma\rho] \overset{\nu x. \overline{N_2} \langle x \rangle}{\longrightarrow} \nu \widetilde{m}. ([(\rho \uplus \{M/x\} )\sigma] | Q_{m+1} [\sigma\rho])$.

$[\sigma\rho] | [\rho\sigma] | Q_m[\sigma\rho] \overset{\nu x. \overline{N_2} \langle x \rangle}{\longrightarrow} [\sigma\rho] | \nu \widetilde{m}. ([(\rho \uplus \{M/x\} )\sigma] | Q_{m+1} [\sigma\rho])$.

The left-hand side is normal, so by Lemma \ref{change-label} and \ref{sub-ex}, 


$[\sigma\rho] | [\rho\sigma] | Q_m[\sigma\rho] \overset{\nu x. \overline{N_1\sigma} \langle x \rangle}{\longrightarrow} [\sigma\rho] | \nu \widetilde{m}. ([(\rho \uplus \{M/x\} )\sigma] | Q_{m+1} [\sigma\rho])$.

By Lemma \ref{erase-sigma}, $[\rho\sigma] | Q_m[\sigma\rho] \overset{\nu x. \overline{N_1\sigma} \langle x \rangle}{\longrightarrow} \nu \widetilde{m}. ([(\rho \uplus \{M/x\} )\sigma] | Q_{m+1} [\sigma\rho])$.

$C_m\sigma \overset{\nu x. \overline{N_1\sigma} \langle x \rangle}{\longrightarrow} \nu \widetilde{m}. C_{m+1}\sigma$.

We get \fbox{$\nu \widetilde{r_m}. C_m\sigma_m \overset{\nu x. \overline{N_1\sigma_m} \langle x \rangle}{\longrightarrow} \nu \widetilde{r_{m+1}}. C_{m+1}\sigma_{m+1}$}.

In case \ref{Ain}, $A_m \rho \overset{N(M)[\rho\sigma]}{\longrightarrow} A_{m+1}\rho$.

$[\sigma\rho] | P[\rho\sigma] \overset{N(M)[\rho\sigma]}{\longrightarrow} A_{m+1}\rho$.


By Lemma \ref{change-label}, $[\sigma\rho] | P[\rho\sigma] \overset{N(M)\rho}{\longrightarrow} A_{m+1}\rho$. Note that $(N(M)[\rho\sigma])[\sigma\rho] = (N(M)\rho)[\sigma\rho]$.

$A_m\rho \overset{N(M)\rho}{\longrightarrow} A_{m+1}\rho$.

We get \fbox{$\nu \widetilde{s_m}. A_m\rho_m \overset{N(M)\rho_m}{\longrightarrow} \nu \widetilde{s_{m+1}}. A_{m+1}\rho_{m+1}$}.

Case \ref{Aout} is similar to case \ref{Ain}. We get \fbox{$\nu \widetilde{s_m}. A_m\rho_m \overset{\nu x. \overline{N}\langle x \rangle \rho_m}{\longrightarrow} \nu \widetilde{s_{m+1}}. A_{m+1}\rho_{m+1}$}.
\subsection{Composition of Traces}
Given trace-equivalent processes $A, B$, an arbitrary process $C$, and a trace of $A|C$, we construct a statically equivalent trace of $B|C$. By subsection \ref{trans}, we can suppose that a given trace \textbf{tr} of $A|C$ is a concurrent normal form without loss of generality. We extract traces $\mathbf{tr}'$ and $\mathbf{tr}'''$ of $A$ and $C$ as mentioned in subsection \ref{extra}. There exists a safe trace $\mathbf{tr}''$ of $B$ such that it satisfies the conditions below:

Each transition $\nu \widetilde{s_m}. A_m \rho_m \overset{\mu}{\longrightarrow} \nu \widetilde{s_{m+1}}. A_{m+1}\rho_{m+1}$ in $\mathbf{tr}'$ corresponds to the transition $\nu \widetilde{s_m'}. B_m \overset{\mu}{\Longrightarrow} \nu \widetilde{s_{m+1}'}. B_{m+1}$ in $\mathbf{tr}''$. Moreover, $B_m$ is of the form $\zeta_m | S_m$, it is normal and $\mathrm{dom}(\zeta_m) = \mathrm{dom}(\sigma_m)$. Especially, $S_m$ is closed.

Each process in a constructed trace is of the form $\nu \widetilde{r_m}\widetilde{s_m'}. (\nu \widetilde{x_m}. B_m | \nu \widetilde{y_m}. C_m\zeta_m)$.

We again omit many subscripts. Note that an internal reduction and a receive action do not change a frame. In addition, $\widetilde{s_m'}$ contains no names in other processes because of bind-exclusiveness. Let $\widetilde{z} = \mathrm{dom}(\sigma)$.
\begin{enumerate}
\item Case \ref{Atau} in Definition \ref{cnf}.
\begin{align}
\nu \widetilde{r}\widetilde{s}.  (\nu \widetilde{x}. A_m\rho | \nu \widetilde{y}. C_m \sigma) \longrightarrow & \nu \widetilde{r}\widetilde{s}. (\nu \widetilde{x}. A_{m+1}\rho | \nu \widetilde{y}. C_m \sigma). \label{3a} \\
\nu \widetilde{s}.  A_m \rho  \longrightarrow & \nu \widetilde{s}. A_{m+1}\rho. \label{Aint3a} \\
\nu \widetilde{s'}. B_m \Longrightarrow & \nu \widetilde{s'}. B_{m+1}. \label{Bint3a}
\end{align}
(\ref{Aint3a}) is in $\mathbf{tr}'$ and corrsponds to (\ref{3a}). (\ref{Bint3a}) is in $\mathbf{tr}''$ and corresponds to (\ref{Aint3a}).

By Lemma \ref{drop-nu}, we can suppose that $B_m \Longrightarrow B_{m+1}$. We get a desired transition
\begin{center}
\fbox{$\nu \widetilde{r}\widetilde{s'}. (\nu \widetilde{x}. B_m | \nu \widetilde{y}. C_m \zeta) \Longrightarrow \nu \widetilde{r}\widetilde{s'}. (\nu \widetilde{x}. B_{m+1} | \nu \widetilde{y}. C_m \zeta)$.}
\end{center}
\item Case \ref{Ctau}.
\begin{align}
\nu \widetilde{r}\widetilde{s}.  (\nu \widetilde{x}. A_m\rho | \nu \widetilde{y}. C_m \sigma) \longrightarrow & \nu \widetilde{r}\widetilde{s}. (\nu \widetilde{x}. A_m\rho | \nu \widetilde{y}. C_{m+1} \sigma). \label{3b} \\
\nu \widetilde{r}. C_m \sigma \longrightarrow & \nu \widetilde{r}. C_{m+1}\sigma \label{Cint3b}
\end{align}
(\ref{Cint3b}) is in $\mathbf{tr}'''$ and corresponds to (\ref{3b}).

By Lemma \ref{drop-nu}, we can suppose that $C_m \sigma \longrightarrow C_{m+1} \sigma$.

$C_m \sigma \equiv \nu \widetilde{z}. (C_m | \sigma) \equiv \nu \widetilde{z}. (\rho | Q | \sigma) \equiv \nu \widetilde{z}. (\rho \uplus \sigma | Q) \equiv \nu \widetilde{z}. (\rho \uplus \sigma | Q')$.

$Q'$ is a process obtained from $Q$ to substitute only for the part related to the transition.

$\rho | Q' \longrightarrow C_{m+1}$.

Similarly, $C_m \zeta \equiv \nu \widetilde{z}. (\rho \uplus \zeta | Q'')$.

$Q''$ is a process obtained from $Q$ to substitute similarly using $\rho \uplus \zeta$.

$\nu \widetilde{s}. [\sigma\rho] \approx_s \nu \widetilde{s'}. \zeta$. Recall that $\nu \widetilde{s}. A_m\rho \approx_s \nu \widetilde{s}'. B_n$

$\nu \widetilde{s}. \rho \uplus \sigma \approx_s \nu \widetilde{s'}. \rho \uplus \zeta$. Note that $\mathrm{n}(\rho) \cap \widetilde{s} = \emptyset$.

Thus, two terms affected by $\rho \uplus \sigma$ at $C_m\sigma \longrightarrow C_{m+1}\sigma$ are also equal in $Q''$.

$\rho | Q'' \longrightarrow C_{m+1}$.

$C_m \zeta \longrightarrow C_{m+1}\zeta$.

We get a desired transition
\begin{center}
\fbox{$\nu \widetilde{r}\widetilde{s'}. (\nu \widetilde{x}. B_m | \nu \widetilde{y}. C_m \zeta) \longrightarrow \nu \widetilde{r}\widetilde{s'}. (\nu \widetilde{x}. B_m | \nu \widetilde{y}. C_{m+1} \zeta)$.}
\end{center}
\item Case \ref{AinCout}.
\begin{align}
\nu \widetilde{r}\widetilde{s}.  (\nu \widetilde{x}. A_m\rho | \nu \widetilde{y}. C_m \sigma) \longrightarrow & \nu \widetilde{r}\widetilde{m}\widetilde{s}. (\nu \widetilde{x}. A_{m+1}(\rho \uplus \{M/x\}) | \nu \widetilde{y}x. C_{m+1} \sigma). \label{3c} \\
\nu \widetilde{s}. A_m \rho \overset{N_2 \rho(M)}{\longrightarrow} & \nu \widetilde{s}. A_{m+1}(\rho \uplus \{M/x\}) \label{Ain3c} \\
\nu \widetilde{r}. C_m \sigma \overset{\nu x. \overline{N_1\sigma}\langle x \rangle}{\longrightarrow} & \nu \widetilde{r}\widetilde{m}. C_{m+1}\sigma \label{Cout3c} \\
\nu \widetilde{s'}. B_m \overset{N_2\rho(M)}{\Longrightarrow} & \nu \widetilde{s'}. B_{m+1} \label{Bin3c}
\end{align}
(\ref{Ain3c}) is in $\mathbf{tr}'$. (\ref{Cout3c}) is in $\mathbf{tr}'''$. Both correspond to (\ref{3c}). (\ref{Bin3c}) is in $\mathbf{tr}''$ and corresponds to (\ref{Ain3c}).

By Lemma \ref{drop-nu}, we can suppose that $C_m \sigma \overset{\nu x. \overline{N_2}\langle x \rangle}{\longrightarrow} \nu \widetilde{m}.C_{m+1}\sigma$.

$C_m \sigma \equiv \nu \widetilde{z}. (\rho \uplus \sigma | Q')$.

$Q'$ is a process obtained from $Q$ to substitute only for the part related to transition.

$\rho | Q' \overset{\nu x. \overline{N}\langle x \rangle}{\longrightarrow} \nu \widetilde{m}.C_{m+1}$, where $N = N_1[\sigma\rho] = N_2[\rho\sigma]$.

Similarly, $C_m\zeta \equiv \nu \widetilde{z}. (\rho \uplus \zeta | Q'')$.

$Q''$ is a process obtained from $Q$ to substitute similarly using $\rho \uplus \zeta$.

$\nu \widetilde{s}. [\sigma\rho] \approx_s \nu \widetilde{s'}. \zeta$.

$\nu \widetilde{s}. \rho \uplus \sigma \approx_s \nu \widetilde{s'}. \rho \uplus \zeta$.

Some channel in $Q$ becomes equal to $N = N_2[\rho\sigma] = N_2(\rho \uplus \sigma)$ by $\rho \uplus \sigma$, so this channel becomes equal to $N_2(\rho \uplus \zeta)$ by $\rho \uplus \zeta$. Note that $\widetilde{s} \cap \mathrm{n}(N_2) = \emptyset$.

Thus, $\rho | Q'' \overset{\nu x. \overline{N_2\rho\zeta}\langle x \rangle}{\longrightarrow} \nu \widetilde{m}. C_{m+1}$. Note that $[\rho\zeta] = \rho\zeta$.

\begin{align*}
\nu \widetilde{r}\widetilde{s'}. (\nu \widetilde{x}. B_m | \nu \widetilde{y}. C_m \zeta) 
\equiv & \nu \widetilde{r}\widetilde{s'}. (\nu \widetilde{x}\widetilde{y}. (\zeta |\rho\zeta | S | Q\zeta))
\end{align*}
$\nu \widetilde{s'}. B \Longrightarrow \nu \widetilde{s'}. B' \overset{N_2\rho(M)}{\longrightarrow} \nu \widetilde{s'}. B'' \Longrightarrow \nu \widetilde{s'}.B_{m+1}$ for some $B'$ and $B''$ because $\nu \widetilde{s'}. B_m \Longrightarrow \nu \widetilde{s'}. B_{m+1}$.

We suppose that 
\begin{itemize}
\item $B' \equiv \zeta | \hat{S}$,
\item $\hat{S} \equiv \nu \widetilde{p}. (N_2\rho(x). S' | S_2)$,
\item $S \Longrightarrow \hat{S} \overset{N_2\rho(M\rho)}{\longrightarrow} D' \Longrightarrow D$,
\item $B_{m+1} = \zeta | D$,
\item $Q\zeta \overset{\nu x. \overline{N_2\rho\zeta}\langle x \rangle}{\longrightarrow} E$, and
\item $\nu \widetilde{m}. C_{m+1}\zeta \equiv \rho\zeta | E$.
\end{itemize}
Let $Q\zeta \equiv (\nu \widetilde{m}\widetilde{q}.(\overline{N_2\rho}\langle M \rangle. Q' | Q_2))\zeta$. This is possible because of $Q\zeta \overset{\nu x. \overline{N_2 \rho\zeta}\langle x \rangle}{\longrightarrow} E$.

By $\nu \widetilde{s'}.B \Longrightarrow \nu \widetilde{s'}. B'$ and $\hat{S} \overset{N_2\rho(M)}{\Longrightarrow} D$,
\begin{align*}
\nu \widetilde{r}\widetilde{s'}. (\nu \widetilde{x}. B_m | \nu \widetilde{y}. C_m \zeta) \Longrightarrow & \nu \widetilde{r}\widetilde{s'}. (\nu \widetilde{x}\widetilde{y}. (\zeta | \rho\zeta | \nu \widetilde{p}. (N_2\rho(x). S' | S_2) | (\nu \widetilde{m} \widetilde{q}. (\overline{N_2\rho} \langle M \rangle. Q' | Q_2))\zeta) \\
\Longrightarrow &  \nu \widetilde{r}\widetilde{s'}. (\nu \widetilde{x}\widetilde{y}. (\zeta | \rho\zeta | \nu x.(D | (\nu \widetilde{m}\widetilde{q}. (Q' | \{M/x\} | Q_2))\zeta))) \\
\equiv & \nu \widetilde{r}\widetilde{m}\widetilde{s'}. (\nu \widetilde{x}. (\zeta | D) | \nu \widetilde{y}x. (\rho \uplus \{M/x\} | \nu \widetilde{q}. (Q' | Q_2)))\zeta).
\end{align*}
We get the desired transition
\begin{center}
\fbox{$\nu \widetilde{r}\widetilde{s'}. (\nu \widetilde{x}. B_m | \nu \widetilde{y}. C_m \zeta) \Longrightarrow \nu \widetilde{r}\widetilde{m}\widetilde{s'}. (\nu \widetilde{x}. B_{m+1} | \nu \widetilde{y}x. C_{m+1}\zeta)$.}
\end{center}
\item Case \ref{AoutCin}.
\begin{align}
\nu \widetilde{r}\widetilde{s}. (\nu \widetilde{x}. A_m\rho | \nu \widetilde{y}. C_m \sigma) \longrightarrow & \nu \widetilde{r}\widetilde{s}\widetilde{m}. (\nu \widetilde{x}x. A_{m+1}\rho | \nu \widetilde{y}. C_{m+1}(\sigma \uplus \{M/x\})) \label{3d} \\
\nu \widetilde{s}. A_m\rho \overset{\nu x. \overline{N_2\rho}\langle x \rangle}{\longrightarrow} & \nu \widetilde{s}\widetilde{m}. A_{m+1}\rho \label{Aout3d} \\
\nu \widetilde{r}. C_m\sigma \overset{N_1\sigma(M)}{\longrightarrow} & \nu \widetilde{r}. C_{m+1}(\sigma \uplus \{M/x\}) \label{Cin3d} \\
\nu \widetilde{s'}. B_m \overset{\nu x. \overline{N_2\rho} \langle x \rangle}{\Longrightarrow} & \nu \widetilde{s'}\widetilde{m'}. B_{m+1} \label{Bout3d}
\end{align}
(\ref{Aout3d}) is in $\mathbf{tr}'$. (\ref{Cin3d}) is in $\mathbf{tr}'''$. Both correspond to (\ref{3d}). (\ref{Bout3d}) is in $\mathbf{tr}''$ and corresponds to (\ref{Aout3d}).

$C_m\sigma \overset{N_2(x)}{\longrightarrow} C_{m+1}\sigma$.

$C_m\sigma \equiv \nu \widetilde{z}. (\rho \uplus \sigma | Q')$, where $Q'$ is a process obtained from $Q$ to substitute only for the input channel.

$\rho | Q' \overset{N(x)}{\longrightarrow} C_{m+1}$, where $N=N_1[\sigma\rho]=N_2[\rho\sigma]$.

Similarly, $C_m\zeta \equiv \nu \widetilde{z}. (\rho \uplus \zeta | Q'')$.

$Q''$ is a process obtained from $Q$ to substitute similarly using $\rho \uplus \zeta$.

$\nu \widetilde{s}. (\rho \uplus \sigma) \approx_s \nu \widetilde{s'}. (\rho \uplus \zeta)$.

Some channel in $Q$ is equal to $N=N_2[\rho\sigma]=N_2(\rho\uplus\sigma)$ by $\rho \uplus \sigma$, so this channel becomes equal to $N_2(\rho \uplus \zeta)$ by $\rho \uplus \zeta$. Note that $\widetilde{s} \cap \mathrm{n}(N_2) = \emptyset$.

Thus, $\rho | Q'' \overset{N_2\rho\zeta(x)}{\longrightarrow} C_{m+1}$. Note that $[\rho\zeta] = \rho\zeta$.

$\nu \widetilde{s'}. B \Longrightarrow \nu \widetilde{s'}. B' \overset{\nu x. \overline{N_2\rho}\langle x \rangle}{\longrightarrow} \nu \widetilde{s''}. B'' \Longrightarrow \nu \widetilde{s'}.B_{m+1}$ for some $B'$ and $B''$.

We suppose that 
\begin{itemize}
\item $B' \equiv \zeta | \hat{S}$,
\item $\hat{S} \equiv \nu \widetilde{m'}\widetilde{p}.(\overline{N_2\rho}\langle M \rangle. S' | S_2)$,
\item $S \Longrightarrow \hat{S} \overset{\nu x. \overline{N_2\rho}\langle x \rangle}{\longrightarrow} D' \Longrightarrow D$,
\item $B_{m+1} = \zeta | D$,
\item $Q\zeta \overset{N_2\rho\zeta(x)}{\longrightarrow} E$, and
\item $C_{m+1}\zeta \equiv \rho\zeta | E$.
\end{itemize}
Let $Q\zeta \equiv (\nu \widetilde{q}.(N_2\rho(x). Q' | Q_2))\zeta$.

By $S \Longrightarrow \hat{S}$,
\begin{align*}
&\nu \widetilde{r} \widetilde{s'}. (\nu \widetilde{x}. B_m | \nu \widetilde{y}. C_m\zeta) \\
\equiv &  \nu \widetilde{r}\widetilde{s'}. (\nu \widetilde{x}\widetilde{y}. (\zeta | \rho\zeta | S | Q\zeta)) \\
\Longrightarrow &  \nu \widetilde{r}\widetilde{s'}. (\nu \widetilde{x} \widetilde{y}. (\zeta | \rho\zeta | \nu \widetilde{m'}\widetilde{p}. (\overline{N_2\rho}\langle M \rangle. S' | S_2) | (\nu \widetilde{q}. (N_2\rho(x). Q' | Q_2))\zeta) \\
\longrightarrow &  \nu \widetilde{r}\widetilde{s'}. (\nu \widetilde{x} \widetilde{y}. (\zeta | \rho\zeta | \nu x. (\nu \widetilde{m'}. (\{M/x\} | \nu \widetilde{p}. (S' | S_2)) | (\nu \widetilde{q}. Q' | Q_2))\zeta) \\
\equiv &  \nu \widetilde{r}\widetilde{s'}\widetilde{m'}. (\nu \widetilde{x}x. (\zeta \uplus \{M/x\} | \nu \widetilde{p}. (S' | S_2)) | \nu \widetilde{y}. (\rho | \nu \widetilde{q}. (Q' | Q_2))(\zeta \uplus \{M/x\})) \\
\Longrightarrow &  \nu \widetilde{r}\widetilde{s'}\widetilde{m'}. (\nu\widetilde{x}x. B_{m+1} | \nu\widetilde{y}. C_{m+1}(\zeta \uplus \{M/x\}))
\end{align*}
We get a desired transition
\begin{center}
\fbox{$\nu \widetilde{r}\widetilde{s'}. (\nu \widetilde{x}. B_m | \nu \widetilde{y}. C_m \zeta) \Longrightarrow \nu \widetilde{r}\widetilde{s'}\widetilde{m'}. (\nu \widetilde{x}x. B_{m+1} | \nu \widetilde{y}. C_{m+1}(\zeta \uplus \{M/x\}))$.}
\end{center}
\item Case \ref{Ain}.
\begin{align}
\nu \widetilde{r}\widetilde{s}. (\nu \widetilde{x}. A_m\rho | \nu \widetilde{y}. C_m \sigma) \overset{N(M)}{\longrightarrow} & \nu \widetilde{r}\widetilde{s}. (\nu \widetilde{x}. A_{m+1}\rho | \nu \widetilde{y}. C_m\sigma) \label{4a} \\
\nu \widetilde{s}. A_m\rho \overset{N(M)\rho}{\longrightarrow} & \nu \widetilde{s}. A_{m+1}\rho \label{Ain4a} \\
\nu \widetilde{s'}. B_m \overset{N(M)\rho}{\Longrightarrow} & \nu \widetilde{s'}. B_{m+1} \label{Bin4a}
\end{align}
(\ref{Ain4a}) is in $\mathbf{tr}'$ and corresponds to (\ref{4a}). Recall subsection \ref{extra}. (\ref{Bin4a}) is in $\mathbf{tr}'$ and corresponds to (\ref{Ain4a}).

By Lemma \ref{drop-nu}, we can suppose that $B_m \overset{N(M)\rho}{\Longrightarrow} B_{m+1}$.

By Lemma \ref{drop-sigma}, $\rho | B_m \overset{N(M)}{\Longrightarrow} \rho | B_{m+1}$.

We get a desired transition
\begin{center}
\fbox{$\nu \widetilde{r}\widetilde{s'}. (\nu \widetilde{x}. B_m | \nu \widetilde{y}. C_m \zeta) \overset{N(M)}{\Longrightarrow} \nu \widetilde{r}\widetilde{s'}. (\nu \widetilde{x}. B_{m+1} | \nu \widetilde{y}. C_m\zeta)$.}
\end{center}
\item Case \ref{Cin}.
\begin{align}
\nu \widetilde{r}\widetilde{s}. (\nu \widetilde{x}. A_m\rho | \nu \widetilde{y}. C_m \sigma) \overset{N(M)}{\longrightarrow} & \nu \widetilde{r}\widetilde{s}. (\nu \widetilde{x}. A_m\rho | \nu \widetilde{y}. C_{m+1}\sigma) \label{4b} \\
\nu \widetilde{r}. C_m\sigma \overset{N(M)\sigma}{\longrightarrow} & \nu \widetilde{r}.C_{m+1}\sigma \label{Cin4c}
\end{align}
(\ref{Cin4c}) is in $\mathbf{tr}'''$ and corresponds to (\ref{4b}).

$C_m\sigma \overset{N(M)\sigma}{\longrightarrow} C_{m+1}\sigma$.

$C_m\sigma \equiv \nu \widetilde{z}. (\rho \uplus \sigma | Q')$ where $Q'$ is a process obtained from $Q$ to substitute only for the input channel.

$\rho | Q' \overset{N(\rho \uplus \sigma)(M)}{\longrightarrow} C_{m+1}$.

Similarly, $C_m\zeta \equiv \nu \widetilde{z}. (\rho \uplus \zeta | Q'')$.

$Q''$ is a process obtained from $Q$ to substitute similarly using $\rho \uplus \zeta$.

Some channel in $Q$ is equal to $N(\rho \uplus \zeta)$ by $\rho \uplus \zeta$ because $\nu \widetilde{s}. (\rho \uplus \sigma) \approx_s \nu \widetilde{s'}. (\rho \uplus \zeta)$. Note that $\widetilde{s} \cap \mathrm{n}(N_2) = \emptyset$.

Thus, $\rho | Q'' \overset{N(\rho \uplus \zeta)(M)}{\longrightarrow} C_{m+1}$.

$\rho \uplus \zeta | Q'' \overset{N(\rho \uplus \zeta)(M)}{\longrightarrow} \zeta | C_{m+1}$. 

By Lemma \ref{change-label}, $\rho \uplus \zeta | Q'' \overset{N(M)\zeta}{\longrightarrow} \zeta | C_{m+1}$. 

$C_m\zeta \overset{N(M)\zeta}{\longrightarrow} C_{m+1}\zeta$.

By Lemma \ref{drop-sigma}, $\zeta | C_m\zeta \overset{N(M)}{\longrightarrow} \zeta | C_{m+1}\zeta$. We get a desired transition
\begin{center}
\fbox{$\nu \widetilde{r}\widetilde{s'}. (\nu \widetilde{x}. B_m | \nu \widetilde{y}. C_m \zeta) \overset{N(M)}{\longrightarrow} \nu \widetilde{r}\widetilde{s'}. (\nu \widetilde{x}. B_m | \nu \widetilde{y}. C_{m+1}\zeta)$.}
\end{center}
\item Case \ref{Aout}.
\begin{align}
\nu \widetilde{r}\widetilde{s}. (\nu \widetilde{x}. A_m\rho | \nu \widetilde{y}. C_m \sigma) \overset{\nu x. \overline{N}\langle x \rangle}{\longrightarrow} & \nu \widetilde{r}\widetilde{s}\widetilde{m}. (\nu \widetilde{x}. A_{m+1}\rho | \nu \widetilde{y}. C_m(\sigma \uplus \{M/x\})) \label{5a} \\
\nu \widetilde{s}. A_m\rho \overset{\nu x. \overline{N}\langle x \rangle \rho}{\longrightarrow} & \nu \widetilde{s}\widetilde{m}.A_{m+1}\rho \label{Aout5a} \\
\nu \widetilde{s'}. B_m \overset{\nu x. \overline{N}\langle x \rangle \rho}{\Longrightarrow} & \nu \widetilde{s'}\widetilde{m'}. B_{m+1} \label{Bout5a} \\
\nu \widetilde{r}\widetilde{s'}. (\nu \widetilde{x}. B_m | \nu \widetilde{y}. C_m \zeta) \overset{\nu x. \overline{N}\langle x \rangle \rho}{\Longrightarrow} & \nu \widetilde{r}\widetilde{s'}\widetilde{m'}. (\nu \widetilde{x}. B_{m+1} | \nu \widetilde{y}. C_m(\zeta \uplus \{M/x\})) \label{desired5a}
\end{align}
(\ref{Aout5a}) is in $\mathbf{tr}'$ and corresponds to (\ref{5a}). Recall subsection \ref{extra}. (\ref{Bout5a}) is in $\mathbf{tr}''$ and corresponds to (\ref{Aout5a}). (\ref{desired5a}) is deriverd from (\ref{Bout5a}). Note that $x$ does not appear in $C_m$.

(\ref{desired5a}) contains an output, so it changes a frame. We have to check that static equivalence is preserved.

$\nu \widetilde{s}.[\sigma\rho] \approx_s \nu \widetilde{s'}. \zeta$ because $\nu \widetilde{s}. A_m\rho \approx_s \nu \widetilde{s'}. B_m$.

$\nu \widetilde{s}\widetilde{m}. [(\sigma \uplus \{M/x\})\rho] \approx_s \nu \widetilde{s'}\widetilde{m'}. (\zeta \uplus \{M'/x\})$ because $\nu \widetilde{s}\widetilde{m}.A_{m+1}\rho \approx_s \nu \widetilde{s'}\widetilde{m'}. B_{m+1}$.

In addition, $\nu \widetilde{r}\widetilde{s}. (\nu \widetilde{x}.[\sigma\rho] | \nu \widetilde{y}. [\rho\sigma]) \approx_s \nu \widetilde{r}\widetilde{s'}. (\nu \widetilde{x}. \zeta | \nu \widetilde{y}. \rho\zeta)$. 

This is because $\nu \widetilde{r}\widetilde{s}. (\nu \widetilde{x}. A_m\rho | \nu \widetilde{y}. C_m \sigma) \approx_s \nu \widetilde{r}\widetilde{s'}. (\nu \widetilde{x}. B_m | \nu \widetilde{y}. C_m \zeta)$.

$P[\rho\sigma] \overset{\nu x. \overline{(N\rho)[\sigma\rho]}\langle x \rangle}{\longrightarrow} D \land \nu \widetilde{m}. A_{m+1}\rho \equiv [\sigma\rho] | D$ for some $D$. By Lemma \ref{shift-sigma}. Recall Definition \ref{cnf}.

Let $P[\rho\sigma] \equiv (\nu \widetilde{m}\widetilde{l}. (\overline{N}\langle M \rangle. P' | P_2))[\rho\sigma] \land D \equiv (\nu \widetilde{m}. (\{M/x\} | \nu \widetilde{l}. (P' | P_2)))[\rho\sigma]$.

Then, $A_{m+1}\rho \equiv (\sigma \uplus \{M/x\} | \nu \widetilde{l}. (P' | P_2))\rho$.
\begin{align*}
&\nu \widetilde{r}\widetilde{s}\widetilde{m}. (\nu \widetilde{x}. [(\sigma \uplus \{M/x\})\rho] | \nu \widetilde{y}. [\rho(\sigma \uplus \{M/x\})]) \\
\equiv & \nu \widetilde{r}. (\nu \widetilde{x}. \nu \widetilde{s}\widetilde{m}. [(\sigma \uplus \{M/x\})\rho] | \nu \widetilde{y}. \rho) \\
\approx_s & \nu \widetilde{r}. (\nu \widetilde{x}. \nu \widetilde{s'}\widetilde{m'}. (\zeta \uplus \{M'/x\}) | \nu \widetilde{y}. \rho) \\
\equiv & \nu \widetilde{r}\widetilde{s'}\widetilde{m'}. (\nu \widetilde{x}. (\zeta \uplus \{M'/x\}) | \nu \widetilde{y}. \rho (\zeta \uplus \{M'/x\}))
\end{align*}
Note that $\nu \widetilde{s}\widetilde{m}. [(\sigma \uplus \{M/x\})\rho] \approx_s \nu \widetilde{s'}\widetilde{m'}. (\zeta \uplus \{M'/x\})$.

Therefore, 

$\nu \widetilde{r}\widetilde{s}\widetilde{m}. (\nu \widetilde{x}. A_{m+1}\rho | \nu \widetilde{y}. C_m(\sigma \uplus \{M/x\})) \approx_s \nu \widetilde{r}\widetilde{s'}\widetilde{m'}. (\nu \widetilde{x}. B_{m+1}\rho | \nu \widetilde{y}. C_m(\zeta \uplus \{M'/x\}))$.

Hence, the transition below is a desired one.
\begin{center}
\fbox{$\nu \widetilde{r}\widetilde{s'}. (\nu \widetilde{x}. B_m | \nu \widetilde{y}. C_m \zeta) \overset{\nu x. \overline{N}\langle x \rangle}{\Longrightarrow} \nu \widetilde{r}\widetilde{s'}\widetilde{m'}. (\nu \widetilde{x}. B_{m+1}\rho | \nu \widetilde{y}. C_m(\zeta \uplus \{M/x\}))$.}
\end{center}
\item Case \ref{Cout}.
\begin{align}
\nu \widetilde{r}\widetilde{s}. (\nu \widetilde{x}. A_m\rho | \nu \widetilde{y}. C_m \sigma) \overset{\nu x. \overline{N}\langle x \rangle}{\longrightarrow} & \nu \widetilde{r}\widetilde{m}\widetilde{s}. (\nu \widetilde{x}. A_m(\rho \uplus \{M/x\}) | \nu \widetilde{y}. C_{m+1}\sigma) \label{5b} \\
\nu \widetilde{r}. C_m\sigma \overset{\nu x. \overline{N}\langle x \rangle \sigma}{\longrightarrow} & \nu \widetilde{r}\widetilde{m}. C_{m+1}\sigma \label{Cin5b}
\end{align}
(\ref{Cin5b}) is in $\mathbf{tr}'''$ and corresponds to (\ref{5b}).

In the same way with Case \ref{Cin}, $C_m\zeta \overset{\nu x. \overline{N\zeta}\langle x \rangle}{\longrightarrow} C_{m+1}\zeta$.

$\nu \widetilde{r}\widetilde{s'}. (\nu \widetilde{x}. B_m | \nu \widetilde{y}. C_m\zeta) \overset{\nu x. \overline{N}\langle x \rangle}{\longrightarrow} \nu \widetilde{r}\widetilde{m}\widetilde{s'}. (\nu \widetilde{x}. B_m | \nu \widetilde{y}. C_{m+1}\zeta)$.

We have to check that static equivalence is preserved.

Now, $\nu \widetilde{r}\widetilde{s}. (\nu \widetilde{x}. [\sigma\rho] | \nu \widetilde{y}. [\rho\sigma]) \approx_s \nu \widetilde{r}\widetilde{s'}. (\nu \widetilde{x}. \zeta | \nu \widetilde{y}. \rho \zeta)$.

This is because $\nu \widetilde{r}\widetilde{s}. (\nu \widetilde{x}. A_m\rho | \nu \widetilde{y}. C_m \sigma) \approx_s \nu \widetilde{r}\widetilde{s'}. (\nu \widetilde{x}. B_m | \nu \widetilde{y}. C_m \zeta)$.

$\nu \widetilde{s}.[\sigma\rho] \approx_s \nu \widetilde{s'}. \zeta$ because $\nu \widetilde{s}. A_m\rho \approx_s \nu \widetilde{s'}. B_m$.

$\nu \widetilde{s}. [\sigma\rho] | \nu \widetilde{y}. (\rho \uplus \{M/x\}) \approx_s \nu \widetilde{s'}. \zeta | \nu \widetilde{y}. (\rho \uplus \{M/x\})$.

Thus, $\nu \widetilde{r}\widetilde{m}\widetilde{s}. (\nu \widetilde{x}. [\sigma\rho] | \nu \widetilde{y}. (\rho \uplus \{M/x\})) \approx_s \nu \widetilde{r}\widetilde{m}\widetilde{s'}. (\nu \widetilde{x}. \zeta | \nu \widetilde{y}. (\sigma \uplus \{M/x\}))$.

Hence, the transition below is a desired one.
\begin{center}
\fbox{$\nu \widetilde{r}\widetilde{s'}. (\nu \widetilde{x}. B_m | \nu \widetilde{y}. C_m\zeta) \overset{\nu x. \overline{N}\langle x \rangle}{\longrightarrow} \nu \widetilde{r}\widetilde{m}\widetilde{s'}. (\nu \widetilde{x}. B_m | \nu \widetilde{y}. C_{m+1}\zeta)$.}
\end{center}
\end{enumerate}
Thus we have finished the proof of Lemma \ref{para} and Theorem \ref{main-cong}.
\chapter{An Epistemic Logic for the Applied Pi Calculus\label{epi}}
From now, we suppose that all processes are name-variable-distinct.  We redefine trace equivalence as below.
\begin{definition}
\begin{align*}
\mathrm{tr}(A) &= \{\mathbf{tr}|\mathbf{tr}\ \mathrm{is\ a\ trace\ of}\ A\} \\
\mathrm{tr_{max}}(A) &= \{\mathbf{tr} \in \mathrm{tr}(A)|\mathbf{tr}\ \mathrm{is\ maximal}\}
\end{align*}
\end{definition}
\begin{definition}
Let $A$ and $B$ be closed processes.
\begin{align*}
A \subseteq_t B &\overset{\mathrm{def}}{\Leftrightarrow} \forall \mathbf{tr} \in \mathrm{tr}(A) \exists \mathbf{tr}' \in \mathrm{tr}(B)\ s.t.\ \mathbf{tr} \sim_t \mathbf{tr}', \\
A \approx_t B &\overset{\mathrm{def}}{\Leftrightarrow} A \subseteq_t B \land B \subseteq_t A.
\end{align*}
Let A and B be two processes. Let $\sigma$ be a map that maps a variable in $(\mathrm{fv}(A)\setminus \mathrm{dom}(A)) \cup (\mathrm{fv}(B)\setminus \mathrm{dom}(B))$ to a ground term. When $A\sigma \approx_t B\sigma$ for all $\sigma$ and capture-avoidings, we also denote as $A \approx_t B$.
\end{definition}
Please note that this definition is equivalent to the previous definition because of Proposition \ref{std-def}.
\section{Syntax}
We adopt equality of terms as a primitive proposition. We give syntax of formulas.
\begin{align*}
\delta ::=& \top | M_1=M_2 | M\in \mathrm{dom} | \delta_1 \lor \delta_2 | \lnot \delta & \mathrm{Static\ formula} \\
\varphi ::=& \delta | \varphi_1 \lor \varphi_2 | \lnot \varphi | \langle \mu \rangle_- \varphi | F\varphi | K\varphi & \mathrm{Modal\ formula}
\end{align*}
where $M_1, M_2$ and $M$ are terms, and $\mu$ is an action. 
\section{Semantics}
%

Our logic is an LTL-like logic with an epistemic operator.
Let $A$ be a name-variable-distinct extended process that $\mathrm{fv}(A)\setminus \mathrm{dom}(A) = \widetilde{x}$, $\rho$ be an assignment from $\widetilde{x}$ to ground terms, \textbf{tr} be a trace of $A\rho$ and $0 \leq i \leq |\mathbf{tr}|, M_1$ and $M_2$ be terms.

We suppose that $\delta$ and $\varphi$ contain no variables other than $\widetilde{x} \cup \mathrm{dom}(\mathbf{tr}[i])$.
\begin{align*}
A, \rho, \mathbf{tr}, i \models \top&\ \mathrm{always\ holds.} \\
A, \rho, \mathbf{tr}, i \models M_1=M_2\ \mathrm{iff}&\ (M_1\rho=M_2\rho)\mathrm{fr}(\mathbf{tr}[i]) \\
A, \rho, \mathbf{tr}, i \models x \in \mathrm{dom} \ \mathrm{iff}&\ x \in \mathrm{dom}(\mathbf{tr}[i]) \\
A, \rho, \mathbf{tr}, i \models M \in \mathrm{dom}& \ \mathrm{never\ hold\ when}\ M\ \mathrm{is\ not\ a\ variable.} \\
A, \rho, \mathbf{tr}, i \models \delta_1 \lor \delta_2 \ \mathrm{iff}&\ A, \rho, \mathbf{tr}, i \models \delta_1 \lor A, \rho, \mathbf{tr}, i \models \delta_2 \\
A, \rho, \mathbf{tr}, i \models \lnot \delta \ \mathrm{iff}&\ A, \rho, \mathbf{tr}, i \not\models \delta \\
A, \rho, \mathbf{tr}, i \models \varphi_1 \lor \varphi_2 \ \mathrm{iff}&\ A, \rho, \mathbf{tr}, i \models \varphi_1 \lor A, \rho, \mathbf{tr}, i \models \varphi_2 \\
A, \rho, \mathbf{tr}, i \models \lnot \varphi \ \mathrm{iff}&\ A, \rho, \mathbf{tr}, i \not\models \varphi \\
A, \rho, \mathbf{tr}, i \models \langle \mu \rangle_- \varphi\ \mathrm{iff}&\ \mathbf{tr}[i-1] \overset{\mu}{\Longrightarrow} \mathbf{tr}[i]\ \mathrm{in}\ \mathbf{tr}\ \land  A, \rho, \mathbf{tr}, i-1 \models \varphi \\
A, \rho, \mathbf{tr}, i \models F\varphi\ \mathrm{iff}&\ \exists j \geq i\ \mathrm{s.t.}\ A, \rho, \mathbf{tr}, j \models \varphi \\
A, \rho, \mathbf{tr}, i \models K\varphi\ \mathrm{iff}&\ \forall \rho'\ \forall \mathbf{tr}'\in \mathrm{tr}(A\rho'); \mathbf{tr}[0, i] \sim_t \mathbf{tr}'[0, i] \Rightarrow A, \rho', \mathbf{tr}', i \models \varphi
\end{align*}
We also define validity of formulas containing free variables. 

Let $\mathrm{dom}(\mathbf{tr}[i]) = \widetilde{y}$. We suppose that $\varphi$ contains no variables other than $\widetilde{x}, \widetilde{y}$ and $\widetilde{z}$.
\begin{equation*}
A, \rho, \mathbf{tr}, i \models \varphi(\widetilde{x}, \widetilde{y}, \widetilde{z})\ \mathrm{iff}\ \forall \widetilde{M}: \mathrm{closed}; A, \rho, \mathbf{tr}, i \models \varphi(\widetilde{x}, \widetilde{y}, \widetilde{M}).
\end{equation*}

\begin{definition} 
$A \models \varphi$ if and only if $\forall \rho\ \forall \mathbf{tr} \in \mathrm{tr}(A\rho); A, \rho, \mathbf{tr}, 0 \models \varphi$.
\end{definition}
\begin{definition} 
$A \sqsubseteq_s B$ if and only if $\forall \delta\ \forall \rho; A, \rho, A\rho, 0\models \delta \Rightarrow B, \rho, B\rho, 0 \models \delta$. 

$A \equiv_s B$ if and only if $A \sqsubseteq_s B \land B \sqsubseteq_s A$.
\end{definition}
\begin{definition} 
$A \sqsubseteq_L B \overset{\mathrm{def}}{\Leftrightarrow} \forall \rho \forall \mathbf{tr} \in \mathrm{tr}(A\rho) \exists \mathbf{tr}' \in \mathrm{tr}(B\rho)$

$\mathrm{s.t.}\ \forall i \forall \varphi; [A, \rho, \mathbf{tr}, i \models \varphi \Leftrightarrow B, \rho, \mathbf{tr}', i \models \varphi]$.


$A \equiv_L B$ if and only if $A \sqsubseteq_L B \land B \sqsubseteq_L A$.
%
\end{definition}
\begin{prop}
We assume that $A \sqsubseteq_L B$. Then $B \models \varphi \Rightarrow A \models \varphi$.
\begin{proof}
Assuming that $B \models \varphi$, we arbitrarily take $\rho$ and $\mathbf{tr} \in \mathrm{tr}(A\rho)$.

By assumption, there exists $\mathbf{tr}' \in \mathrm{tr}(B\rho)$ such that 
\begin{equation}\label{upper}
\forall i \forall \varphi; [A, \rho, \mathbf{tr}, i \models \varphi \Leftrightarrow B, \rho, \mathbf{tr}', i \models \varphi].
\end{equation}
Then, $B, \rho, \mathbf{tr}', 0 \models \varphi$ because of $B \models \varphi$.

By (\ref{upper}), $A, \rho, \mathbf{tr}, 0 \models \varphi$.

By arbitrariness of $\rho$ and $\mathbf{tr}$, it holds that $A \models \varphi$.
\end{proof}
\end{prop}
\section{Correspondence with Trace Equivalence}
We prove that trace equivalent processes satisfy the same formulas and vice versa.
\begin{lemma}\label{sta-log} 
$\mathrm{[}\forall \rho; A\rho \approx_s B\rho\mathrm{]}$ $\Leftrightarrow A \equiv_s B$ 
\end{lemma}
\begin{proof}
$\Rightarrow$) We prove $A \sqsubseteq_s B$ by induction on the syntax of static formulas.

Let $\widetilde{x} = (\mathrm{fv}(A)\setminus \mathrm{dom}(A)) \cup (\mathrm{fv}(B)\setminus \mathrm{dom}(B))$.

We arbitrarily take an assignment $\rho$ from $\widetilde{x}$ to terms. We suppose $A, \rho, A\rho, 0\models \delta$ and $\delta$ contains no variables other than free variables of $A$. Here, we omitted restriction of a domain.
\begin{enumerate}
\item $\top$.

Trivially, $B, \rho, B\rho, 0\models \top$.

%
%
\item $M_1 = M_2$.

By assumption, $(M_1\rho = M_2\rho)\mathrm{fr}(A\rho)$.

By $A\rho \approx_s B\rho$, $(M_1\rho = M_2\rho)\mathrm{fr}(B\rho)$.

This means that $B, \rho, B\rho, 0\models M_1 = M_2$.

\item $M \in \mathrm{dom}$.

$M$ must be a variable $x$.

By definition, $x \in \mathrm{dom}(A\rho)$.

By $A\rho \approx_s B\rho$, $x \in \mathrm{dom}(B\rho)$.

This means that $B, \rho, B\rho, 0\models x \in \mathrm{dom}$.

%
%
%
%
\item $\delta_1 \lor \delta_2$

By assumption, $A, \rho, A\rho, 0\models \delta_1 \lor \delta_2$.

By definition, $A, \rho, A\rho, 0\models \delta_1$ or $A, \rho, A\rho, 0\models \delta_2$.

By induction hypothesis, $B, \rho, B\rho, 0\models \delta_1$ or $B, \rho, B\rho, 0\models \delta_2$.

This means that $B, \rho, B\rho, 0\models \delta_1 \lor \delta_2$.

\item $\lnot \delta$.

By assumption, $A, \rho, A\rho, 0\not\models \delta$.

By induction hypothesis, $B, \rho, B\rho, 0\not\models \delta$.

This means that $B, \rho, B\rho, 0\models \lnot\delta$.
%
%
%
%
\end{enumerate}
Thus, $A \sqsubseteq_s B$ and vice versa.

$\Leftarrow$) We arbitrarily take an assignment $\rho$ from $\widetilde{x}$ to terms.

We arbitrarily take $x \in \mathrm{dom}(A\rho)$. Then, $A, \rho, A\rho, 0\models x \in \mathrm{dom}$.

By assumption, $B, \rho, B\rho, 0\models x \in \mathrm{dom}$. That is, $x \in \mathrm{dom}(B\rho)$.

Thus, $\mathrm{dom}(A\rho) \subseteq \mathrm{dom}(B\rho)$. The converse is similar.
%
%
%
%
%

We suppose that $(M=N)\mathrm{fr}(A\rho)$. This means that $A, \rho, A\rho, 0\models M = N$. By assumption, $B, \rho, B\rho, 0\models M = N$. In other words, $(M=N)\mathrm{fr}(B\rho)$. Similarly, $(M=N)\mathrm{fr}(B\rho) \Rightarrow (M=N)\mathrm{fr}(A\rho)$. Therefore, $A\rho \approx_s B\rho$.
\end{proof}
\begin{theorem}\label{correspondence}
\begin{enumerate}
\item\label{tr-lg} $A \approx_t B \Rightarrow A \sqsubseteq_L B$.
\item\label{lg-tr} $A \sqsubseteq_L B \Rightarrow A \subseteq_t B$.
\end{enumerate}
\end{theorem}
\begin{proof}
Let $\widetilde{x} = (\mathrm{fv}(A)\setminus \mathrm{dom}(A)) \cup (\mathrm{fv}(B)\setminus \mathrm{dom}(B))$. 

\ref{tr-lg}) We prove
\begin{equation}\label{howto}
\forall \rho \forall \mathbf{tr} \in \mathrm{tr}(A\rho) \forall \mathbf{tr}' \in \mathrm{tr}(B\rho); \mathbf{tr} \sim_t \mathbf{tr}' \Rightarrow \forall i \forall \varphi; [A, \rho, \mathbf{tr}, i \models \varphi \Leftrightarrow B, \rho, \mathbf{tr}', i \models \varphi]
\end{equation}
by induction on the syntax of modal formulas. Here, we omitted restriction of a domain.

We take arbitrarily an assignment $\rho$ from $\widetilde{x}$ to terms. We also arbitrarily take traces \textbf{tr} and $\mathbf{tr}'$ of $A\rho$ and $B\rho$.  We arbitrarily take $i$.
\begin{enumerate}
\item $\delta$.

We suppose that $A, \rho, \mathbf{tr}, i \models \delta$.

This means that $\mathbf{tr}[i], 0, \mathbf{tr}[i], 0\models \delta$. Here, the first 0 expresses an empty map.

By $\mathbf{tr} \sim_t \mathbf{tr}'$, $\mathbf{tr}[i] \approx_s \mathbf{tr}'[i]$. By Lemma \ref{sta-log}, $\mathbf{tr}[i] \equiv_s \mathbf{tr}'[i]$.

Thus, $\mathbf{tr}'[i], 0, \mathbf{tr}'[i], 0\models \delta$. This implies that $B, \rho, \mathbf{tr}', i\models \delta$

The converse is similar.

\item $\varphi_1 \lor \varphi_2$.

By induction hypothesis, 

$A, \rho, \mathbf{tr}, i\models \varphi_1 \Leftrightarrow B, \rho, \mathbf{tr}', i\models \varphi_1$,
and $A, \rho, \mathbf{tr}, i\models \varphi_2 \Leftrightarrow B, \rho, \mathbf{tr}', i\models \varphi_2$

Thus, $A, \rho, \mathbf{tr}, i\models \varphi_1 \lor \varphi_2 \Leftrightarrow B, \rho, \mathbf{tr}', i\models \varphi_1 \lor \varphi_2$.

\item $\lnot \varphi$.

By induction hypothesis, $A, \rho, \mathbf{tr}, i\not\models \varphi \Leftrightarrow B, \rho, \mathbf{tr}', i\not\models \varphi$.

Thus, $A, \rho, \mathbf{tr}, i\models \lnot\varphi \Leftrightarrow B, \rho, \mathbf{tr}', i\models \lnot\varphi$.

\item $\langle \mu \rangle_- \varphi$.

We suppose that $A, \rho, \mathbf{tr}, i \models \langle \mu \rangle_- \varphi$.

By definition, $\mathbf{tr}[i-1] \overset{\mu}{\Longrightarrow} \mathbf{tr}[i]\ \mathrm{in}\ \mathbf{tr}\ \land  A, \rho, \mathbf{tr}, i-1 \models \varphi$.

$\mathbf{tr}'[i-1] \overset{\mu}{\Longrightarrow} \mathbf{tr}'[i]\ \mathrm{in}\ \mathbf{tr}'$ because $\mathbf{tr} \sim_t \mathbf{tr}'$.

By induction hypothesis, $B, \rho, \mathbf{tr}, i-1 \models \varphi$.

Thus, $B, \rho, \mathbf{tr}', i \models \langle \mu \rangle_- \varphi$.

The converse is similar.

\item $F\varphi$.

We suppose that $A, \rho, \mathbf{tr}, i \models F\varphi$.

By definition, $\exists j \geq i\ \mathrm{s.t.}\ A, \rho, \mathbf{tr}, j \models \varphi$.

By induction hypothesis, $B, \rho, \mathbf{tr}', j \models \varphi$.

Thus, $B, \rho, \mathbf{tr}', i \models F\varphi$.

The converse is similar.

\item $K\varphi$.

We suppose that $A, \rho, \mathbf{tr}, i \models K\varphi$.

By definition, $\forall \rho'\ \forall \mathbf{tr}''\in \mathrm{tr}(A\rho'); \mathbf{tr}[0, i] \sim_t \mathbf{tr}''[0, i] \Rightarrow A, \rho', \mathbf{tr}'', i \models \varphi$.

We arbitrarily take $\rho'$ and $\mathbf{tr}''' \in \mathrm{tr}(B\rho')$ such as $\mathbf{tr}'[0, i] \sim_t \mathbf{tr}'''[0, i]$

By assumption, there exists a trace $\mathbf{tr}''$ of $A\rho'$ such as $\mathbf{tr}''' \sim_t \mathbf{tr}''$.

Now, $\mathbf{tr}[0, i] \sim_t \mathbf{tr}'[0, i] \sim_t \mathbf{tr}'''[0, i] \sim_t \mathbf{tr}''[0, i]$, so $A, \rho', \mathbf{tr}'', i \models \varphi$.

By induction hypothesis, $B, \rho', \mathbf{tr}''', i\models \varphi$.

By arbitrariness of $\rho'$ and $\mathbf{tr}'''$, it follows that $B, \rho, \mathbf{tr}', i\models K\varphi$.

The converse is similar.
\end{enumerate}
We complete to prove (\ref{howto}), and it immediately follows that $A \approx_t B \Rightarrow A \sqsubseteq_L B$.

\ref{lg-tr}) We take arbitrarily an assignment $\rho$ from $\widetilde{x}$ to terms. We also take arbitrarily a trace \textbf{tr} of $A\rho$.

By assumption, $\exists \mathbf{tr}' \in \mathrm{tr}(B\rho) \mathrm{s.t.} \forall i \forall \varphi; [A, \rho, \mathbf{tr}, i \models \varphi \Leftrightarrow B, \rho, \mathbf{tr}', i \models \varphi]$.

We prove $\mathbf{tr} \sim_t \mathbf{tr}'$.

Let $\mathbf{tr} = A_0 \coloneqq A\rho \overset{\mu_1}{\Longrightarrow} A_1 \overset{\mu_2}{\Longrightarrow},..., \overset{\mu_n}{\Longrightarrow} A_n$. It holds that $A, \rho, \mathbf{tr}, 0\models F\langle \mu_n \rangle ... \langle \mu_1 \rangle \top$, so it also holds that $B, \rho, \mathbf{tr}', 0\models F\langle \mu_n \rangle ... \langle \mu_1 \rangle \top$.

Moreover, $\forall \mu; A, \rho, \mathbf{tr}, 0\not\models F\langle \mu \rangle \langle \mu_n \rangle ... \langle \mu_1 \rangle \top$ and $\forall \mu; A, \rho, \mathbf{tr}, 0\not\models F\langle \mu_n \rangle ... \langle \mu_1 \rangle \langle \mu \rangle \top$. Thus, $\forall \mu; B, \rho, \mathbf{tr}', 0\not\models F\langle \mu \rangle \langle \mu_n \rangle ... \langle \mu_1 \rangle \top$ and $\forall \mu; B, \rho, \mathbf{tr}', 0\not\models F\langle \mu_n \rangle ... \langle \mu_1 \rangle \langle \mu \rangle \top$.

Therefore, $\mathbf{tr}'$ is of the form $B_0 \coloneqq B\rho \overset{\mu_1}{\Longrightarrow} B_1 \overset{\mu_2}{\Longrightarrow},..., \overset{\mu_n}{\Longrightarrow} B_n$.

We arbitrarily take $i$. We suppose that $(M=N)\mathrm{fr}(A_i)$.

Now, $A, \rho, \mathbf{tr}, i\models M=N$. Hence, $B, \rho, \mathbf{tr}', i\models M=N$.

That is, $(M=N)\mathrm{fr}(B_i)$. Thus $(M=N)\mathrm{fr}(A_i) \Rightarrow (M=N)\mathrm{fr}(B_i)$ and the converse similarly holds.

It follows that $A_i \approx_s B_i$. Therefore, $\mathbf{tr} \sim_t \mathbf{tr}'$.

By arbitrariness of tr, it immediately follows that $A \sqsubseteq_L B \Rightarrow A \subseteq_t B$.
\end{proof}
We can immediately conclude that the next theorem holds.
\begin{theorem}
$A \approx_t B \Leftrightarrow A \equiv_L B$.
\end{theorem}
\section{Applications}
In this section, we abbreviate $\lnot F \lnot \varphi$ to $G\varphi$ ,$\lnot K \lnot \varphi$ to $P\varphi$, $\lnot (\lnot \varphi_1 \lor \lnot \varphi_2)$ to $\varphi_1 \land \varphi_2$, $\lnot \varphi_1 \lor \varphi_2$ to $\varphi_1 \rightarrow \varphi_2$ and $\lnot(M_1=M_2)$ to $M_1 \neq M_2$.

At first, we define minimal secrecy. This can be regarded as a generalization of minimal anonymity  \cite[Definition 2.3]{mano2010role}.
\begin{definition}
$x$ is minimally secret with respect to $\delta$ in $A$ iff $A \models G(\delta(x) \rightarrow P(\lnot \delta(x)))$.
\end{definition}
As a matter of fact, this is a very weak property. For instance, although $A$ and $B$ satisfy it, $A|B$ does not always satisfy. In addition, although $x$ is minimally secret with respect to nontrivial $\delta$, $x$ is not always minimally secret with respect to $\lnot \delta$.
\begin{example}
Let $\delta(z): z \neq a \land z \neq b$. Let
\begin{align*}
P &= \mathrm{if} ~x=a~ \mathrm{then} ~\overline{c}\langle s \rangle~ \mathrm{else} ~\overline{d}\langle s \rangle, \\
Q &= \mathrm{if} ~x=b~ \mathrm{then} ~\overline{c}\langle s \rangle~ \mathrm{else} ~\overline{d}\langle s \rangle.
\end{align*}
Then
\begin{align*}
P \models G(\delta(x) \rightarrow P(\lnot \delta(x))), \\
Q \models G(\delta(x) \rightarrow P(\lnot \delta(x))), \\
P \not\models G(\lnot\delta(x) \rightarrow P( \delta(x))), \\
P|Q \not\models  G(\delta(x) \rightarrow P(\lnot \delta(x))).
\end{align*}
In fact, when $\rho = [x \mapsto a]$, we can take the trace below:
\begin{equation*}
P\rho \overset{\nu y. \overline{c}\langle y \rangle}{\longrightarrow} \{s/y\}.
\end{equation*}
However, when $x$ is assigned to a term other than $a$ and $b$, there exist no traces whose actions correspond to the above.

When $\rho = [x \mapsto c]$, we can take the trace below:
\begin{equation*}
P\rho|Q\rho \overset{\nu y. \overline{d}\langle y \rangle}{\longrightarrow} Q\rho|\{s/y\} \overset{\nu z. \overline{d}\langle z \rangle}{\longrightarrow} \{s/y, s/z\}.
\end{equation*}
However, when $x$ is assigned to $a$ or $b$, there exist no traces whose actions correspond to the above.
\end{example}
Besides, although $x$ is minimally secret with respect to $\delta$ in $A$, $x$ is not always secret in $A^2$.
\begin{example}
Let $\delta(z): z=a$. Let
\begin{align*}
P &= \mathrm{if} ~x=a~ \mathrm{then} ~(\overline{a}\langle s \rangle + \overline{b}\langle s \rangle)~ \mathrm{else} ~\overline{a}\langle s \rangle, \\
Q &= \mathrm{if} ~x=b~ \mathrm{then} ~\overline{b}\langle s \rangle~ \mathrm{else} ~\overline{c}\langle s \rangle.
\end{align*}
Then
\begin{align*}
P+Q \models G(\delta(x) \rightarrow P(\lnot \delta(x))), \\
(P+Q)^2 \not\models G(\delta(x) \rightarrow P(\lnot \delta(x))).
\end{align*}
In fact, when $\rho = [x \mapsto a]$, we can take the trace below:
\begin{equation*}
(P+Q)^2\rho \overset{\nu y. \overline{b}\langle y \rangle}{\longrightarrow} (P+Q)\rho | \{s/y\} \overset{\nu z. \overline{c}\langle z \rangle}{\longrightarrow} \{s/y, s/z\}.
\end{equation*}
However, when $x$ is assigned to a term other than $a$, there exist no traces whose actions correspond to the above.
\end{example}
What is more, minimal secrecy is not preserved by restriction. Even though $\delta$ does not contain the restricted name, preservation does not hold in general.
\begin{example}
Let $\delta(z): z=a$. Let
\begin{align*}
P &= \mathrm{if} ~x=a~ \mathrm{then} ~\overline{b}\langle s \rangle, \\
Q &= \mathrm{if} ~x=n~ \mathrm{then} ~\overline{b}\langle s \rangle.
\end{align*}
Then
\begin{align*}
P + Q \models& G(\delta(x) \rightarrow P(\lnot \delta(x))), \\
\nu n. (P + Q) \not\models& G(\delta(x) \rightarrow P(\lnot \delta(x))).
\end{align*}
In fact, when $\rho = [x \mapsto a]$, we can take the trace below:
\begin{equation*}
\nu n. (P+Q)\rho \longrightarrow (\nu n. P)\rho \overset{\nu y. \overline{b}\langle y \rangle}{\longrightarrow} \{s/y\}.
\end{equation*}
However, when $x$ is assigned to a term other than $a$, there exist no traces whose actions correspond to the above. Note that $(\nu n. Q)\sigma$ can do nothing for any assignment $\sigma$ because of capture-avoiding.
\end{example}
Moreover, although $x$ is minimally secret with respect to $\delta_1$ and $\delta_2$ in $A$, $x$ is not always minimally secret with respect to $\delta_1 \lor \delta_2$. $\delta_2 = \lnot \delta_1$ is a counterexample.

On the other hand, $\land$ preserves minimal secrecy.
\begin{prop}
If $x$ is minimally secret with respect to $\delta_1$ and $\delta_2$ in $A$, then $x$ is minimally secret with respect to $\delta_1 \land \delta_2$ in $A$.
\end{prop}
\begin{proof}
We arbitrarily take $\rho$, $\mathbf{tr} \in \mathrm{tr}(A\rho)$ and $i$. We suppose that $A, \rho, \mathbf{tr}, i\models \delta_1 \land \delta_2$. 

By definition, $A, \rho, \mathbf{tr}, i\models \delta_1$.

By assumption, $A, \rho, \mathbf{tr}, i\models P(\lnot \delta_1)$.

It immediately follows that $A, \rho, \mathbf{tr}, i\models P(\lnot \delta_1 \lor \lnot \delta_2)$.

By arbitrariness of $\rho, \mathbf{tr}$ and $i$, we conclude that $A \models G(\delta_1 \land \delta_2 \rightarrow P(\lnot \delta_1 \lor \lnot \delta_2))$.
\end{proof}
We define total secrecy. This can also be regarded as a generalization of total anonimity \cite[Definition 2.4]{mano2010role}.
\begin{definition}\label{total-sec}
$x$ is totally secret in $A(x, \widetilde{y})$ iff 
\begin{equation*}
\forall \delta(z, \widetilde{z}, \widetilde{w}); A(x, \widetilde{y})\models G(\delta (x, \widetilde{y}, \widetilde{w}) \rightarrow P(\lnot \delta (x, \widetilde{y}, \widetilde{w})))
\end{equation*}
where $\delta$ contains no variables other than ones in $\{z\} \cup \widetilde{z} \cup \widetilde{w}$ and satisfies that $\forall \widetilde{N} \forall \psi \exists M:ground\ \mathrm{s.t.}\ \psi \models \lnot \delta (M, \widetilde{N}, \widetilde{w})$. Besieds,  $|\widetilde{y}| = |\widetilde{z}|$ and $\widetilde{w} \cap (\{x\} \cup \widetilde{y}) = \emptyset$.
\end{definition}
\begin{prop}\label{to-sec}
$x$ is totally secret in $A(x, \widetilde{y}) \Leftrightarrow A(x, \widetilde{y}) \approx_t A(x', \widetilde{y})$.
\end{prop}
\begin{proof}
$\Rightarrow$) We suppose for the sake of contradiction that $A(x, \widetilde{y}) \not\approx_t A(x', \widetilde{y})$

Then, there exist $M_1, M_2$ and $\widetilde{N}$ which are closed and whose every name is not in bn($A$) such that $A(M_1, \widetilde{N}) \not\approx_t A(M_2, \widetilde{N})$.

We suppose that $A(M_1, \widetilde{N}) \not\subseteq_t A(M_2, \widetilde{N})$ without loss of generality. That is, there exists $\mathbf{tr} \in \mathrm{tr}(A(M_1, \widetilde{N}))$ such that any trace of $A(M_2, \widetilde{N})$ is not statically equivalent to \textbf{tr}.

Let $\delta(z, \widetilde{z})$ be $z \neq M_2 \lor \widetilde{z} \neq \widetilde{N}$.
\begin{equation*}
A(x, \widetilde{y}), (x \mapsto M_1, \widetilde{y} \mapsto \widetilde{N}), \mathbf{tr}, |\mathbf{tr}|\models \delta(x, \widetilde{y}) 
\end{equation*}
The below follows because of how to take $\mathbf{tr}$.
\begin{equation*}
A(x, \widetilde{y}), (x \mapsto M_1, \widetilde{y} \mapsto \widetilde{N}), \mathbf{tr}, |\mathbf{tr}|\models K\delta(x, \widetilde{y})
\end{equation*}
This contradicts total secrecy.

$\Leftarrow$) We arbitrarily take $\delta, \rho, \mathbf{tr}$ and i, where $\delta$ meets the demand of Definition \ref{total-sec}.

We suppose that $A(x, \widetilde{y}), \rho, \mathbf{tr}, i\models \delta(x, \widetilde{y}, \widetilde{w})$.

We take $M$ such that $\mathrm{fr}(\mathbf{tr}[i])\models \lnot \delta(M, \rho(\widetilde{y}), \widetilde{w})$.
Let $\rho'$ be
\[
\rho'(y) = \begin{cases}
M & (y = x') \\
\rho(y) & (otherwise).
\end{cases}
\]
By assumption, $A(\rho(x), \rho(\widetilde{y})) \approx_t A(M, \rho(\widetilde{y}))$.

Hence, there exists $\mathbf{tr}' \in \mathrm{tr}(A(M, \rho(\widetilde{y})))$ such that $\mathbf{tr} \sim_t \mathbf{tr}'$.

By Lemma \ref{sta-log}, $\mathbf{tr}[i] \equiv_s \mathbf{tr}'[i]$ and $\mathbf{tr}'[i], 0, \mathbf{tr}'[i], 0\models \lnot \delta(M, \rho(\widetilde{y}), \widetilde{w})$.

Thus, $A(x, \widetilde{y}), \rho', \mathbf{tr}', i\models \lnot \delta(M, \rho(\widetilde{y}), \widetilde{w})$.

Therefore, $A(x, \widetilde{y}), \rho, \mathbf{tr}, i\models P(\lnot \delta(x, \rho(\widetilde{y}), \widetilde{w}))$.

That is, $A(x, \widetilde{y}) \models G(\delta(x, \widetilde{y}, \widetilde{w}) \rightarrow P(\lnot \delta(x, \rho(\widetilde{y}), \widetilde{w})))$.
\end{proof}
\begin{theorem}
If $x$ is totally secret in $A(x, \widetilde{y})$, then $x$ is also totally secret in $E[A(x, \widetilde{y})]$ for all contexts $E[\_]$.
\end{theorem}
\begin{proof}
It immediately follows from Proposition \ref{to-sec} and Theorem \ref{main-cong}.
\end{proof}
Our framework can also introduce role interchangeability \cite[Subsection 2.2.3]{tsukada2016compositional} to the applied pi calculus.
\begin{definition}
Let $\mathrm{fv}(A)\setminus \mathrm{dom}(A) \subseteq \{x_1,...,x_p\}$, $J$ be a finite set and $I = \{1,...,p\}$.

$(x_i, \delta_k)$ is role interchangeable with respect to $\{\delta_j(z_j, \widetilde{y_j})\}_{j \in J}$ in $A$ iff
\begin{equation*}
A(x_1,...,x_p) \models G(\delta_k(x_i, \widetilde{y_k}) \rightarrow \bigwedge_{l \in I} \bigwedge_{j \in J}(\delta_j(x_l, \widetilde{y_j}) \rightarrow P(\delta_k(x_l, \widetilde{y_k}) \land \delta_j(x_i, \widetilde{y_j}))))
\end{equation*}
where $\widetilde{y_j} \cap \{x_1,...,x_p\} = \emptyset$.
\end{definition}
\begin{prop}\label{nscon}\leavevmode \par
$\forall \widetilde{M} \forall i \forall \mathrm{\mathbf{tr}\in tr}(A(M_1,...,M_p))$

$\exists \widetilde{N} \exists \mathrm{\mathbf{tr}' \in tr}(A(M_i, N_2,...,N_{i-1}, M_1, N_{i+1},..., N_p)) \mathrm{s.t.} \mathbf{tr} \sim_t \mathbf{tr}'$

$\Leftrightarrow (x_1, \delta_k)$ is role interchangeable with respect to $\{\delta_j\}_{j \in J}$ in $A$ for all $\{\delta_j\}_{j \in J}$ and $k$.
\end{prop}
\begin{proof}
$\Rightarrow$) We arbitrarily take $\widetilde{M}$ and $i$.

Let $\rho = [x_1 \mapsto M_1 ,..., x_p \mapsto M_p]$.

We arbitrarily take $\mathbf{tr} \in \mathrm{tr}(A\rho)$ and $t$.

We suppose that $A, \rho, \mathbf{tr}, t \models \delta_k(x_1, \widetilde{y_k}) \land \delta_j(x_i, \widetilde{y_j})$.

By assumption, $\exists \widetilde{N} \exists \mathrm{\mathbf{tr}' \in tr}(A(M_i, N_2,...,N_{i-1}, M_1, N_{i+1},..., N_p)) \mathrm{s.t.}\ \mathbf{tr} \sim_t \mathbf{tr}'$.

Let $\rho' = [x_1 \mapsto M_i, x_2 \mapsto N_2,...,x_{i-1} \mapsto N_{i-1}, x_i \mapsto M_i, x_{i+1} \mapsto M_{i+1},...,x_P \mapsto N_p]$.

By Lemma \ref{sta-log}, $\mathbf{tr}[t] \equiv_s \mathbf{tr}'[t]$ and $\mathbf{tr}'[t], 0, \mathbf{tr}'[t], 0\models \delta_k(\rho(x_1), \widetilde{y_k}) \land \delta_j(\rho(x_i), \widetilde{y_j})$.

Thus, $A, \rho', \mathbf{tr}', t \models  \delta_k(x_i, \widetilde{y_k}) \land \delta_j(x_1, \widetilde{y_j})$.

Hence, $A, \rho, \mathbf{tr}, t \models P(\delta_k(x_i, \widetilde{y_k}) \land \delta_j(x_1, \widetilde{y_j}))$.

That is, $A \models G(\delta_k(x_1, \widetilde{y_k}) \rightarrow \bigwedge_{i \in I} \bigwedge_{j \in J}(\delta_j(x_i, \widetilde{y_j}) \rightarrow P(\delta_k(x_i, \widetilde{y_k}) \land \delta_j(x_1, \widetilde{y_j}))))$.

$\Leftarrow$) We arbitrarily take $\widetilde{M}$ and $i$.

Let $\rho = [x_1 \mapsto M_1 ,..., x_p \mapsto M_p]$.

We arbitrarily take $\mathbf{tr} \in \mathrm{tr}(A\rho)$.

Let $\delta_1(z) : z=M_1$ and $\delta_i(z) : z=M_i$.

By assumption, $(x_1, \delta_1)$ is role interchangeable with respect to $\{ \delta_1, \delta_i \}$.

Hence, $A \models G(\delta_1(x_1) \rightarrow (\delta_i(x_i) \rightarrow P(\delta_1(x_i) \land \delta_i(x_1))))$.

Therefore, $A, \rho, \mathbf{tr}, |\mathbf{tr}| \models \delta_1(x_1) \rightarrow (\delta_i(x_i) \rightarrow P(\delta_1(x_i) \land \delta_i(x_1)))$.

Because $A, \rho, \mathbf{tr}, |\mathbf{tr}| \models \delta_1(x_1) \land \delta_i(x_i)$, it holds that $A, \rho, \mathbf{tr}, t \models P(\delta_1(x_i) \land \delta_i(x_1))$.

That is, $\exists \rho' \exists \mathbf{tr}' \in \mathrm{tr}(A\rho') \mathrm{s.t.}\ \mathbf{tr} \sim_t \mathbf{tr}' \land A, \rho', \mathbf{tr}', |\mathbf{tr}'| \models \delta_1(x_i) \land \delta_i(x_1)$.

This means that $\rho'(x_1)=M_i \land \rho'(x_i)=M_1$.

Let $N_j = \rho'(x_j) (j \neq 1, i)$. 

Then, $\exists \mathrm{\mathbf{tr}' \in tr}(A(M_i, N_2,...,N_{i-1}, M_1, N_{i+1},..., N_p)) \mathrm{s.t.}\ \mathbf{tr} \sim_t \mathbf{tr}'$ holds.
\end{proof}
\begin{corollary}
$\forall l \in I \setminus \{i\}; A(x_1,...,x_i,...,x_l,...,x_p) \approx_t A(x_1,...,x_l,...,x_i,...,x_p)$

$\Rightarrow (x_i, \delta_k)$ is role interchangeable with respect to $\{\delta_j\}_{j \in J}$ in $A$ for all $\{\delta_j\}_{j \in J}$ and $k$.
\end{corollary}
%
%
%
%
%
%
When $p=2$, the converse holds. It immediately follows from Proposition \ref{nscon}.
\begin{prop}
$A(x_1, x_2) \approx_t A(x_2, x_1)$

$\Leftrightarrow (x_1, \delta_k)$ is role interchangeable with respect to $\{\delta_j\}_{j \in J}$ in $A$ for all $\{\delta_j\}_{j \in J}$ and $k$.
\end{prop}
%
%
%
%
%
%
%
%
We define openness. This can be regarded as a generalization of identity \cite[Subsection 2.2.5]{tsukada2016compositional}.
\begin{definition}
$x$ is open with respect to $\delta$ in $A$ iff
\begin{equation*}
\forall \rho \forall \mathbf{tr} \in \mathrm{tr_{max}}(A\rho); A, \rho, \mathbf{tr}, |\mathbf{tr}| \models \delta(x) \rightarrow K\delta(x)
\end{equation*}
\end{definition}
This is also very weak. As a matter of fact, openness is neither preserved by restriction nor parallel composition.
\begin{example}
Let $\delta(z): z=m$. Let
\begin{equation*}
P = \mathrm{if} ~x=m~ \mathrm{then} ~\overline{a}\langle n \rangle~ \mathrm{else} ~\overline{a}\langle f(n) \rangle.
\end{equation*}
Then $\forall \rho \forall \mathbf{tr} \in \mathrm{tr_{max}}(P\rho);$
\begin{equation*}
P, \rho, \mathbf{tr}, |\mathbf{tr}| \models \delta(x) \rightarrow K\delta(x).
\end{equation*}
However,
\begin{equation*}
\nu n.P, [x \mapsto m], \nu n.P\{m/x\} \overset{\nu y.\overline{a}\langle y \rangle}{\longrightarrow} \nu n.\{n/y\} , 1 \not\models \delta(x) \rightarrow K\delta(x).
\end{equation*}
\end{example}
\begin{example}
Let $\delta(z): z=m$. Let
\begin{align*}
P &= \mathrm{if} ~x=m~ \mathrm{then} ~\overline{a}\langle n \rangle~ \mathrm{else} ~\overline{b}\langle n \rangle, \\
Q &= \mathrm{if} ~x=m~ \mathrm{then} ~\overline{b}\langle n \rangle~ \mathrm{else} ~\overline{a}\langle n \rangle.
\end{align*}
We arbitrarily take $\rho, \mathbf{tr} \in \mathrm{tr}(P\rho)$ and $\mathbf{tr}' \in \mathrm{tr}(Q\rho)$. Then
\begin{align*}
P, \rho, \mathbf{tr}, |\mathbf{tr}| \models \delta(x) \rightarrow K\delta(x), \\
Q, \rho, \mathbf{tr}', |\mathbf{tr}'| \models \delta(x) \rightarrow K\delta(x), \\
P | Q, \sigma, \mathbf{tr}_0, 2 \not\models \delta(x) \rightarrow K\delta(x),
\end{align*}
\noindent
where $\sigma = [x \mapsto m]$ and $\mathbf{tr}_0$ is 
\begin{equation*}
\mathbf{tr}_0: P\sigma | Q\sigma \overset{\nu y.\overline{a}\langle y \rangle}{\longrightarrow} Q\sigma | \{n/y\} \overset{\nu z.\overline{b}\langle z \rangle}{\longrightarrow} \{n/y, n/z\}.
\end{equation*}
In fact, when $\sigma' \neq \sigma$, we can take the trace below:
\begin{equation*}
\mathbf{tr}_1: P\sigma' | Q\sigma' \overset{\nu y.\overline{a}\langle y \rangle}{\longrightarrow} P\sigma' | \{n/y\} \overset{\nu z.\overline{b}\langle z \rangle}{\longrightarrow} \{n/y, n/z\}.
\end{equation*}
Then, $\mathbf{tr}_0 \sim_t \mathbf{tr}_1$. That is, $x$ is not open with respect to $\delta$ in $P|Q$.
\end{example}
\begin{problem}\label{epi-pr}\leavevmode \par
\textbf{Input:} An extended process $A$, an assignment $\rho$, a trace $\mathbf{tr} \in \mathrm{tr}(A)$, a non-negative number $i \leq |\mathbf{tr}|$ and a formula $\varphi$.

\textbf{Question:} Does $A, \rho, \mathbf{tr}, i \models \varphi$ hold?
\end{problem}
This problem is trivially undecidable in general because word problem is undecidable in general.
\begin{prop}
Problem \ref{epi-pr} is undecidable in general, even when the word problem in $\Sigma$ is decidable.
\end{prop}
\begin{proof}
We again reduce the decision problem for static equivalence to Problem \ref{epi-pr}.

Let $\varphi$ and $\psi$ be frames. We assume that $\mathrm{dom}(\varphi) = \mathrm{dom}(\psi)$.

Let $\varphi = \nu \widetilde{n}. \{\fracinline{M_1}{x_1},...,\fracinline{M_l}{x_l}\}, \psi = \nu \widetilde{m}. \{\fracinline{N_1}{x_1},...,\fracinline{N_l}{x_l}\}$.

Let $P = \nu \widetilde{n}. \overline{a}\langle M_1 \rangle ... \overline{a}\langle M_l \rangle, Q = \nu \widetilde{m}. \overline{a}\langle N_1 \rangle ... \overline{a}\langle N_l \rangle$, where $a \notin \widetilde{n} \cup \widetilde{m}$.

Let $A = \mathrm{if} ~x=b~ \mathrm{then} ~P~ \mathrm{else} ~Q$ and $\rho : x \mapsto b$.

Let \textbf{tr} be $A\rho \overset{\nu x_1. \overline{a}\langle x_1 \rangle}{\longrightarrow} \nu \widetilde{n}.(\overline{a}\langle M_2 \rangle ... \overline{a}\langle M_l \rangle | \{\fracinline{M_1}{x_1}\}) \overset{\nu x_2. \overline{a}\langle x_2 \rangle}{\longrightarrow} ... \overset{\nu x_l. \overline{a}\langle x_l \rangle}{\longrightarrow} \nu \widetilde{n}. \{\fracinline{M_1}{x_1},...,\fracinline{M_l}{x_l}\}$.

We prove that $\varphi \approx_s \psi \Leftrightarrow A, \rho, \mathbf{tr}, i \models P(x \neq b)$.

\textbf{tr} is the only trace of $A$ which is statically equivalent to \textbf{tr}.

We arbitrarily take an assignment $\rho'$ which does not map $x$ to $b$.

A trace $\mathbf{tr}' \in \mathrm{tr}(A\rho')$ whose actions correspond to \textbf{tr} is the only below:
\begin{equation*}
A\rho' \overset{\nu x_1. \overline{a}\langle x_1 \rangle}{\Longrightarrow} \nu \widetilde{m}. (\overline{a}\langle N_2 \rangle ... \overline{a}\langle N_l \rangle | \{\fracinline{N_1}{x_1}\}) \overset{\nu x_2. \overline{a} \langle x_2 \rangle}{\longrightarrow} ... \overset{\nu x_l. \overline{a} \langle x_l \rangle}{\longrightarrow} \nu \widetilde{m}. \{\fracinline{N_1}{x_1},...,\fracinline{N_l}{x_l}\}.
\end{equation*}
Because $A, \rho', \mathbf{tr}', i \models x \neq b$, it holds that $A, \rho, \mathbf{tr}, i \models P(x \neq b) \Rightarrow \varphi \approx_s \psi$.

On the other hand, it holds that $\mathbf{tr} \sim_t \mathbf{tr}'$ if $\varphi \approx_s \psi$. This is similar to the proof of Proposition \ref{tr-ud}.

Hence, it follows that $\varphi \approx_s \psi \Rightarrow A, \rho, \mathbf{tr}, i \models P(x \neq b)$.
\end{proof}
\begin{lemma}\label{fin-ass}
Let $A$ be an extended process and $B$ a closed extended process.

Consider a modal formula $\varphi$, assignments $\rho,\rho'$ and an arbitrary $\mathbf{tr}\in \mathrm{tr}(B)$, together with a permutation $\pi$ which does not alter names in $A$, $\mathbf{tr}[0,i]$ and $\varphi$ for some $i\leq |\mathbf{tr}|$.

We assume that $\rho'$ is obtained by $\pi$ from $\rho$.

Let $X_i = \{\mathbf{tr}' \in \mathrm{tr}(A\rho) | \mathbf{tr}'[0,i] \sim_t \mathbf{tr}[0,i]\}$ and $Y_i = \{\mathbf{tr}' \in \mathrm{tr}(A\rho') | \mathbf{tr}'[0,i] \sim_t \mathbf{tr}[0,i]\}$.

Then there exists a bijection $f_i: X_i \rightarrow Y_i$ such that
\begin{equation}
A, \rho, \mathbf{tr}', i \models \varphi \Leftrightarrow A, \rho', f_i(\mathbf{tr}'), i \models \varphi.
\end{equation}
\end{lemma}
\begin{proof}
We define $f_i$ as $f_i(\mathbf{tr}') = \pi (\mathbf{tr}')$.

By assumption,  $\rho' = \pi (\rho)$.
\begin{equation*}
A, \rho, \mathbf{tr}', i \models \varphi \Leftrightarrow A, \pi(\rho), \pi(\mathbf{tr}'), i \models \varphi \Leftrightarrow A, \rho', f_i(\mathbf{tr}'), i \models \varphi.
\end{equation*}
\end{proof}
\begin{lemma}\label{quo-sp}
Let $T$ be a finite set of variables. Let $S$ be a finite set of names.

We define an equivalence relation $\asymp_S$ between assignments which is a map from $T$ to names.
\begin{multline}
\asymp_S = \{(\rho, \rho') | \mathrm{There\ exists\ a\ permutation}\ \pi \ \mathrm{which\ does\ not\ change\ names\ in}\ S \ \\ \mathrm{such\ that}\ \rho' \mathrm{is\ obtained\ by}\ \pi \ \mathrm{from}\ \rho.\}
\end{multline}
Then the quotient space by $\asymp_S$ is finite.
\end{lemma}
\begin{proof}
Let $T = \{x_1 ,..., x_l , y_1 ,..., y_{m_n}\}$ and
\begin{equation}
\rho = [x_1 \mapsto a_1 ,..., x_l \mapsto a_l , y_1 \mapsto b_1 ,..., y_{m_1} \mapsto b_1 , y_{m_1 +1} \mapsto b_2 ,..., y_{m_2} \mapsto b_2 ,..., y_{m_n} \mapsto b_n],
\end{equation}
where $a_1 ,..., a_l$ are names in $S$ and $b_1 ,..., b_n$ are names which is not in $S$. In addition, $i \neq j \Rightarrow b_i \neq b_j$, but we do not assume that $i \neq j \Rightarrow a_i \neq a_j$.

Then $\rho \asymp_S \rho' \Leftrightarrow [\rho(x_i) = \rho'(x_i)$ and $\rho'(y_i) = \pi(\rho(y_i))$ for some $\pi$].

Each equivalence class is determined by a division of variables such as described above. The number of such divisions is finite, so the quotient space is finite.
\end{proof}
We a bit change semantics of $K\varphi$.
\begin{equation*}
A, \rho, \mathbf{tr}, i \models K\varphi\ \mathrm{iff}\ \forall \rho'\ \forall \mathbf{tr}'\in \mathrm{tr}(A\rho'); \mathbf{tr}[0, i] \sim_t \mathbf{tr}'[0, i] \Rightarrow A, \rho', \mathbf{tr}', i \models \varphi,
\end{equation*}
where $\rho'$ is an assignment to names.
\begin{prop}
If static equivalence and word problem in $\Sigma$ are decidable and $\rho$ is restricted to an assignment to names, then Problem \ref{epi-pr} is decidable.
\end{prop}
\begin{proof}
We prove by induction on $\varphi$.
\begin{enumerate}
\item $\top$.

$A, \rho, \mathbf{tr}, i \models \top$ always holds, so this is decidable.

\item $M_1 = M_2$.

By assumption, this is decidable.

\item $M \in \mathrm{dom}$.

This is trivially decidable.

\item $\delta_1 \lor \delta_2$.

By induction hypothesis, $\delta_1$ and $\delta_2$ are decidable, so $\delta_1 \lor \delta_2$ is so.

Moreover, $\lnot \delta, \varphi_1 \lor \varphi_2$ and $\lnot \varphi$ are similar.
%
%
%
%
%

\item $\langle \mu \rangle_- \varphi$.

Whether $\mathbf{tr}[i-1] \overset{\mu}{\longrightarrow} \mathbf{tr}[i]$ holds in \textbf{tr} is clearly decidable.

By induction hypothesis, $A, \rho, \mathbf{tr}, i-1 \models \varphi$ is decidable, so $\langle \mu \rangle_- \varphi$ is so.

\item $F \varphi$.

By induction hypothesis, $A, \rho, \mathbf{tr}, i \models \varphi ,..., A, \rho, \mathbf{tr}, |\mathbf{tr}| \models \varphi$ are decidable, so $F \varphi$ is so.

\item $K \varphi$.

We consider a quotient space by $\asymp_{\mathrm{n}(\mathbf{tr}[0,i])\cup \mathrm{n}(A)\cup \mathrm{n}(\varphi)}$. By Lemma \ref{quo-sp}, this space is finite. Here, $\mathrm{n}(\mathbf{tr}[0,i])$ is a set of names which appear in $\mathbf{tr}[0,i]$.

Each equivalence class is determined by a division of variables, so this is computable.

We arbitrarily take a representative $\rho'$ of each equivalence class. 

By Proposition \ref{tr-de}, whether there exists $\mathbf{tr}' \in \mathrm{tr}(A\rho')$ such that $\mathbf{tr}[0,i] \sim_t \mathbf{tr}'[0,i]$ is decidable.

If such traces exist, the number of them is finite.

We arbitrarily take $\mathbf{tr}' \in \mathrm{tr}(A\rho')$ such as $\mathbf{tr}'[0,i] \sim_t \mathbf{tr}[0,i]$.

By induction hypothesis, $A, \rho', \mathbf{tr}', i \models \varphi$ is decidable.

We arbitrarily take $\rho'' \asymp_{\mathrm{n}(\mathbf{tr}[0,i])\cup \mathrm{n}(A)\cup \mathrm{n}(\varphi)} \rho'$.

By Lemma \ref{fin-ass},
\begin{equation}
A, \rho', \mathbf{tr}', i \models \varphi \Leftrightarrow A, \rho'', f_i(\mathbf{tr}'), i \models \varphi
\end{equation}
for some $\pi$ such that $\rho'' = \pi(\rho')$. Moreover, $\mathbf{tr}'[0,i] \sim_t f_i(\mathbf{tr}')[0,i]$. This is because $f_i(\mathbf{tr}'[0,i]) \sim_t \mathbf{tr}[0,i]$.

This is why we only have to check whether a representative $\rho'$ of each equivalence class satisfies that $A, \rho', \mathbf{tr}', i \models \varphi$.

This procedure can always be completed because of the finiteness of the quotient space.

In other words, $A, \rho, \mathbf{tr}, i \models K\varphi$ is decidable.
\end{enumerate}
\end{proof}
It is proved in  \cite[Theorem 1]{abadi2006deciding} that static equivalence on a convergent subterm theory is decidable, so the corollary below immediately follows.
\begin{corollary}
If the equational theory on $\Sigma$ is a convergent subterm theory and $\rho$ is restricted to an assignment to names, then Problem \ref{epi-pr} is decidable.
\end{corollary}
We again change semantics.
\begin{equation*}
A \models \varphi\ \mathrm{iff}\ \forall \rho\ \forall \mathbf{tr}\in \mathrm{tr}(A\rho); A, \rho, \mathbf{tr}, 0 \models \varphi,
\end{equation*}
where $\rho'$ is an assignment to names and inputted messages in \textbf{tr} are only variables.
\begin{problem}\label{epi-pr2} \leavevmode \par
\textbf{Input:} An extended process $A$ and a formula $\varphi$.

\textbf{Question:} Does $A \models \varphi$ hold?
\end{problem}
\begin{corollary}
If the equational theory on $\Sigma$ is a convergent subterm theory and $A$ contains no replications, Problem \ref{epi-pr2} is decidable. 
\end{corollary}
\chapter{Related Work}
\section{Process Algebras}
Process algebras are special labelled transition systems. Many kinds of them are proposed. Calculus of Communicating Systems (CCS) \cite{Milner:1989:CC:63446} is one of the origins. In CCS, a process can do interaction and silent action. This is a difference between an automaton and a process. However, communication topology cannot change because a process cannot exchange messages.

The pi calculus \cite{Milner:1992:CMP:162037.162038,Milner:1992:CMP:162037.162039,Milner:1999:CMS:329902} is an extension of CCS. The pi calculus processes can pass and create names, so this model handles mobility, that is, change of topology of process connection. The pi calculus is simple, but it is powerful enough to simulate the lambda calculus.

The spi calculus \cite{abadi1999calculus} is an extension of the pi calculus. It enables us to handle symmetric key encryption based on the Dolev-Yao model \cite{1056650}. In the spi calculus, two cyphertexts obtained by encrypting different plaintexts are indistinguishable unless an observer gets a secret key. Abadi and Gordon proved authenticity and secrecy of Wide Mouthed Frog protocol to use the spi calculus.

We focused on the applied pi calculus \cite{DBLP:journals/corr/AbadiBF16}. It is also an extension of the pi calculus. It can handle an arbitrary algebra, so in particular, it can handle public key encryption. Proverif \cite{blanchet2010proverif} is a tool to check bisimilarity and it is used for formal methods. In this calculus, a process can send not only names but also terms via an alias variable. We can handle not only merely secrecy but also stricter properties such as non-malleability because of these features.
%
%
%
%
\section{Logics}
Combining logic and a labelled transition system is well investigated.

Hennessy-Milner logic \cite{10.1007/3-540-10003-2_79} is an origin of logic about the behaviour of labelled transition systems. It is a modal logic characterizing observational congruence on labelled transition systems. Namely, systems satisfying the same modal formulas are observational equivalent when these labelled transition systems are image-finite.


Logical characterizations of strong static equivalence and labelled bisimulation in the applied pi calculus are provided in  \cite{pedersen2006logics}. This logic resembles Hennessy-Milner logic, but it also states relations between terms. Strong static equivalence demands that two frames enable a term to reduce same times and it is strictly stronger than ordinal static equivalence. In the report, only convergent subterm theories are considered. 

An epistemic logic for the applied pi calculus was already developed by Chadha, Delaune and Kremer \cite{chadha2009epistemic}. They defined formulas $\mathrm{\textbf{Has}}$ and $\widehat{\mathrm{\textbf{evt}}}$. $\mathrm{\textbf{Has}}$ directly represents an attacker's knowledge, and $\widehat{\mathrm{\textbf{evt}}}$ means that a particular event had occurred. They also suggested that trace equivalence is more suitable than labelled bisimilarity when we handle privacy.  On the other hand, a correspondent relation between logic and behaviour of processes was not provided. As a matter of fact, $\alpha$-equivalent processes do not always satisfy the same formulas.

Knight, Mardare and Panangaden \cite{knight2012combining} provided an epistemic logic for a labelled transition system. This framework is based on Hennessy-Milner logic, and it handles multiple agent's knowledge. Knowledge is based on a sequence of transitions, which is called a history.  They also proved weak completeness to construct Fischer-Ladner closure. However, compositionality was not discussed.
\section{Formal Approaches}
Formal methods enable us to prove that a security protocol satisfies desired properties. Many methods were developed.

Nowadays, multiple cryptographic protocols are often composed (e.g. electronic voting). Universally composability \cite{canetti2001universally} ensures that security properties are preserved by embedding in other protocols. In this framework, a program is expressed as a probabilistic polynomial time interactive Turing machine. In addition, security properties are represented as indistinguishability between an actual protocol and an ideal protocol. 

Dolev-Yao model assumes that it is impossible to break a cypher without a key, but of course, real encryption algorithms are not perfect. That is, soundness of Dolev-Yao model is not trivial. However, it was proved in  \cite{micciancio2004soundness} that Dolev-Yao model is sound if encryption is IND-CCA secure.

Hoare logic is used to prove the correctness of a program. Variants of Hoare logic are often used to prove safeness of a security protocol. 

The safeness of encryption is often proved using games. Namely, it is proved that the probability is almost $\frac{1}{2}$ that a probabilistic polynomial time attacker decodes a cryptogram of $1$ bit. In this regard, the game is transformed into a trivial game. Probabilistic Hoare logic ensures that transformation of games is reasonable. For instance, the security of ElGamal cryptosystem was proved by this logic in  \cite{corin2005probabilistic}, but the proof was not perfectly formal.

Protocol composition logic \cite{datta2007protocol} is also a variant of Hoare logic. It enables us to modularize a proof for a protocol because it is possible sequentially to compose Hoare triples. In the paper, ISO-9798-3 protocol was divided into two parts, and it was proved that the protocol is safe composing proofs that each part is safe.

Privacy and anonymity are also well studied and formulated.

Three privacy-type properties, vote-privacy, receipt-freeness and coercion-resistance, of electronic voting were formulated in  \cite{delaune2009verifying} using the applied pi calculus. It is noteworthy that indistinguishability is represented by labelled bisimilarity.  Three voting protocols are considered in the paper. For instance, it was proved that FOO92 \cite{fujioka1992practical} satisfies privacy and does not satisfy receipt-free.

(Strong) Unlinkability and (strong) anonymity were formulated by the applied pi calculus in  \cite{arapinis2010analysing}. These are defined for special forms of processes called well-formed $p$-party protocols. Unlinkability was formulated using traces, while strong unlinkability was formulated using labelled bisimilarity. Anonymity was similar. It is proved that unlinkability is not stronger than anonymity and vice versa. Namely, they are independent notions.

Tsukada, Sakurada, Mano and Manabe \cite{tsukada2016compositional} studied sequential and parallel compositionality of security notions to use an epistemic logic for a multiagent system. They proved that neither anonymity nor privacy is generally preserved by composition. They also provided a sufficient condition for preservation. In addition, role interchangeability implies privacy and anonymity under the suitable assumption. Privacy of FOO92 was also proved in the multi-agent system. However, this word ``parallel'' merely means that the same agent acts two actions. In other words, concurrency was not considered.
\chapter{Conclusion}
\section{Summary}
In this thesis, we proved that trace equivalence is a congruence and  provided an epistemic logic for the applied pi calculus to handle secrecy. We defined concurrent normal traces to use partial normal forms for analyzing transitions of parallel composed processes. In addition, we formulated secrecy, role-interchangeability, and openness to generalize privacy, role-interchangeability, and onymity regarding multiagent systems. Moreover, we associated trace equivalence with total secrecy.
Minimal secrecy is not preserved by an application of context, but total secrecy is preserved because of congruency of trace equivalence. We also give a sufficient condition for role-interchangeability.

The definition of concurrent normal traces is very complex. This is caused by difficulty in handling bound names.

We conclude that trace equivalence is a suitable notion to express indistinguishability in the view of security in the presence of a non-adaptive active adversary.

\section{Future Work}
In this thesis, we focused on trace equivalence. Many interesting problems remain.

It is proved in  \cite[Proposition 5]{abadi2006deciding} that even static equivalence is not decidable. On the other hand,  \cite[Theorem 1]{abadi2006deciding} proves that static equivalence is decidable in polynomial time for convergent subterm theories. We intend to study conditions to make trace equivalence decidable.

Secondly, our epistemic logic states only an adversary's knowledge. We intend to construct a logic for a process's knowledge. It will bridge a gap between multiagent systems and process calculi.

Thirdly, formalizations of other security properties such as non-malleability and unlinkability are also next topics.

Finally, what logic is suitable for security in the presence of an adaptive attacker is still open. 

\newpage


\bibliographystyle{plain}
\bibliography{reference}

\renewcommand{\thesection}{\Alph{section}}
\appendix
\chapter{Lemmas for Chapter 3}
We prove lemmas used in chapter \ref{cong}.

\begin{lemma}[Lemma \ref{drop-sigma}]
$P \overset{\alpha \sigma}{\longrightarrow} A \Rightarrow \sigma | P \overset{\alpha}{\longrightarrow} \sigma | A$.
\end{lemma}
\begin{proof}
We prove by case analysis.
\begin{enumerate}
\item $\alpha = N(M)$

By  \cite[Lemma B.10]{DBLP:journals/corr/AbadiBF16}, $P \equiv \nu \widetilde{n}.(N\sigma (x).P'|P_2)$ and $A \equiv \nu \widetilde{n}.(P'\{M\sigma /x \} | P_2)$ for some $\widetilde{n}, P', P_2$.

Hence, $\sigma | P \equiv \sigma | \nu \widetilde{n}.(N (x).P'|P_2)$ and $\sigma | A \equiv \sigma | \nu \widetilde{n}.(P'\{M /x \} | P_2)$.

Therefore, $\sigma | P \overset{\alpha}{\longrightarrow} \sigma | A$.

\item $\alpha = \nu x.\overline{N} \langle x \rangle$

By  \cite[Lemma B.10]{DBLP:journals/corr/AbadiBF16}, $P \equiv \nu \widetilde{n}.(\overline{N\sigma} \langle M\sigma \rangle .P'|P_2)$ and $A \equiv \nu \widetilde{n}.(P' | \{M\sigma /x \} | P_2)$ for some $\widetilde{n}, P', P_2$.

Hence, $\sigma | P \equiv \sigma | \nu \widetilde{n}.(\overline{N} \langle M \rangle.P'|P_2)$ and $\sigma | A \equiv \sigma | \nu \widetilde{n}.(P' | \{M /x \} | P_2)$.

Therefore, $\sigma | P \overset{\alpha}{\longrightarrow} \sigma | A$.
\end{enumerate}
\end{proof}

\begin{lemma}[Lemma \ref{drop-nu}]
$\nu u.A \overset{\mu}{\longrightarrow} B \land A:\mathrm{closed} \land \mathrm{fv}(\mu) \subseteq \mathrm{dom}(\nu u.A) \land \mathrm{n}(\mu) \cap \mathrm{bn}(\nu u.A) = \emptyset$

$\Rightarrow \exists B' \ s.t.\ A \overset{\mu}{\longrightarrow} B' \land B \equiv \nu u.B'$
\end{lemma}
\begin{proof}
Let $\mathrm{pnf}(A) = \nu \widetilde{n}.(\sigma | P)$. We prove by case analysis.
\begin{enumerate}
\item $u$ is a name $n$, and $\mu$ is silent.

$\mathrm{pnf}(\nu n.A) = \nu n \widetilde{n}.(\sigma | P)$.

By  \cite[Lemma B.23]{DBLP:journals/corr/AbadiBF16}, $P \longrightarrow P'$ and $B \equiv \nu n \widetilde{n}.(\sigma | P')$ for some closed $P'$. 

Therefore, $ \nu \widetilde{n}.(\sigma | P) \longrightarrow  \nu \widetilde{n}.(\sigma | P')$ because internal reductions are closed by an application of an evaluation context.

$B' =  \nu \widetilde{n}.(\sigma | P')$ satisfies this lemma.

\item $u$ is a name $n$, and $\mu$ is a labelled action $\alpha$.

$\mathrm{pnf}(\nu n.A) = \nu n \widetilde{n}.(\sigma | P)$.

By  \cite[Lemma B.19]{DBLP:journals/corr/AbadiBF16}, $P \overset{\alpha \sigma}{\longrightarrow} C$ and $B \equiv \nu n \widetilde{n}.(\sigma | C)$ for some $C$.

By Lemma \ref{drop-sigma}, $\sigma | P \overset{\alpha}{\longrightarrow} \sigma | C$.

Therefore, $ \nu \widetilde{n}.(\sigma | P) \overset{\alpha}{\longrightarrow}  \nu \widetilde{n}.(\sigma | C)$.

$B' =  \nu \widetilde{n}.(\sigma | C)$ satisfies this lemma.

\item $u$ is a variable $x$, and $\mu$ is silent.

$\mathrm{pnf}(\nu x.A) = \nu \widetilde{n}.(\sigma' | P)$ where $\sigma' = \sigma_{|dom(\sigma)\setminus \{x\}}$. Note that $\sigma' \equiv \nu x.\sigma$.

By  \cite[Lemma B.23]{DBLP:journals/corr/AbadiBF16}, $P \longrightarrow P'$ and $B \equiv \nu \widetilde{n}.(\sigma' | P')$ for some closed $P'$.

Therefore, $ \nu \widetilde{n}.(\sigma | P) \longrightarrow  \nu \widetilde{n}.(\sigma | P')$ because internal reductions are closed by an application of an evaluation context.

$B' =  \nu \widetilde{n}.(\sigma | P')$ satisfies this lemma.

\item $u$ is a variable $x$, and $\mu$ is a labelled action $\alpha$.

$\mathrm{pnf}(\nu x.A) = \nu \widetilde{n}.(\sigma' | P)$ where $\sigma' = \sigma_{|dom(\sigma)\setminus \{x\}}$.

By  \cite[Lemma B.19]{DBLP:journals/corr/AbadiBF16}, $P \overset{\alpha \sigma'}{\longrightarrow} C$ and $B \equiv \nu \widetilde{n}.(\sigma' | C)$ for some $C$.

By Lemma \ref{drop-sigma}, $\sigma' | P \overset{\alpha}{\longrightarrow} \sigma' | C$.

Therefore, $ \nu \widetilde{n}.(\sigma | P) \overset{\alpha}{\longrightarrow}  \nu \widetilde{n}.(\sigma | C)$.

$B' =  \nu \widetilde{n}.(\sigma | C)$ satisfies this lemma.
\end{enumerate}

\end{proof}

\begin{lemma}[Lemma \ref{change-label}]
Let $\nu \widetilde{n}. (\sigma | P)$ be a closed normal process.

$\nu \widetilde{n}. (\sigma | P) \overset{\alpha}{\longrightarrow} A \land \alpha\sigma = \beta\sigma \land \mathrm{fv}(\alpha) \subseteq \mathrm{dom}(\sigma) \land \widetilde{n} \cap (\mathrm{n}(\alpha) \cup \mathrm{n}(\beta)) = \emptyset$

$ \Rightarrow \nu \widetilde{n}. (\sigma | P) \overset{\beta}{\longrightarrow} A$.
\end{lemma}
\begin{proof}
Let $\nu \widetilde{n}. (\sigma | P)$ be a closed normal process.

By  \cite[Lemma B.19]{DBLP:journals/corr/AbadiBF16}, $P \overset{\alpha\sigma}{\longrightarrow} B \land A \equiv \nu \widetilde{n}. (\sigma | B)$.

By the assumption, $P \overset{\beta\sigma}{\longrightarrow} B$.

By Lemma \ref{drop-sigma}, $\sigma | P \overset{\beta}{\longrightarrow} \sigma | B$, so $\nu \widetilde{n}. (\sigma | P) \overset{\beta}{\longrightarrow} \nu \widetilde{n}. (\sigma | B)$.

Therefore, $\nu \widetilde{n}. (\sigma | P) \overset{\beta}{\longrightarrow} A$.
\end{proof}
\newpage
\begin{lemma}[Lemma \ref{erase-sigma}]\leavevmode \par
$\sigma | A \overset{\mu}{\longrightarrow} \sigma | B \land \mathrm{dom}(\sigma) \cap \mathrm{fv}(\mu) = \emptyset \land x \in \mathrm{dom}(\sigma) \Rightarrow x\sigma$: closed 

$\Rightarrow A\sigma \overset{\mu}{\longrightarrow} B\sigma$.
\end{lemma}
\begin{proof}
Let $\widetilde{x} = \mathrm{dom}(\sigma)$.

$\nu \widetilde{x}. (\sigma | A) \overset{\mu}{\longrightarrow} \nu \widetilde{x}. (\sigma | B)$.

$A\sigma \overset{\mu}{\longrightarrow} B\sigma$.
%
%
\end{proof}
\begin{lemma}[Lemma \ref{shift-sigma}]
$\sigma | P \overset{\mu}{\longrightarrow} B \land \sigma | P$: closed normal $\land \mathrm{fv}(\mu) \subseteq \mathrm{dom}(\sigma)$

$\Rightarrow \exists B'\ \mathrm{s.t.}\ P\sigma \overset{\mu\sigma}{\longrightarrow} B' \land B \equiv \sigma | B'$.
\end{lemma}
\begin{proof}\leavevmode \par
\begin{enumerate}
\item $\mu$ is silent.

By  \cite[Lemma B.23]{DBLP:journals/corr/AbadiBF16}, $P \longrightarrow Q \land B \equiv \sigma | Q$.

By  \cite[Lemma B.3]{DBLP:journals/corr/AbadiBF16}, $P\sigma \longrightarrow Q\sigma$.

$\sigma | Q\sigma \equiv \sigma | Q \equiv B$.

\item $\mu$ is a labelled action $\alpha$.

By  \cite[Lemma B.19]{DBLP:journals/corr/AbadiBF16}, $P \overset{\alpha\sigma}{\longrightarrow} C \land B \equiv \sigma | C$ for some $C$.

By  \cite[Lemma B.11]{DBLP:journals/corr/AbadiBF16}, $P\sigma \overset{\alpha\sigma}{\longrightarrow} C\sigma$.

$\sigma | C\sigma \equiv \sigma | C \equiv B $.
\end{enumerate}
\end{proof}

\begin{lemma}[Lemma \ref{bra}]
$\sigma\rho \equiv [\sigma\rho]$.
\end{lemma}
\begin{proof}
Let $\widetilde{y} = \mathrm{dom}(\rho)$. Then, $\sigma\rho \equiv \nu \widetilde{y}.(\sigma | \rho) \equiv \nu \widetilde{y}.(\sigma \uplus \rho) \equiv [\sigma\rho]$.
\end{proof}
\begin{lemma}[Lemma \ref{sub-ex}]
$\rho\sigma \equiv \rho[\sigma\rho]$.
\end{lemma}
\begin{proof}
Let $\widetilde{x} = \mathrm{dom}(\sigma)$. Then,
\begin{align*}
\rho\sigma &\equiv \nu \widetilde{x}. (\rho | \sigma) \\
&\equiv \nu \widetilde{x}. (\rho | \sigma\rho) \\
&\equiv \nu \widetilde{x}. (\rho | [\sigma\rho])\ \ \mathrm{By\ Lemma\ \ref{bra}} \\
&\equiv \rho[\sigma\rho].
\end{align*}
\end{proof}
\begin{lemma}[Lemma \ref{cket}]
$(\sigma | P)\rho \equiv [\sigma\rho] | P[\rho\sigma]$.
\end{lemma}
\begin{proof}
Let $\widetilde{y} = \mathrm{dom}(\rho)$. Then,
\begin{align*}
(\sigma | P)\rho &= \sigma\rho | P\rho \\
&\equiv [\sigma\rho] | P\rho \\
&\equiv [\sigma\rho] | P\rho[\sigma\rho] \\
&\equiv [\sigma\rho] | \nu \widetilde{y}. (P | \rho[\sigma\rho]) \\
&\equiv [\sigma\rho] | \nu \widetilde{y}. (P | \rho\sigma) \ \ \mathrm{By\ Lemma\ \ref{sub-ex}} \\
&\equiv [\sigma\rho] | P\rho\sigma.
\end{align*}
\end{proof}

\end{document}